\documentclass[conference]{IEEEtran}
\IEEEoverridecommandlockouts
\usepackage{cite}
\usepackage{amsmath,amssymb,amsfonts}
\usepackage{graphicx}
\usepackage{textcomp}
\usepackage{xcolor}
\def\BibTeX{{\rm B\kern-.05em{\sc i\kern-.025em b}\kern-.08em
    T\kern-.1667em\lower.7ex\hbox{E}\kern-.125emX}}

\usepackage{booktabs}
\pdfoutput=1  %for arXive to compile using pdflatex

\newif\iffullversion\fullversiontrue
\iffullversion
\newcommand{\fullversion}[2]{#2}
\else
\newcommand{\fullversion}[2]{#1}
\fi

\usepackage{graphicx}
\usepackage[hyphens]{url}
\usepackage{amsmath,amsfonts,amsthm}
\usepackage{latexsym}
\usepackage{multirow}

\newtheorem{theorem}{Theorem}
\newtheorem{lemma}{Lemma}
\newtheorem{corollary}[theorem]{Corollary}
\usepackage{verbatim}

\usepackage{stmaryrd}
\usepackage{xspace}
\usepackage{algpseudocode}
\usepackage[ruled,lined,boxed,commentsnumbered]{algorithm2e}
\usepackage{pgfplots}  % For making plots of data
\usepackage{tikz}
\usepackage{algorithmicx}
\usepackage{graphicx}
\graphicspath{{images/}}
\usepackage{subcaption} %To create subfigures
\usepackage{microtype}  %Some how magically makes sure equations and citations dont run into the side margin. No idea how (or what else it does).
\usepackage{hyperref}
\hypersetup{breaklinks=true}
\pgfplotsset{compat=1.5}   % Needed for 'ylabel shift' to work

\newcommand{\ignore}[1]{}         %% We shouldn't use this if possible. E.g. it breaks syntax highlighting. Also its hard to tell what latex code matters.

\newcommand{\bpdf}{\mathtt{bpdf}} 

\newcommand{\namedref}[2]{\hyperref[#2]{#1~\ref*{#2}}}
\newcommand{\figlab}[1]{\label{fig:#1}}

\newenvironment{proofof}[1]{\noindent {\em Proof of #1.  }}{}%\QED}

\newenvironment{remindertheorem}[1]{\medskip \noindent {\bf Reminder of  #1.  }\em}{}

\newenvironment{reminderlemma}[1]{\medskip \noindent {\bf Reminder of  #1.  }\em}{}

\newcommand{\plotcyclelist}[0]{
    cycle list = { { gray, dashed}, {black}, {violet}, {teal},  {gray}, {red}, {blue}, {cyan}, {magenta}, {green}, {orange}, {blue, dashed}, {cyan, dashed}, {magenta, dashed}, {green, dashed}, {orange, dashed}, %{yellow} 
    }
}

\newcommand{\plotlegendpos}[0]{
    legend pos = outer north east
}
\newcommand{\plotlegendentries}[0]{
    legend entries = {FrequencyUB, LPUB,  LPLB, PriorLB, SamplingLB, ExtendedLB:Markov, ExtendedLB:JtR, ExtendedLB:Hashcat, ExtendedLB:PCFG, ExtendedLB:NN, Markov, JtR, Hashcat, PCFG, NN,
    Distinct(S)}
}
\newcommand{\plotlegends}[0]{
    \addlegendimage{no markers, black}
    \addlegendimage{no markers, violet}
      \addlegendimage{no markers, teal}
      \addlegendimage{no markers, gray}
      \addlegendimage{no markers, red}
      
       \addlegendimage{no markers, dashed, blue}
       \addlegendimage{no markers, dashed, cyan}
      \addlegendimage{no markers, dashed, magenta}
      \addlegendimage{no markers, dashed, green}
      \addlegendimage{no markers, dashed, orange}
      
      \addlegendimage{no markers, blue}
      \addlegendimage{no markers, cyan}
      \addlegendimage{no markers, magenta}
      
      \addlegendimage{no markers, green}
      \addlegendimage{no markers, orange}
      \addlegendimage{no markers, dashed, gray}
     
}
\begin{document}

\title{Towards a Rigorous Statistical Analysis of Empirical Password Datasets
\thanks{This is the full version of IEEE S\&P 2023 conference paper with missing proofs and complementary figures.}
}

\author{\IEEEauthorblockN{Jeremiah Blocki\IEEEauthorrefmark{1}}
\IEEEcompsocitemizethanks{\IEEEcompsocthanksitem\IEEEauthorrefmark{1}The authors are listed in alphabetical order.}
\IEEEauthorblockA{\textit{Purdue University}  \\
jblocki@purdue.edu}
\and
\IEEEauthorblockN{Peiyuan Liu\IEEEauthorrefmark{1}}
\IEEEauthorblockA{\textit{Purdue University} \\
liu2039@purdue.edu}
}

\maketitle

\begin{abstract}
A central challenge in password security is to characterize the attacker's guessing curve i.e., what is the probability that the attacker will crack a random user's password within the first $G$ guesses. A key challenge is that the guessing curve depends on the attacker's guessing strategy and the distribution of user passwords, both of which are unknown to us. In this work we aim to follow Kerckhoffs's principle and analyze the performance of an optimal attacker who knows the password distribution. Let $\lambda_G$ denote the probability that such an attacker can crack a random user's password within $G$ guesses. We develop several statistically rigorous techniques to upper and lower bound $\lambda_G$ given  $N$ independent samples from the unknown password distribution  $\mathcal{P}$. We show that our upper/lower bounds on $\lambda_G$ hold with high confidence and we apply our techniques to analyze eight large password datasets. Our empirical analysis shows that even state-of-the-art password cracking models are often significantly less guess efficient than an attacker who can optimize its attack based on its (partial) knowledge of the password distribution. We also apply our statistical tools to re-examine different models of the password distribution i.e., the empirical password distribution and Zipf's Law. We find that the empirical distribution closely matches our upper/lower bounds on $\lambda_G$ when  the guessing number $G$ is not too large i.e., $G \ll N$. However, for larger values of $G$ our empirical analysis rigorously demonstrates that the empirical distribution (resp. Zipf's Law) overestimates the attacker's success rate. We apply our statistical techniques to upper/lower bound the effectiveness of password throttling mechanisms (key-stretching) which are used to reduce the number of attacker guesses $G$. Finally, if we are willing to make an additional assumption about the way users respond to password restrictions, we can use our statistical techniques to evaluate the effectiveness of various password composition policies which restrict the passwords that users may select. 
\end{abstract}


\section{Introduction}
Understanding and characterizing the distribution over user chosen passwords is a key challenge in the field of cybersecurity with important implications towards the development of robust security policies.  How expensive does a password hashing algorithm need to be to deter an offline brute-force attacker? Can we strengthen the password distribution by imposing restrictions, e.g., requiring passwords to include numbers and/or capital letters, on the passwords that users can pick? If so are the security gains significant enough to justify the usability costs? Characterizing the attacker's guessing curve is a fundamental challenge which is central to addressing the above questions. In particular, we would like to know the probability that an attacker will crack a random user's password after $G$ attempts as $G$ ranges from small (online password spraying attacker) to large (offline attack). As a concrete motivating application suppose that an organization's current authentication policy locks a user account if there are more than $G=100$ incorrect login attempts within a 30 day period. The organization is considering adopting a stricter lockout policy which locks down the account after $G=50$ incorrect logins within 30 days, but will only adopt such a change if the policy change substantially reduces the risk posed by an online attacker. An immediate challenge is that the answer to these questions depends on the distribution $\mathcal{P}$ over user passwords as well as the attacker's guessing strategy both of which are unknown.  

One natural attempt to address the above questions is to fix a state-of-the art password cracking model $M$ and assume that the attacker uses this model to generate password guesses. Now our organization might estimate that adopting the stricter lockout policy (resp. doubling the cost of the password hash function) reduces the success rate of an online (resp. offline) password attacker by $\lambda_{G,M}-\lambda_{G/2,M}$ where $\lambda_{G,M}$ denotes the probability that an attacker will crack a randomly sampled password $pwd \leftarrow \mathcal{P}$ within $G$ guesses when using model $M$. Since the distribution $\mathcal{P}$ is unknown the organization cannot compute $\lambda_{G,M}$ directly. However, it is easy to approximate $\lambda_{G,M}$ empirically given a few samples $S=(s_1,\ldots, s_N)$ from the unknown password distribution $\mathcal{P}$.

One downside to the above approach is that our policy decision is dependent on the specific password cracking model $M$. Password cracking models can be developed using a variety of different techniques including: Neural Networks~\cite{USENIX:MUSKBCC16}, Probabilistic Context Free Grammars~\cite{SP:WAMG09,NDSS:VerColTho14}, Markov Models \cite{CCS:NarShm05b,NDSS:CasDurPer12,SP:MYLL14,durmuth2015omen} and sophisticated rule lists~\cite{SP:LNGCU19} used in password cracking software such as John the Ripper (JtR)~\cite{johntheripperlink} and Hashcat \cite{hashcatlink}. Basing policy decisions on a model specific guessing curve $\lambda_{G,M}$ may be unwise as this curve can vary greatly from model to model. In particular, we do not know which model $M$ the attacker will use and it is plausible that the attacker's model will be more sophisticated than any publicly available model. Furthermore, the attacker might continue to improve the cracking model $M$ over time. 

In an attempt to address this challenge prior work (e.g., \mbox{\cite{USENIX:USBCCKKMMS15,USENIX:MUSKBCC16}}) proposed using the minimum guessing number heuristic to approximate the performance of an attacker who may have a more advanced password cracking model with more training data and/or more sophisticated rule lists/dictionaries. Under this heuristic we use multiple different password cracking models (Neural Networks\cite{USENIX:MUSKBCC16}, Probabilistic Context Free Grammars \mbox{\cite{SP:WAMG09,NDSS:VerColTho14}}, Markov Models \mbox{\cite{CCS:NarShm05b,NDSS:CasDurPer12,SP:MYLL14,durmuth2015omen}}) to estimate the attacker's guessing curve. In particular, if a password $pwd$ is cracked within $G$ guesses under {\em any} of these models then, under the minimum guessing number heuristic, we assume that the real world attacker would also crack this password within $G$ guesses. Our analysis indicates that, even we apply the minimum guessing number heuristic, we can still significantly underestimate the performance of a real world attacker.

Kerckhoffs's principle states one should design systems under the assumption that the attacker will eventually gain full familiarity with them. Applying this principal within the context of password security our organization should assume that the attacker knows the password distribution and consider the guessing curve $\lambda_G$ of such an attacker when making policy decisions. Here, $\lambda_G$ denotes the cumulative probability mass of the top $G$ passwords in the (unknown) password distribution $\mathcal{P}$ i.e., the probability that an attacker who knows the distribution can guess a random password $pwd \leftarrow \mathcal{P}$ within $G$ guesses. While a real world attacker may not have perfect knowledge of the password distribution, we should expect that password cracking models will improve over time and  performance of the real world attacker may even approach the ideal guessing curve $\lambda_G$ in the limit. Thus, when making password policy decisions the conservative option is to follow Kerckhoffs's principle and consider the guessing curve of the ideal attacker who knows the password distribution instead of using heuristics to guesstimate the performance of an unknown attacker.
An additional advantage of following Kerckhoffs's principle is that the policy recommendations will be stable i.e., they would not need to be reevaluated every time a more sophisticated password cracking model is developed. However, an immediate challenge is that the distribution $\mathcal{P}$ is unknown to us making it impossible to directly compute $\lambda_G$.    

One common heuristic to estimate $\lambda_G$ is to use the empirical distribution derived from a breached password dataset e.g., as in \cite{SP:CAAJR16,CCS:CWPCR17,SP:BloHarZho18,NDSS:BloDatBon16,SP:Bonneau12,CSF:BloDat16,bai2021dahash}.
In particular, given a dataset $S $ of $N$ user passwords we can let $\hat{f}_i$ denote the frequency of the $i$th most common password $pwd$ in the sample $S$. Assuming that the samples $S$ were drawn iid from our password distribution $\mathcal{P}$ we can then estimate that an attacker making $G$ guesses per account will crack a random user's password with probability $\hat{\lambda}_G = \sum_{i=1}^G \hat{f}_i/N$. A downside to this approach is that the empirical estimate $\hat{\lambda}_G$ can substantially overestimate the true value $\lambda_G$. Indeed we will always have $\hat{\lambda}_G = 1$ whenever $G \geq N$ even when the password distribution $\mathcal{P}$ is very strong. For example, suppose that $\mathcal{P}$ is the uniform distribution over strong 56-bit passwords so that an attacker will crack a random user's password with probability at most $\lambda_G = 2^{-56}G$ within $G$ guesses. However, if we use the empirical distribution with a very large sample size $N=2^{33}$ (larger than the global population) we would still have $\hat{\lambda}_N = 1 \gg 2^{-23} = \lambda_N$. Thus, it may also be undesirable to base policy decisions on the empirical distribution $\hat{\lambda}_G$ --- especially in offline attack settings where $G$ can be large.

In this paper we address the following questions: Can we confidently derive accurate upper and lower bounds on $\lambda_G$? When (if ever) can we use the empirical distribution $\hat{\lambda}_G$ to accurately model the real distribution? How guess efficient are state-of-the-art password cracking models? 

Despite their shortcomings and many attempts to replace them passwords remain entrenched as the dominant form of authentication on the internet and are likely to play a critical role in the foreseeable future because they are easy to use, easy to deploy and users are already familiar with them  \cite{SP:BHVS12}. Thus, characterizing the distribution $\lambda_G$ of user chosen passwords will continue to be a important challenge in the field of cybersecurity.

\subsection{Our Contributions}
We develop a statistical framework to upper and lower bound the guessing curve $(\lambda_G)$ of an ideal attacker who knows the password distribution. We stress that we make no a priori assumptions about the shape of the password distribution. Instead, to apply our statistical framework we only require independent samples $S=(s_1,\ldots,s_n)$ from the unknown password distribution $\mathcal{P}$. All of our bounds can be shown to hold with high probability over the randomly selected password samples. In practice, the password samples could either be obtained from prior breached password datasets or (ideally) an organization could adapt prior work \cite{SP:Bonneau12,NDSS:BloDatBon16} to obtain samples from its own users in a secure and privacy preserving manner --- see discussion in Section \ref{subsec:securesamples}.  
We apply our techniques to analyze several empirical password datasets and illustrate how the upper/lower bounds can be used to guide policy decisions. 
\vspace{0.05cm}

\noindent{\bf Upper Bounds.} We develop two techniques to obtain high confidence upper bounds on the guessing curve of an ideal attacker using the empirical password distribution and linear programming. We first show that $\lambda_G$ is (with high probability) upper bounded by the empirical estimate $\hat{\lambda}_G$ i.e., with high probability we have $\lambda_G \leq \hat{\lambda}_G+\epsilon$ for a very small constant $\epsilon > 0$. Empirical analysis shows that this upper bound is often tight when $G$ ranges from small to moderately large. However, the upper bound becomes less and less tight as $G$ increases and devolves into the trivial upper bound $\lambda_G \leq 1$ once $G \geq \mathbf{Distinct}(S)$ exceeds the number of distinct passwords in our sample $S$. Thus, we develop a second approach to upper bound $\lambda_G$ when $G$ is large using linear programming (LP). Our LP approach adapts techniques of Valiant and Valiant~\cite{valiant2017estimating} for estimating properties of a distribution when the number of samples is smaller than the support of the distribution. Intuitively, our LP searches for a distribution which maximizes $\lambda_G$ subject to the constraint that the distribution is ``sufficiently consistent" with our sample $S$. With high probability, the real distribution $\mathcal{P}$ will also be consistent with all of our linear constraints so we will obtain a valid upper bound by maximizing over all consistent distributions. Empirical analysis shows that our LP upper bounds are superior for sufficiently large $G$ and that we often obtain non-trivial upper bounds even when $G \geq \mathbf{Distinct}(S)$.
\vspace{0.05cm}

\noindent{\bf Lower Bounds.} We develop three techniques to obtain high confidence lower bounds on the guessing curve $(\lambda_G)$ of an ideal attacker: sampling, model-sampling hybrid and linear programming. Additionally, we show how to adapt prior analysis of Blocki et al.\cite{SP:BloHarZho18} to obtain another high confidence lower bound. Our first sampling technique is a simple algorithm inspired by Good-Turing frequency estimation. The algorithm randomly partitions our sample $S$ into two components $D_1$ and $D_2$ of size $N-d$ and $d$ respectively, builds a dictionary $T(D_1,G)$ containing the $G$ most common passwords in $D_1$, and computes the fraction %$\tilde{\lambda}_{G,S_1,S_2}$ 
of passwords in $D_2$ which appear in this dictionary --- we remark that a real world attacker who obtains the samples $D_1$ can also build the dictionary $T(D_1,G)$ and then use this dictionary to crack other passwords. Empirical analysis demonstrates that the resulting lower bounds are nearly tight for smaller guessing numbers (e.g., $G \leq 10^6$), but the lower bound will plateau at $\approx 1-\frac{\mathbf{Unique}(S)}{N}$ once $G > \mathbf{Distinct}(S)$, where  $\mathbf{Unique}(S)$ counts the number of passwords which appear exactly once in our sample $S$. We provide two additional techniques to push past this barrier and derive stronger lower bounds when $G$ is large. First, we can adapt the linear program described earlier to instead search for a distribution that {\em minimizes}  $\lambda_G$ subject to the constraint that the distribution is ``sufficiently consistent" with our sample $S$.  Empirical analysis confirms that this approach generates tighter lower bounds when $G > \mathbf{Distinct(S)} $ allowing us to push past the $1-\frac{\mathbf{Unique(S)}}{N}$ barrier. Finally, we show how a password cracking model $M$ can be used to extend the first lower bound when  $G > \mathbf{Distinct}(D_1)$ using a hybrid approach: as before we partition our sample $S$ into two components $D_1$ and $D_2$ and then we output the fraction of passwords in $D_2$ which {\em either} appear in $D_1$ or that appear in the top $G-\mathbf{Distinct}(D_1)$ guesses generated by our model $M$. Empirical analysis shows that this combined bound can improve on our prior lower bounds when $G$ is very large. 
\vspace{0.05cm}

 \noindent{\bf Empirical Analysis.} We apply our theoretical upper and lower bounds to analyze eight password datasets (i.e. Yahoo!, RockYou, 000webhost, Neopets, Battlefield Heroes, Brazzers, Clixsense, CSDN), and we compare our bounds with state-of-the-art password cracking techniques such as Markov models~\cite{CCS:NarShm05b,NDSS:CasDurPer12,SP:MYLL14,durmuth2015omen}, Probabilistic Context-free Grammars (PCFG) ~ \cite{SP:WAMG09,NDSS:VerColTho14}, neural networks ~\cite{USENIX:MUSKBCC16}, John the Ripper (JtR) and Hashcat. We find that our new techniques for lower bounding $\lambda_G$ significantly improve upon prior work of Blocki et al.~\cite{SP:BloHarZho18} e.g., for RockYou dataset with guessing budget $G=1.3\times 10^8$  our lower bounds are $\lambda_G \geq 62.64\%$ (sampling) and $\lambda_G \geq 72.70\%$ (linear programming) in comparison to the weaker lower bound $\lambda_G \geq 53.95\%$ obtained from \cite{SP:BloHarZho18}. Our empirical analysis also shows that, for smaller values of $G$, our upper and lower bounds on $\lambda_G$ are very close and that the empirical estimate $\hat{\lambda}_G$ is sandwiched between these two values. This provides an answer to our second question i.e., the empirical guessing curve $\hat{\lambda}_G$ closely approximates the real guessing curve as long as the guessing number $G$ is not too large. By contrast, for some larger values of $G$ we can demonstrate with high confidence that the empirical guessing curve $\hat{\lambda}_G$ significantly overestimates $\lambda_G$ indicating that the empirical distribution should not be used to approximate the real guessing curve in these settings. We also compare our upper/lower bounds to CDF-Zipf curves fit to the empirical dataset \cite{wang2017zipf} and, we identify cases where the CDF-Zipf curve overestimates $\lambda_G$ by {\em at least} $12\%$.  
 
 We find that our lower bounds on $\lambda_G$ are often significantly higher than the guessing curves obtained using state-of-the-art password cracking models~\cite{USENIX:MUSKBCC16,CCS:DelFil15,SP:WAMG09,SP:LNGCU19} indicating that there is still room to develop improved password cracking models which are more guess efficient. For example, an attacker making $G \approx 8.4$ million guesses per account would crack {\em at most} $14\%$ of passwords in the 000webhost using any of the state-of-the-art password cracking models we analyzed. By contrast, our high confidence\footnote{The probability of an errant lower (or upper) bound can be upper bounded by a small constant $\delta$. In our experiments we tuned our parameters such that $\delta \leq 0.01$.} lower bounds show that an attacker who knows the password distribution would crack {\em at least} $39.16\%$ of 000webhost passwords. Similar observations held for other datasets. This provides compelling statistical evidence even the most sophisticated password cracking models still have a large room for improvement in guess efficiency. There has been a push towards designing moderately expensive password hashing algorithms to discourage offline attacks e.g., see \cite{PHC,biryukov2016argon2,CCS:AlwBloHAr17,percival2009stronger,SP:BloHarZho18}. If password guessing becomes more expensive then attacker will have additional incentive to develop guess efficient models. 

\noindent{\bf Implications for Password Policies.} We apply our statistical techniques to quantify the security benefits of adopting a stricter lockout policy and/or increasing the cost of the password hash function by a multiplicative factor $b$ i.e., such that cost of making $G/b$ password guesses after this cost increase is identical to the cost of making $G$ guesses beforehand. In particular, we  derive upper and lower bounds on the security benefit $(\lambda_G - \lambda_{G/b})$ of such policy changes for various values of $b$. We also apply our statistical techniques to analyze several prominent password composition policies. For example, the Yahoo! password frequency corpus $S$ can be divided into two sets $S_{0}$ and $S_{1}$ representing passwords picked before/after Yahoo! adopted six character minimum policy for user passwords. We can then apply our statistical techniques to compare the distributions of user passwords with and without this restriction. If we are willing to adopt the normalized probabilities model \cite{EC:BKPS13},  a heuristic assumption about the way user's react to password composition policies, then we can apply our statistical techniques to quantify the performance of arbitrary password composition policies using other password datasets such as 000webhost and RockYou --- see discussion and analysis in Section \ref{sec:applicationtopasswordpolicies}. 
\subsection{Related Work}

{\bf Password Hashing and Memory Hard Functions:} Key-stretching was proposed as early as 1979~\cite{morris1979password} as a way to protect lower-entropy passwords against offline brute force attacks by making the password hash function moderately expensive to compute. Password hashing algorithms such as BCRYPT~\cite{provos1999bcrypt} and PBKDF2~\cite{kaliski2000pkcs} use hash iteration to control guessing costs, but are potentially vulnerable to an attacker who uses FPGAs or Application Specific Integrated Circuits (ASICs) to dramatically reduce guessing costs. Blocki et al.\cite{SP:BloHarZho18} argued that hash iteration alone cannot provide sufficient protection for user password without introducing an unacceptably long delay during the authentication process i.e., minutes. Memory hard functions such as such as scrypt~\cite{percival2009stronger}, Argon2~\cite{biryukov2016argon2} or DRSample~\cite{CCS:AlwBloHAr17} are designed to force an attacker to allocate large amounts of memory for the duration of computation and are believed to be ASIC resistant. When attempting to tune the cost parameters of our key-stretching algorithm it will be important to understand the password distribution $\lambda_G$ i.e., a defender might want to compute $\lambda_{G}-\lambda_{G/2}$ to decide whether or not doubling the cost parameter would substantially reduce the $\%$ of passwords cracked by an offline attacker. We also note that as organizations start to use moderately expensive memory hard password hash functions like scrypt~\cite{percival2009stronger}, Argon2~\cite{biryukov2016argon2} or DRSample~\cite{CCS:AlwBloHAr17} offline attackers will have stronger incentives to develop guess efficient password cracking models. Thus, it will be desirable to base security decisions on $\lambda_G$ instead of using a password cracking model.   

{\bf Offline Password Cracking Models and Defenses: } Offline password cracking has been studied for decades. Researchers have proposed many password cracking algorithms based on probabilistic password models such as Probabilistic Context-free Grammars (PCFG)~\cite{SP:WAMG09,NDSS:VerColTho14}, Markov models~\cite{CCS:NarShm05b,NDSS:CasDurPer12,SP:MYLL14,durmuth2015omen}, and neural networks~\cite{USENIX:MUSKBCC16}. Monte-Carlo strength estimation \cite{CCS:DelFil15} is a tool which allows us to efficiently approximate the guessing number of a given without requiring the defender to simulate the full attack. Liu et al.~\cite{SP:LNGCU19} developed tools to estimate guessing numbers for software tools such as John the Ripper (JtR)~\cite{johntheripperlink} and Hashcat~\cite{hashcatlink} which are used more frequently by real world attackers. Juels and Rivest~\cite{CCS:JueRiv13} suggested the use of honeywords (fake passwords) to help detect offline attacks. Compromised Credential Checking services such as HaveIBeenPwned and Google Password Checkup can be used to help alert users when one of their passwords have been breached~\cite{USENIX:TPYRKIBPPBB19,CCS:LPASCR19}. Distributed password hashing e.g., \cite{USENIX:ECSJR15} ensures that the information needed to evaluate the password hash function is distributed across multiple servers so that a hacker who breaks into any individual server will not be able to mount an offline attack. Multifactor authentication provides another defense against password cracking attacks~\cite{CCS:BJRSY06,MFA}. 

{\bf Strengthening the Password Distribution: } A large body of research has focused on encouraging (or forcing) users to pick stronger passwords. Password composition policies~\cite{Shay2010,campbell2011impact,Komanduri2011} require users to pick passwords that comply with particular requirements e.g., passwords must at least one number and one upper case letter or passwords must be at least 8 characters long. Password composition policies often induce a substantial  usability burden \cite{Adams1999,Shay2010} and can often be counter-productive~\cite{Komanduri2011,EC:BKPS13}. Telepathwords~\cite{USENIX:KSCHS14} is a password meter which encourages users to select stronger passwords by displaying realtime predictions for the next character that the user will type. Bonneau and Schecter~\cite{USENIX:BonSch14} introduced the notion of incremental password strengthening where users are continually nudged to memorize one more character of a strong $56$-bit password. Another line of work has focused on developing strategies for users to generate stronger passwords e.g., see \cite{yan2004password,CCS:YLCXP16,AC:BloBluDat13,NDSS:BKCD15}. To determine whether a particular intervention (e.g., password generation strategy/composition policy/strength meter/ etc..) strengthened the password distribution we need to characterize the attacker's guessing curve before/after the intervention.

{\bf Empirical Distribution:} The empirical password distribution has been used to evaluate many password research ideas. Harsha et al.~\cite{bicycleAttacks} used empirical distributions derived from the LinkedIn and RockYou datasets to quantify the advantage of an attacker after learning the length of the user's password. The empirical password distribution has also been used to tune and evaluate distribution-aware password throttling mechanisms \cite{blocki2020dalock,tian2019stopguessing} to defend against online attacks, distribution-aware mechanisms to tune relevant cost parameters for password hashing~\cite{bai2021dahash,bai2020information,CSF:BloDat16}, achieve (personalized) password typo correction~\cite{SP:CAAJR16,CCS:CWPCR17} and evaluate the security of Compromised Credential Checking protocols ~\cite{CCS:LPASCR19}.

{\bf Estimating Properties of (Password) Distributions:} Valiant and Valiant~\cite{valiant2017estimating} proposes an approach to accurately estimate key properties of {\em any} distribution over at most $k$ distinct elements using $N=O(k/\log k)$ independent and identically distributed (i.i.d.) samples. However, we cannot directly apply the results of Valiant and Valiant~\cite{valiant2017estimating} to bound $\lambda_G$ in our password setting as the number of distinct passwords $k$ in the support of the password distribution $\mathcal{P}$ is unknown to us and we almost certainly have $k \gg N^2$ even for our largest password datasets. However, we are able to adapt the techniques of Valiant and Valiant~\cite{valiant2017estimating} when developing the linear program that we use to upper/lower bound $\lambda_G$.  Bonneau~\cite{SP:Bonneau12} collected a password frequency corpus derived from approximately $7 \times 10^7$ Yahoo! passwords and used this corpus to estimate properties of the Yahoo! password distribution. Subsampling was one of the key heuristics used to identify stable statistical estimates e.g., Bonneau found that the value $\hat{\lambda}_{{10}}$ was stable under subsampling ~\cite{SP:Bonneau12}  heuristically concluding that $\lambda_{{10}} \approx \hat{\lambda}_{{10}}$. By contrast, we provide rigorous statistical techniques to upper/lower bound $\lambda_G$ directly and precisely bound the probability of an error $\delta$. We also stress that the value $\hat{\lambda}_G$ will not be stable with respect to subsampling for many larger values of $G > {10}$ while our techniques can still be applied to obtain rigorous upper/lower bounds on $\lambda_G$.

\section{Attack Model and Notation}
\subsection{Notation} %{\bf Notation: } 
We let $\mathcal{P}$ denote an arbitrary password distribution over passwords $pwd_1,pwd_2, \ldots$ and we let $p_i$ denote the probability of sampling $pwd_i$. We use $s \leftarrow \mathcal{P}$ to denote a random sample from the distribution and $S=\{s_1,...,s_N\} \leftarrow \mathcal{P}^N$ to denote a multiset of $N$ independent and identically distributed samples from $\mathcal{P}$. We will occasionally abuse notation and also use $\mathcal{P}$ to denote the set of passwords $\{pwd_1,pwd_2,\ldots \}$ in the support of the distribution. We assume that the passwords are ordered in descending order of probability such that $p_1 \geq p_2 \geq p_3 \ldots$ and in general $p_i \geq p_{i+1}$. Note that passwords are sampled with replacement so we could have $s_i = s_j$ for $i<j$. It will be convenient to define $f^{S}_i \doteq \left| \{ j : s_j = pwd_i\}\right|$ as the number of times $pwd_i$ appears in the sample $S$ --- the superscript $S$ may be omitted when the sample set $S$ is clear from context (Note that if $i<j$ we could still have $f^S_i < f^S_{j}$ even though $p_i > p_{j}$ since the passwords $pwd_i$ are ordered by probability \textit{not} by their frequency in the sample $S$.) We will also define $F^S_{i} \doteq |\{pwd_j:f^S_j=i\}|$ as the number of distinct passwords that appear exactly $i$ times in $S$ and denote $F^S=\left( F^S_{1},F^S_{2},...,F^S_{N} \right)$ as the frequency encoding of our sample $S$. Under this notation we have $\mathbf{Unique}(S)=F^S_1$, $\mathbf{Distinct}(S)=\sum_{i\leq N} F^S_i$ and $N=\sum_{i \leq N} i \times F^S_i $. 

It will be convenient to let $L_S$ denote the list of all distinct passwords in $S$ ordered by frequency $f_i^S$ --- ties can be broken in arbitrary order e.g., lexicographic. We also define the set $T(S,G)$ which contains the first $G$ passwords in $L_S$. If $G \geq \mathbf{Distinct}(S)$ then $T(S,G)$ is simply the set of all distinct passwords in the sample $S$. Finally, we will use $\bpdf(i,N,p) \doteq {N \choose i} p^i (1-p)^{N-i}$ to denote the binomial probability density function i.e., if password $pwd$ has probability $p$ and we draw $N$ samples from our password distribution $\mathcal{P}$ then $\bpdf(i,N,p)$ denotes the probability that $pwd$ is sampled exactly $i$ times. 

\subsection{Attacker Model} %{\bf Attacker Model: } 
We consider an attacker who knows the password distribution $\mathcal{P}$ but does not have additional information about the sampled passwords $S \leftarrow \mathcal{P}^N$. In particular, for each $i$ the attacker knows  $pwd_i$ and $p_i$. For each sampled user password $s_i \in S$ the attacker is given $G$ guesses to crack the password $s_i$. For an online attacker $G$ will typically be small as an authentication server can lock the account after several consecutive incorrect login attempts. By contrast, $G$ will be much larger for an offline attacker who has stolen the (salted) cryptographic hash of the user's password can check as many passwords as s/he wants by comparing the (salted) cryptographic hash $h^i=H(u_i, s_i)$ with the hash $h_j^i =H(u_i, pwd_j)$ for each $j \leq G$. An offline attacker is limited only by the resources s/he is willing to invest cracking and by the cost of repeatedly evaluating the password hash function. 

Whether $G$ is large or small the optimal strategy for the attacker is always to check the $G$ most probable passwords $pwd_1,pwd_2,\ldots, pwd_G$ in the distribution. We use the random variable $\lambda(S,G)\doteq \sum_{i\leq G}f^S_i/N$ to denote the percentage of passwords in $S$ cracked within $G$ guesses. Observe that the expected value of $\lambda(S,G)$ is $ \mathbb{E}(\lambda(S,G)) = \sum_{i\leq G}p_i=\lambda_G $. In the next section we will show that the random variable $\lambda(S,G)$  is tightly concentrated around its mean $\lambda_G$ i.e., except with negligible probability we will have $\left| \lambda_G - \lambda(S,G) \right| \leq \epsilon$. In this sense upper/lower bounding $\lambda_G$ and $\lambda(S,G)$ can be seen as (nearly) equivalent problems.

\section{Theoretical Upper/Lower Bounds}\label{sec:theoreticalanalysis}
In this section we introduce several algorithms to generate high-confidence upper bounds and lower bounds on $\lambda_G$ given $N$ independent and identically distributed (iid) samples $S=\{s_1,\ldots s_n\}$ from our password distribution $\mathcal{P}$. An upper bound $\mathsf{UB}(S,G)$ (resp. an lower bound $\mathsf{LB}(S,G)$) derived from the sample $S$ holds with confidence $1-\delta$ if $\Pr[\lambda_G \geq \mathsf{UB}(S,G)] \leq \delta$ (resp. $\Pr[\lambda_G \leq \mathsf{LB}(S,G)] \leq \delta$) where the randomness is taken over the selection of $S \leftarrow \mathcal{P}^N$. Before presenting our results, we first introduce the well-known bounded differences inequality~\cite{mcdiarmid1989method} (also called McDiarmid's inequality), which will be useful in our proofs.
\begin{theorem}(Bounded Differences Inequality~\cite{mcdiarmid1989method})
Suppose that $(X_1,...,X_n)\in\Omega$ are independent random variables. Let $f:\Omega \rightarrow \mathbb{R}$ satisfy the bounded differences property with constants $c_1,...,c_n$, i.e., for all $i\in\{1,...,n\}$ and all $x,x'\in\Omega$ that differ only at the i-th coordinate, the output of the function $|f(x)-f(x')|\leq c_i$. Then,
\begin{align*}
&\Pr[f(X_1,...,X_n)-\mathbb{E}(f(X_1,...,X_n))\geq t] \leq \delta, \\
&\Pr[f(X_1,...,X_n)-\mathbb{E}(f(X_1,...,X_n))\leq -t] \leq \delta,
\end{align*}
where $\delta = \exp\left(\frac{-2t^2}{\sum_{i=1}^n c_i^2}\right)$.
\label{thm:BoundedDifferencesInequality}
\end{theorem}

As an immediate application of McDiarmid's inequality we can prove that, except with negligible probability, we have $\left| \lambda(S,G)-\lambda_G\right| \leq \epsilon$ i.e., $\lambda(S,G)$ is tightly concentrated around its mean $\mathbb{E}[\lambda(S,G)]=\lambda_G$ when the sample size $N$ is large enough --- see Theorem \ref{thm:ExpectationConcentration}. Thus, one strategy to derive a high confidence upper/lower bound for $\lambda_G$ is to derive a high confidence upper/lower bound for $\lambda(S,G)$ and we will immediately obtain a high confidence upper/lower bound for $\lambda_G$ as a corollary. The proof of Theorem~\ref{thm:ExpectationConcentration} is in Appendix~\ref{app:othermissingproofs}. Intuitively, we have $|\lambda(S,G)-\lambda(S',G)|\leq 1/N$ whenever $S=\{s_1,...,s_i,...,s_N\}$ and $S'=\{s_1,...,s'_{i},...,s_N\}$ is obtained by swapping out $s_i$ for $s'_i$. Thus, we can apply McDiarmid's inequality.

\newcommand{\expectationconcentration}{For any guessing number $G \geq 0$ and any $0 \leq \epsilon \leq 1$ we have:
\begin{align*}
    \Pr[\lambda(S,G) \leq \lambda_G + \epsilon] & \geq 1 - \exp\left(-2N\epsilon^2\right)  \mbox{~, and} \\
    \Pr[\lambda(S,G) \geq \lambda_G - \epsilon] & \geq 1 - \exp\left(-2N\epsilon^2\right) 
\end{align*}
where the randomness is taken over the sample set $S \leftarrow \mathcal{P}^N$ of size $N$.}

\begin{theorem}\label{thm:ExpectationConcentration} 
\expectationconcentration
\end{theorem}

\subsection{Empirical Distribution as an Upper Bound}\label{sec:upperbound-priorwork}
In this section we show that we can upper bound $\lambda_G$ using the empirical distribution $\hat{\lambda}_G$. In particular, we argue that for any sample $S$ we have $\hat{\lambda}_G \geq \lambda(S,G)$. As an immediate corollary we get that $\Pr[\lambda_G > \hat{\lambda}_G+\epsilon] \leq \delta$ where $\delta = \exp\left(-2N \epsilon^2 \right)$.

\begin{theorem}\label{thm:upperbound-priorwork}
For any sample set $S \leftarrow \mathcal{P}^N$ with size $N$ and any $G >0$ we have $\lambda(S,G) \leq \frac{1}{N}\sum_{i:pwd_i\in T(S,G)} f_i^S$.
\end{theorem}
\begin{proof}
We observe that $\sum_{i:pwd_i\in T(S,G)} f_i^S = \max\limits_{1\leq j_1<\cdots<j_G\leq |P|}\sum_{i\in \{j_1,...,j_G\}} f_i^S  \geq \sum_{i \leq G} f_i^S = N \lambda(S,G)$.
It follows that $\lambda(S,G) \leq \frac{1}{N}\sum_{i\in T(S,G)} f_i^S$ as claimed.
\end{proof}
 Recall that $T(S,G)$ is the set of $G$ most frequent passwords in the sample $S$. 
Applying Theorem~\ref{thm:ExpectationConcentration} to Theorem~\ref{thm:upperbound-priorwork}, we can immediately obtain the upper bound of $\lambda_G$ as below:
\begin{corollary}\label{crl:upperbound-priorwork}
For any guessing number $G\geq 0$ and $\epsilon>0$  we have
$$\Pr\left[ \lambda_G \leq \frac{1}{N}\sum\nolimits_{i\in T(S,G)} f_i + \epsilon \right] \geq 1 - \exp\left(-2N\epsilon^2\right),$$
where the randomness is taken over the selection of $S\leftarrow \mathcal{P}^N$.
\end{corollary}

\subsection{A Lower Bound for $G < N$}\label{sec:bound1}
%The lower bound in Section~\ref{sec:lowerbound-priorwork} only applies $G\geq N$. In this section, we introduce a new idea that can generate tighter lower bound for small guessing number $G<N$.
In this section, we introduce a new idea to lower bound $\lambda_G$. The key idea is to randomly partition our sample  $S$ into two datasets $D_1=\{s_1,...,s_{N-d}\}$ and $D_2=\{s_{N-d+1},...,s_N\}$ with size $N-d$ and $d$. We can then construct a dictionary $T(D_1,G)$ containing the top $G$ passwords in $D_1$ which is used to attack passwords in $D_2$. In particular, we let $h(D_1,D_2,G) = |\{N-d+1 \leq  i \leq N :  s_i \in T(D_1,G)\}|$ denote the number of samples in $D_2$ that are also in $T(D_1,G)$. Because we can view $D_2$ as $d$ independent samples from the password distribution $\mathcal{P}$ we have  $\mathbb{E}\left[h(D_1,D_2,G)/d\right] \leq \lambda_G$ since $\lambda_G$ denotes the expected fraction of passwords in $D_2$ that are cracked using the optimal dictionary instead of $T(D_1,G)$. Adding a slack term $t/d$ we can use the Bounded Differences Inequality to argue that $\Pr[\lambda_G \leq \frac{1}{d}(h(D_1,D_2,G) - t)] \leq \delta $ where $\delta=\exp(-2t^2/d)$.

We remark that $h(D_1,D_2,G)/d$ can also be viewed as the success rate of an attacker who has partial knowledge of the password distribution. In particular, suppose that the attacker has $N-d$ samples $D_1$ from $\mathcal{P}$ e.g., obtained by cracking the corresponding password hashes or by some other means such as phishing. Then the attacker can construct the dictionary $T(D_1,G)$ and use this dictionary to crack the remaining passwords in $D_2$. We can view $h(D_1,D_2,G)$ as the number passwords in $D_2$ that that the attacker would crack within $G$ guesses.  
% Given a parameter $d$, we partitionSince samples in $S$ are independently randomly sampled from the password distribution $\mathcal{P}$, $D_1$ and $D_2$ can be considered as two independent sample sets from $\mathcal{P}^{N-d}$ and $\mathcal{P}^d$ respectively. Let $h(D_1,D_2,G) = |\{N-d+1 \leq  i \leq N :  s_i \in T(D_1,G)\}|$ be the number of samples in $D_2$ that are also in $T(D_1,G)$. Recall that $T(D_1,G)$ is the set of $G$ most frequent passwords in $D_1$. Adding a slack term $t/d$ we argue that $\Pr[\lambda_G \leq \frac{1}{d}(h(D_1,D_2,G) - t)] \leq \delta $ where $\delta=\exp(-2t^2/d)$.

\newcommand{\samplinglbthm}{
For any guessing number $G \geq 1$ and any parameters $0 < d < N$ and $t \geq 0$, we have
\begin{align*}
   \Pr[\lambda_G \geq \frac{1}{d}(h(D_1,D_2,G) - t)] \geq 1 - \exp\left(-2t^2/d\right)
\end{align*}
where the randomness is taken over the samples $D_1 \leftarrow \mathcal{P}^{N-d}$ and $D_2 \leftarrow \mathcal{P}^d$. 
}
\begin{theorem}\label{thm:bound1}
\samplinglbthm
\end{theorem}

When applying Theorem~\ref{thm:bound1} we can select $d \ll N$ and set $t = \sqrt{(d/2) \ln (1/\delta) } $ to get ensure that $t/d=o(1)$ is small and our probability of error is at most $\delta$. \fullversion{The proof of Theorem~\ref{thm:bound1} is deferred to the full version \cite{fullversion} of the paper. The full version \cite{fullversion} also contains a corollary lower bounding $\lambda(S,G)$ using Theorems~\ref{thm:bound1} and \ref{thm:ExpectationConcentration}.}{The proof of Theorem~\ref{thm:bound1} is in Appendix~\ref{app:othermissingproofs}. Corollary \ref{crl:bound1} lower bounding $\lambda(S,G)$ using Theorems~\ref{thm:bound1} and \ref{thm:ExpectationConcentration} is in Appendix~\ref{app:lambdaSG}.} 

% \fullversion{--- due to space limitations all corollaries bounding $\lambda(S,G)$ are deferred to the full version}{--- see Corollary \ref{crl:bound1} in Appendix~\ref{app:lambdaSG}}. As a corollary of Theorem~\ref{thm:ExpectationConcentration} we can also derive a high confidence lower bound for $\lambda(S,G)$ 

The lower bound in Theorem~\ref{thm:bound1} is often tight for smaller values of $G \ll N$, but it will plateau at $G=\mathbf{Distinct}(S)$ since the set $T(D_1,G)$ already contains all of the passwords in $D_1$. When $d \ll N$ and $G = \mathbf{Distinct}(S)$ the lower bound will closely match the Good-Turing estimate $1-\frac{\mathbf{Distinct}(S)}{N}$ for the total probability mass of all passwords in the set $S$.

\subsection{An Extended Lower Bound Using Password Models}\label{sec:extendedLB}
The lower bound from  Section~\ref{sec:bound1} plateau's when  $G \geq \mathbf{Distinct}(S)$. Is it possible to derive tighter lower bounds for larger values of $G$? In this section we show that any password cracking model can be used to derive high confidence lower bounds on $\lambda_G$ and then show how our prior lower bound can be combined with a password cracking model $M$ to derive tighter bounds. 

Let $M(D_1, G)$ be the set of top $G$ password guesses output by an attack model $M$ trained on $D_1$.  We now follow the same approach as before and partition $S$ into two sets $D_1=\{s_1,...,s_{N-d}\}$, $D_2 = \{s_{N-d+1},...,s_{N}\}$. Let $h'_M(D_1,D_2,G)$ be the number of passwords cracked in $D_2$ by making guesses in $M(D_1, G)$. We can prove a generalized lower bound of $\lambda_G$ --- see Theorem \ref{thm:bound2} below.

\begin{theorem}\label{thm:bound2}
For any guessing number $G>0$ and any parameters $0 < d < N$, $t\geq 0$ we have:
$$ \Pr[\lambda_G\geq \frac{1}{d}(h'_M(D_1,D_2,G) - t)]\geq  1-\exp\left(-2t^2/d\right)$$
where the randomness is taken over the sample $S$ of size $N$.
\end{theorem}
\begin{proof}
Since $p_1,p_2,...,p_i,...$ are sorted in decreasing order, we have $\mathbb{E}(h'_M(D_1,D_2,G)) = d\times \sum_{i:pwd_i\in M(D_1,G)}p_i \leq d\times \sum_{i \leq G}p_i = d\times \lambda_G$. Using  Theorem \ref{thm:BoundedDifferencesInequality}, we have:
\begingroup\makeatletter\def\f@size{9.5}\check@mathfonts
\begin{align*}
& \Pr[h'_M(D_2,G) \leq d \sum_{i:pwd_i\in M(D_1,G)}p_i + t] \geq  1-\exp\left(-2t^2/d\right) \\
    & \Rightarrow \Pr[\lambda_G \geq \frac{1}{d}(h'_M(D_2,G) - t)]  \geq 1-\exp\left(-2t^2/d\right)
\end{align*}
\endgroup
%\begin{align*}
% & \Pr[h'_M(D_2,G) \leq d\times \sum_{i:pwd_i\in M(D_1,G)}p_i + t] 
%    \geq  1-\exp\left(-2t^2/d\right) \\
%    & \Rightarrow \Pr[\lambda_G \geq \frac{1}{d}(h'_M(D_2,G) - t)]  \geq 1-\exp\left(-2t^2/d\right)
%\end{align*}
where the last line follows since $\lambda_G \geq \sum_{i:pwd_i\in M(D_1,G)}p_i$.
\end{proof}
As an immediate corollary of Theorem \ref{thm:bound2} we can also lower bound $\lambda(S,G)$ \fullversion{--- see the full version \cite{fullversion}}{--- see Corollary \ref{cor:modelLbound} in Appendix~\ref{app:lambdaSG}}. As a more useful corollary given any model $M$ we can define a hybrid attack model $M^*(D_1,G)$ which, given a dataset $D_1$, first constructs a dictionary $T(D_1,G)$ containing the top $G$ passwords in $D_1$. When trying to crack a new password $pwd \leftarrow \mathcal{P}$ the model $M^*(D_1,G)$ begins by checking each of the passwords in this dictionary $T(D_1,G)$. If $G > \mathbf{Distinct}(D_1)$ and $pwd$ does not appear in the dictionary $T(D_1,G)$ then the model proceeds to generate the remaining $G'=G - \mathbf{Distinct}(D_1)$ guesses using the model $M$ --- we can optionally omit guesses which already appear in our training dataset $D_1$. Note that for $G \leq \mathbf{Distinct}(D_1)$ we have $h'_{M^*}(D_1,D_2,G) = h(D_1,D_2,G)$ where $h(D_1,D_2,G)$ counts the number of samples in $D_2$ that appear in the top $G$ samples from $D_1$. For $G \geq \mathbf{Distinct}(D_1)$ the function $h'_{M^*}(D_1,D_2,G)$ counts the number of samples in $D_2$ that either (1) appear in $D_1$ or (2) appear in $M(D_1,G')$. Thus, at minimum we always have $h'_{M^*}(D_1,D_2,G) \geq \max\{ h(D_1,D_2,G), h'_{M}(D_1,D_2,G')\}$ where $G'=G-\mathbf{Distinct}(D_1)$. Intuitively, the lower bound will be at least as good as our prior approach from Theorem \ref{thm:bound1} and at least as good as the model $M$.

\begin{corollary}\label{crl:bound2}
Let $M$ be a password cracking model and $M^*$ be the corresponding hybrid attack model. Let parameters $G, d>0$ , $t >0$ be given then
\[
\Pr[\lambda_G\geq \frac{1}{d}(h'_{M^*}(D_1,D_2,G) - t)]\geq  1-\delta 
\]
%\begin{align*}
%&\Pr[\lambda_G\geq \frac{1}{d}(h'_{M^*}(D_1,D_2,G) - t)]\geq  1-\delta \\ %%%%\exp\left(\frac{-2t^2}{d}\right) \\
%&\Pr[\lambda(S,G) \geq \frac{1}{d}(h'_{M^*}(D_1,D_2,G) - t) - \epsilon] \\
%&\geq 1 - \exp\left(-2N\epsilon^{2}\right) - \exp\left(\frac{-2t^2}{d}\right)
%\end{align*}
where $\delta = \exp\left(-2t^2/d\right)$ and the randomness is taken over the set $S$ of size $N$.
\end{corollary}

We can also view the lower bound from Corollary \ref{crl:bound2} as denoting the success rate of a hybrid attacker i.e., an attacker who has obtained a cracked/leaked passwords $D_1$ and runs the hybrid attack described above will crack $h_{M^*}'(D_1,D_2,G)$ of the remaining passwords in $D_2$ within $G$ guesses. Corollary~\ref{crl:bound2-lambdaSG} lower bounding $\lambda(S,G)$ using Corollary \ref{crl:bound2} is in \fullversion{the full version \cite{fullversion}.}{Appendix~\ref{app:lambdaSG}.}

\subsection{Upper And Lower Bounds Using Linear Programming}\label{sec:LPbounds}
%In the previous section we showed how to use a password cracking model $M$ to extend our prior sampling based lower bound on $\lambda_G$ when $G > \mathbf{Distinct}(S)$. However, if the password cracking model $M$ is not guess efficient it is still possible that the lower bound will (temporarily) plateau at $G = \mathbf{Distinct}(S)$. 

In this section we propose a different approach to generate upper and lower bounds using linear programming which is inspired by work of Valiant and Valiant~\cite{valiant2017estimating}. As we will see in our empirical analysis these LP bounds tend to be tighter when the guessing number is large though the bounds are slightly worse when the guessing number $G$ is smaller i.e., $G \ll \mathbf{Distinct}(S)$. 

Intuitively, we derive our upper (resp. lower) bounds by designing a linear program to find a distribution $\mathcal{P}'$ that maximizes (resp. minimizes) $\lambda_G'$ subject to various consistency constraints derived from our sample $S \leftarrow \mathcal{P}^n$. To simplify our exposition it is helpful to begin with a simplifying assumption that the probability distribution $\mathcal{P}$ can be represented as a histogram $h_1,\ldots, h_{\ell}$ over a finite probability mesh $X=\{x_1,\ldots,x_{\ell}\}$ i.e., there are exactly $h_i$ passwords in the support of the distribution which have probability exactly $x_i$. We will later show how this simplifying assumption can be removed. 

The values $h_i$ are our unknown variables in our linear program. We first note that any valid probability histogram must satisfy the constraints that $h_1x_1+\ldots + h_{\ell} x_{\ell} = 1$ and $0 \leq h_i$. More significantly, we adapt ideas from Good-Turing Frequency estimation to argue that the known value  $P_{i}^S \doteq \frac{(i+1) F_{i+1}^S}{N-i}$ will be close to the (unknown value) $\sum_{j=1}^{\ell} h_jx_j\cdot \bpdf(i, N, x_j)$ with high probability i.e., we can compute $P_{i}^S$ and add linear constraints on the variables $h_1,\ldots, h_{\ell}$ to ensure that $P_{i}^S \approx \sum_{j=1}^{\ell} h_jx_j\cdot \bpdf(i, N, x_j)$. We include this constraint for each value of $i \leq i'$ where $i'$ is a parameter that will be selected later. We remark that the original password distribution $\mathcal{P}$ will be consistent with each of these constraints with high probability. Thus, with high probability, minimizing (resp. maximizing) $\lambda_G'$ subject to the relevant constraints yields a lower (resp. upper) bound on $\lambda_G$. We can remove our simplifying assumption that $\mathcal{P}$ is consistent with a finite probability mesh by (1) adding a new variable $p$ which intuitively represents the probability mass of all passwords in the distribution with probability $\leq x_{\ell}$, (2) ensuring that the probability mesh is sufficiently fine-grained that every probability value $1 \geq p \geq x_{\ell}$ is close to some point on the mesh, and (3) adding slack terms to each constraint to ensure that the probability histogram for $\mathcal{P}$ is still consistent with all of the constraints even after rounding probability values to fit the mesh.   

We remark that if the dataset $S$ was not sampled independently from some (unknown) distribution $\mathcal{P}$ then it is possible that there will be {\em no} feasible solution to our linear program. Thus, as a side-benefit our linear program can allow us to identify datasets which are inconsistent with our iid assumption. As a concrete example, suppose that the dataset $S$ was sampled by picking $s_1,\ldots, s_{N/2} \leftarrow \mathcal{P}$ independently and then duplicating the last $N/2$ passwords i.e., $s_{N/2 +i} = s_i$ for $i \leq N/2$. The resulting dataset would  (likely) be correctly rejected by our linear program since it is not iid. 

% One of the first tasks will be to identify useful consistency constraints that should hold with high probability over the selection of $S$. This ensures that with high probability (whp) the original password distribution \added{(from which $S$ was sampled) appears in the set of feasible distributions}.  \added{We remark that if the dataset $S$ was not sampled independently from some unknown distribution $\mathcal{P}$ then it is possible that our linear program will have no feasible solutions. This might occur if many entries in the dataset were duplicated e.g., suppose that the first $3N/4$ passwords  $s_1,\ldots, s_{3N/4} \leftarrow \mathcal{P}$ are independently sampled, but the last $N/4$ passwords are duplicates ($s_{3N/4 +i} = s_i$ for $i \leq N/4$) instead of fresh samples. We view ability to  reject datasets which were clearly not independently sampled as a side benefit of our linear program.   }

%\subsubsection{An Ideal Linear Programming Task} 
\subsubsection{A Linear Programming Task with Idealized Settings}
We first describe our linear program in the idealized setting where we assume that our (unknown) password distribution $\mathcal{P}$ is consistent with a finite probability mesh. In particular, we will fix a finite probability mesh $X_{\ell}=\{x_1,\ldots, x_{\ell}\}$ with $x_1 > x_2 > \cdots x_{\ell}$ and assume that for all passwords $pwd_i$ we have $p_i \in X$ i.e., the mesh contains every probability value in our distribution. This allows us to view the original distribution $\mathcal{P}$ as a histogram $H=h_1,\ldots, h_\ell$ where $h_i$ denotes the number of items in the support of $\mathcal{P}$ which occur with probability $x_i$. We use $h_1, \ldots, h_\ell \geq 0$ as variables in our linear program (relaxing the natural constraint that $h_i$ is an integer). Now the linear constraint $\sum_j h_j x_j=1$ encodes the requirement that our probabilities sum to $1$.

%except with probability $2\times\exp\left(\frac{-2(N-i)^2\epsilon_{2,i}^2}{N(i+1)^2}\right)$, we will  have $\frac{(i+1)F^S_{i+1}}{N-i}-\epsilon_{2,i} - \frac{i+1}{N-i} \leq \sum_{j=1}^{\added{\ell}} h_j\times x_j\times \bpdf(i,N,x_j) \leq \frac{(i+1)F^S_{i+1}}{N-i}+\epsilon_{2,i}$ }

Given our sample $S$ of size $N$ recall that $F^S_i$ denotes the number of distinct passwords that appear exactly $i$ times in our sample $S$. Thus, the expected value of $F^S_i$ is $\sum_{j=1}^{\ell} h_j \bpdf(i,N,x_j)$ and $\sum_{j=1}^{l} h_j \times x_j \times \bpdf(i,N,x_j)$ is the expected probability mass of all items that were sampled exactly $i$ times. Adapting ideas from Good-Turing Frequency estimation we can argue that (whp) $\sum_{j=1}^{l} h_j \times x_j \times \bpdf(i,N,x_j)$ will be close to $P_i^S = \frac{(i+1)F^S_{i+1}}{N-i}$. 
    
In particular, Lemma \ref{lem:goodturing}  shows that for each frequency $i$ we will have $P_i^S-\epsilon_{2,i} - \frac{i+1}{N-i} \leq \sum_{j=1}^{\ell} h_j\times x_j\times \bpdf(i,N,x_j) \leq P_i^S+\epsilon_{2,i}$ with high probability. Thus, we will include this linear constraint for each $i \leq i'$ where $i'$ is a parameter that can be tuned. Intuitively, increasing $i'$ adds additional constraints to our linear program which reduces the feasible region and can only improve the upper/lower bound. However, increasing $i'$ can also decrease our confidence $\delta$  since we need to ensure that $\mathcal{P}$ is consistent with {\em all} of the constraints that we generate to argue that the upper/lower bounds are valid.  

%In particular, Lemma \ref{lem:goodturing} shows that, except with probability $2\times\exp\left(\frac{-2(N-i)^2\epsilon_{2,i}^2}{N(i+1)^2}\right)$, we will  have $\frac{(i+1)F^S_{i+1}}{N-i}-\epsilon_{2,i} - \frac{i+1}{N-i} \leq \sum_{j=1}^{\added{\ell}} h_j\times x_j\times \bpdf(i,N,x_j) \leq \frac{(i+1)F^S_{i+1}}{N-i}+\epsilon_{2,i}$ where $\epsilon_{2,i}>0$ is a small constant. Thus, if we ensure that $\sum_{i=1}^{i'} 2\times\exp\left(\frac{-2(N-i)^2\epsilon_{2,i}^2}{N(i+1)^2}\right) \leq \epsilon$ we can add the constraints from Lemma \ref{lem:goodturing}  for each $i \leq i'$ and ensure that, except with probability $\epsilon$, that all of these constraints will hold. In this case the original distribution $\mathcal{P}$ defines a feasible solution to the linear program. 

Intuitively, if we search for a distribution that maximizes (resp. minimizes) $\lambda_G$ we obtain an upper bound (resp. lower bound) since the real distribution $\mathcal{P}$ will be one of the feasible solutions (whp). The remaining challenge is to encode $\lambda_G$ as a linear objective function using the variables $h_1,\ldots,h_{\ell}$. If we happened to know the integers $c,idx$ such that $G= c+ \sum_{j < idx} h_j$ and $0 \leq c \leq h_{idx}$ then we easily could encode $\lambda_G = c\cdot x_{idx} + \sum_{j < idx} h_i \cdot x_i $ as a linear objective function. However, we cannot compute $c$ or $idx$ a priori since the values $h_1,\ldots, h_{\ell}$ are unknown. We deal with this challenge by introducing a separate linear program for each possible value of $idx$ and by adding a new variable $c$ along with the constraints that $0 \leq c \leq h_{idx}$ and $c=G-\sum_{j < idx} h_j$. Letting $y_{idx}^*$ denote the value of optimal solution to our LP with the parameter $idx$ we can combine these solutions to get our final upper/lower bound i.e., $y^* = \max\{y_{idx}^* : idx \leq \ell\}$ (resp. $y^* = \min\{y_{idx}^* : idx \leq \ell\}$) when computing our upper bound (resp. lower bound) where each $y_{idx}^*$ is the value we obtain when maximizing (resp. minimizing) the corresponding LP.  

%Given a histogram $H$ and a guessing number $G$ we can define $i(G,H) = \max \{ j : \sum_{i=1}^{j-1} h_i \leq G \} $ and $c(G,H) = G- \sum_{i=1}^{i(G,H)-1} h_i$. Observe that if the attacker attempts $G$ password guesses they will succeed with probability $c(G,H) x_{i(G,H)} +  \sum_{i=1}^{i(G,H)-1} x_i h_i$ i.e., for $i < i(G,H)$ they will check all $h_i$ passwords with probability $x_i$ and for $i=i(G,H)$ they can only check $c(G,H) < h_{i(G,H)}$ of the  passwords which occur with probability $x_{i(G,H)}$. We would like to find $h_1,\ldots, h_{\ell}$ to maximize (or minimize) $c(G,H) x_{i(G,H)} +  \sum_{i=1}^{i(G,H)-1} x_i h_i$. However, the optimization goal is not quite linear due to the dependence on $i(G,H)$. 

%To solve this problem we introduce another parameter $idx$ which represents a guess for the value of  $i(G,H)$ and we define a separate linear program for each possible guess $1 \leq idx \leq l$. Fixing $idx$ we introduce a new variable $c$ which intuitively corresponds to $c(G,H)$ and we add the constraints that $0 \leq c \leq h_{idx}$ and that $c = G- \sum_{i=1}^{idx-1} h_i$. Now for each fixed $idx$ the objective $c x_{idx} +  \sum_{i=1}^{idx-1} x_i h_i$ is linear and can be maximized/minimized. 

%%%%This is where LP1 originally located%%%

Our linear program $\mathtt{LP1}(G,b,X_{\ell},F^S,idx,i',\mathbf{\epsilon_2})$ is shown below (Linear Programming Task 1). The inputs include the guessing budget $G$, the probability mesh $X_{\ell} = \{x_1,\ldots, x_{\ell}\}$, the set $F^S=\{F^S_1,\ldots, F^S_N\}$, a bit $b \in \{-1,1\}$ which indicates whether we are looking for an upper or lower bound, the (guessed) value $idx$ and parameters $i'$ and $\mathbf{\epsilon_2}=\{\epsilon_{2,0},\ldots,\epsilon_{2,i'}\}$ related to the consistency constraints. 

\fbox{\begin{minipage}{8cm}
\textbf{Linear Programming Task 1: }\\
$\mathtt{LP1}(G,b,X_{\ell},F^S,idx,i',\mathbf{\epsilon_2})$ \\
\textbf{Input Parameters:} $G$, $b$, $X_{\ell}=\{x_1,\ldots,x_{\ell}\}$, $F^S=\{F^S_1,\ldots,F^S_N\}$, $idx$, $i'$, $\mathbf{\epsilon_2}=\{\epsilon_{2,0},\ldots,\epsilon_{2,i'}\}$ \\
\textbf{Variables:} $h_1,\ldots,h_{\ell},c$ \\
%\textbf{Objective:} $\min\left(b\times (\sum_{j<idx}h_j+c)\right)$ \\ 
\textbf{Objective:} $\min\left(b\times (\sum_{j<idx}h_j\times x_j +c\times x_{idx})\right)$ \\ 
\textbf{Constraints:} 
\begin{enumerate}
    %\item $(1-\epsilon_1)Q \leq \sum_{j<i}h_j\times x_j +c\times x_i \leq (1-\epsilon_1)Q$
    \item $\sum_{j<idx}h_j+c = G$
    %\item $\forall 0\leq i\leq i'$, $ \frac{(i+1)F^S_{i+1}}{N+1} - \epsilon_2 \leq \sum_{j=1}^{\added{\ell}} h_j\times x_j\times \bpdf(i,N,x_j) \leq \frac{(i+1)F^S_{i+1}}{N+1} + \epsilon_2$ 
    \item $\forall 0\leq i\leq i'$, $ \frac{(i+1)F^S_{i+1}}{N-i}-\epsilon_{2,i} - \frac{i+1}{N-i} \leq \sum_{j=1}^{\ell} h_j\times x_j\times \bpdf (i,N,x_j) \leq \frac{(i+1)F^S_{i+1}}{N-i}+\epsilon_{2,i}$ 
    \item $ \sum_{j=1}^{\ell} h_j\times x_j = 1$
    \item $0 \leq c \leq h_{idx}$
\end{enumerate}
\end{minipage}}

The following lemma indicates that constraint (2) holds with high probability over the selection of $S \leftarrow \mathcal{P}^N$ when we select the parameters $i'$ and $\mathbf{\epsilon_2}=\{\epsilon_{2,0},\ldots, \epsilon_{2,i'}\}$ properly.
\newcommand{\lemmagoodturingideal}{
For any $i\geq 0$ and $0\leq \epsilon_{2,i}\leq 1$, we have $\frac{(i+1)F^S_{i+1}}{N-i}-\epsilon_{2,i} - \frac{i+1}{N-i} \leq \sum_{j} h_j\times x_j\times \bpdf(i,N,x_j) \leq \frac{(i+1)F^S_{i+1}}{N-i}+\epsilon_{2,i}$ with probability at least $1 - 2\times\exp\left(\frac{-2(N-i)^2\epsilon_{2,i}^2}{N(i+1)^2}\right)$ where the probability is taken over the selection of our sample $S \leftarrow \mathcal{P}^N$.%\mathbb{S}^N$. 
}
\begin{lemma}\label{lem:goodturing}
\lemmagoodturingideal
\end{lemma}

The proof of Lemma~\ref{lem:goodturing} is in Appendix~\ref{app:LPmissingproofs}. Using Lemma~\ref{lem:goodturing} and conclusions we present above, we can show that the upper/lower bounds of $\lambda_G$ derived from our linear program hold with high-confidence. In particular, for any constant $\delta >0$ the parameters $\epsilon_{2,i}$ can be tuned such that, except with probability $\delta$, we have the lower bound (resp. upper bound) $\lambda_G \geq \min\limits_{1\leq idx \leq \ell} \mathtt{LP1}(X_{\ell},F^S,idx,G,1,i',\mathbf{\epsilon_2})$ (resp. $\lambda_G \leq \max\limits_{1\leq idx \leq \ell} \mathtt{LP1}(X_{\ell},F^S,idx,G,-1,i',\mathbf{\epsilon_2})$. See Theorem~\ref{thm:bound3ideal} in Appendix~\ref{app:LPmissingproofs} for a formal statement\fullversion{.}{ and see Corollary~\ref{crl:bound3ideal} in Appendix~\ref{app:lambdaSG} for a corollary upper and lower bounding $\lambda(S,G)$. }  

%\added{To prove this lemma, first we can find that $\sum_{j} h_j\times x_j\times \bpdf(i,N,x_j)$ can be bounded by $\frac{i+1}{N-i}\mathbb{E}(F^S_{i+1})$ with a small term. Then using Theorem~\ref{thm:BoundedDifferencesInequality} we can show that $F^S_{i+1}$ is highly concentrated on $\mathbb{E}(F^S_{i+1})$ when $i$ is small. The full proof can be found in Appendix~\ref{app:LPmissingproofs}.} When the sample size $N$ is large enough and $i$ is small, the bounds we prove in the lemma above hold with high probability by selecting proper value of $\epsilon_{2,i}$.

%As a corollary of Theorem \ref{thm:bound3ideal} we can also bound $\lambda(S,G)$ --- \fullversion{see the full version}{see Corollary \ref{crl:bound3ideal} in Appendix~\ref{app:lambdaSG}}. 

\subsubsection{Intermediate Step: Linear Programming with Countably Infinite Probability Mesh}\label{sec:midLP} In this section we show how to relax the assumption that the real password distribution $\mathcal{P}$ is consistent with a finite probability mesh and replace it with a slightly weaker assumption that $\mathcal{P}$ is consistent with a {\em countably infinite} mesh $X=\{x_1,x_2,...,x_{\ell},x_{\ell+1},...\}$ with $x_1>x_2>x_3 >\cdots$ and $\lim_{i \rightarrow \infty} x_i = 0$. As before we assume that for all passwords $pwd_i$ in the support of $\mathcal{P}$ the probability $p_i \in X$ of sampling this password lies in the mesh $X$. Suppose that $H = (h_1,h_2,\ldots)$ is a histogram encoding of $\mathcal{P}$ and let $p= \sum_{i \geq \ell+1} x_i h_i$ be the total probability mass of all passwords in $\mathcal{P}$ with probability smaller than $x_{\ell}$. Our key idea to ensure that our linear programs remain finite is to introduce a new variable for $p$ and eliminate the variables $h_j$ for each $j > \ell$. For example, we now have the constraints that $p + x_1 h_1+\ldots +x_{\ell}h_{\ell}=1$ and that $0\leq p \leq 1$.

%Our prior linear programming based lower bound relied on an idealized assumption that the real distribution $\mathcal{P}$ is consistent with a finite probability mesh $X_{\ell} = \{x_1,\ldots, x_{\ell}\}$ i.e., for all $i$ the probability $p_i \in X_{\ell}$ of password $pwd_i$ lies in the mesh. In this section we show how to extend the prior approach when $\mathcal{P}$ is consistent with an infinite mesh  In the next section we will show how to remove the assumption that $\mathcal{P}$ is perfectly consistent with $X$ as long as the mesh $X$ is sufficiently fine-grained.  For every integer $l$ we can define the mesh $X_{\ell} = \{x_1,\ldots, x_{\ell}\}$ which contains the top $l$ values in $X$.  Observe that

Lemma~\ref{lem:goodturing} still applies so (whp) we have $(i+1)F^S_{i+1}/(N-i)-\epsilon_{2,i} - \frac{i+1}{N-i} \leq \sum_{j=1}^{\infty} x_j h_j \bpdf(i,N, x_j) \leq (i+1)F^S_{i+1}/(N-i)+\epsilon_{2,i}$ i.e., the expected probability mass of all items that appear $i$ times in our sample $S \leftarrow \mathcal{P}^N$ ($i\geq 0$) is still $\sum_{j=1}^{\infty} x_j h_j \bpdf(i,N, x_j)$. However, we cannot add these constraints to our linear program since we do not have variables for $h_j$ when $j > \ell$. Instead, we rely on the following observations (1) The function $f(x) = \bpdf(0,N,x)$ is monotonically decreasing over the domain $[0,1]$ so we have $\bpdf(0,N,x_{\ell}) < \bpdf(0,N,x_j)$ whenever $j > \ell$. Thus, we can bound the partial sum $\sum_{j=\ell+1}^\infty x_j h_j \bpdf(0,N,x_j)$ as follows $p \geq \sum_{j=\ell+1}^\infty x_j h_j \bpdf(0,N,x_j) \geq p \cdot  \bpdf(0,N,x_{\ell})$. (2) For $i > 0$ the function $f_i(x) = \bpdf(i,N,x)$ is monotonically increasing over the domain $[0,1/N]$. Thus, for $i>0$ we have $\bpdf(i,N,x_{\ell}) > \bpdf(i,N,x_{j})$ whenever $x_j < x_{\ell} < \frac{1}{N}$. In this case we can bound the partial sum $\sum_{j=\ell+1}^\infty x_j h_j \bpdf(i,N,x_j)$ as follows  $p \cdot  \bpdf(i,N, x_{\ell}) \geq  \sum_{j=\ell+1}^\infty x_j h_j \bpdf(0,N,x_j) \geq 0$.   

We can use the above observations to update Constraint (2) in $\mathtt{LP1}$. Similarly, we can replace Constraint (3) in $\mathtt{LP1}$ with $\sum_{j=1}^{\ell} h_j x_j = 1 - p$. We call this updated linear program $\mathtt{LP1a}$%\fullversion{}{ as shown below}. \fullversion{Since this is only an intermediate step, we defer the full description of $\mathtt{LP1a}$ and associated theorems to the full version \cite{fullversion}. }{}
Since this is only an intermediate step, we defer the full description of $\mathtt{LP1a}$ and associated theorems to \fullversion{the full version \cite{fullversion}}{Appendix~\ref{app:IntermediateLPsProofs}}.
We briefly remark that for $i>0$ (resp. $i=0$) we have $\lim_{x  \rightarrow 0} \bpdf(i,N, x) = 0 $ (resp. $\lim_{x \rightarrow 0} \bpdf(0,N,x) = 1$) so the upper/lower bounds on our partial sums are essentially tight and we are not substantially loosening our constraints in $\mathtt{LP1a}$. 

\subsubsection{Final LP: Eliminating Ideal Mesh Assumptions}\label{sec:finalLP-mainbody}

In this section, we present our final linear program to bound $\lambda_G$ without making any idealized assumptions about the distribution $\mathcal{P}$.  Following~\cite{valiant2017estimating}, we fix a particular mesh $X_{\ell,q} = \{x_1,\ldots, x_{\ell} \}$ where we have  $x_{i} = q \cdot x_{i+1} = q^{\ell-i} x_{\ell}$ for each $1 \leq i < \ell$. Note that the parameter $q > 1$ determines how fine-grained our mesh values are. We will pick $x_\ell$ to be suitably small (e.g., $x_\ell = 10^{-4} N^{-1}$) and set $\ell= \left \lfloor\frac{\ln (\frac{1}{x_{\ell}})}{\ln q}\right\rfloor + 1$ so that $x_1 \approx 1$.

%$x_{\ell} = \frac{1}{10^4 N}$, $l=\lfloor\frac{\ln (\frac{1}{x_{\ell}})}{\ln q}\rfloor + 1$ and for each $1 \leq i < l$ we have $x_{i} = q \cdot x_{i+1} = q^{l-i} x_{\ell}$. Here, $q>1$ is an parameter that 

For the purposes of analysis consider a histogram encoding of the real distribution $\mathcal{P}$ i.e.,  $H^r = \{h_1^r, h_2^r \ldots \}$ where the mesh $X^r = \{ x_1^r, x_2^r, \ldots \} = \{ \Pr[pwd] : pwd \in \mathcal{P}\}$ is defined using the exact probabilities in the distribution. Of course the mesh $X^r$ and the associated histogram $H^r$ are unknown so we cannot simply evaluate the formula above. However, it is helpful to imagine rounding each point $x_i^r$ to a value $\mathtt{Round}(x_i^r) \in X_{\ell,q}$ and defining the rounded histogram $h_i = \sum_{j: \mathtt{Round}(x_j^r) = x_i}  h_j^r$ accordingly. Supposing that $G= c + \sum_{i=1}^{idx'} h_i^r$ for some $idx'$ and $c \leq h_{idx'}^r$ we also have $G = c+ \sum_{i=1}^{idx} h_i$ for some $idx$ and $c \leq h_i$. Assuming that we always round down, i.e., $\mathtt{Round}(x_i^r) \leq x_i^r$, then we have $\lambda_G = c' \cdot x_{idx'}^r + \sum_{j=1}^{idx'-1} x_i^r h_i^r \geq c \cdot \mathtt{Round}(x_{idx'}^r) + \sum_{j=1}^{idx'-1} \mathtt{Round}(x_i^r) h_i^r  \geq  c \cdot x_{idx} + \sum_{j=1}^{idx-1} h_i x_i$.

Intuitively, to obtain our lower bound we relax all of our constraints from $\mathtt{LP1a}$ to ensure that our rounded histogram is still consistent even  still hold (whp) even if $\mathcal{P}$ does not precisely fit the mesh $X_{\ell,q}$. For example, we can replace the exact constraint that $\sum_{j=1}^\ell h_j^r \times \mathtt{Round}(x_i^r) = 1-p$ with an approximate version $\frac{1-q}{q} \leq \sum_{j=1}^\ell h_j^r \times \mathtt{Round}(x_i^r) \leq 1-p$ which the rounded histogram will still satisfy. The approach to define our linear program for the upper bound is similar with the key difference that we round up instead of down i.e., $\mathtt{Round}(x_i^r) \geq x_i^r$.  

The linear program $\mathtt{LPlower}$ to lower bound  $\lambda_G$ is shown below. We add slack terms with parameters $q,\mathbf{\epsilon_3},\mathbf{\hat{x}_{\epsilon_3}}$ to ensure that the rounded histogram is consistent with all of the constraints. The linear program $\mathtt{LPupper}$ to upper bound $\lambda_G$ is similar to $\mathtt{LPlower}$. We show it in Appendix~\ref{app:LPupper}. %in Appendix~\ref{app:FinalLPsProofs}. 

\fbox{\begin{minipage}{8cm}
\textbf{Linear Programming Task 2: \\ }$\mathtt{LPlower}(G,X_{\ell},F^S,idx,i',\mathbf{\epsilon_2},\mathbf{\epsilon_3},\mathbf{\hat{x}_{\epsilon_3}})$ \\
\textbf{Input Parameters:} $G$, $X_{\ell}=\{x_1,...,x_{\ell}\}$, $F^S=\{F^S_1,...,F^S_N\}$, $idx$, $i'$, $\mathbf{\epsilon_2}=\{\epsilon_{2,0},\ldots,\epsilon_{2,i'}\}$, $\mathbf{\hat{x}_{\epsilon_3}}=\{\hat{x}_{\epsilon_{3,0}},\ldots,\hat{x}_{\epsilon_{3,i'}}\}$, $\mathbf{\epsilon_3}=\{\epsilon_{3,i}=\frac{1}{q^{i+1}}\left(\frac{1-\hat{x}_{\epsilon_{3,i}}}{1-q\hat{x}_{\epsilon_{3,i}}}\right)^{N-i}-1, 0\leq i\leq i'\}$ \\
\textbf{Variables:} $h_1,...,h_{\ell},c, p$ \\
%\textbf{Objective:} $\min\left(b\times (\sum_{j<idx}h_j+c)\right)$ \\ 
\textbf{Objective:} $\min\left(\sum_{j<idx}h_j\times x_j +c\times x_{idx}\right)$ \\ 
%$\min\left(b\times (\sum_{j<idx}h_j\times x_j +c\times x_{idx})\right)$ \\ 
\textbf{Constraints:} 
\begin{enumerate}
    \item $\sum_{j<idx}h_j+c = G$
    %\item $\forall 0\leq i\leq i'$, $ \frac{(i+1)F^S_{i+1}}{N-i}-\epsilon_2 - \frac{i+1}{N-i} \leq \sum_{j=1}^{\added{\ell}} h_j\times x_j\times \bpdf(i,N,x_j) + p\times \bpdf(i,N,xxxx) \leq \frac{(i+1)F^S_{i+1}}{N-i}+\epsilon_2$ 
    \item $\forall 0\leq i\leq i'$:
    \begin{enumerate}
        \item for $i=0$, $\frac{1}{q^{i+1}}(\frac{(i+1)F^S_{i+1}}{N-i}-\epsilon_{2,i} - \frac{i+1}{N-i} - p) \leq \sum_{j=1}^{\ell} h_j\times x_j\times \bpdf(i,N,x_j) \leq (1+\epsilon_{3,i})(\frac{(i+1)F^S_{i+1}}{N-i}+\epsilon_{2,i} - p\times \bpdf(i,N,q x_{\ell})) + \bpdf(i,N,\hat{x}_{\epsilon_{3,i}})$
        \item for $1\leq i\leq i'$, $\frac{1}{q^{i+1}}(\frac{(i+1)F^S_{i+1}}{N-i}-\epsilon_{2,i} - \frac{i+1}{N-i} - p\times \bpdf(i,N,q x_{\ell})) \leq \sum_{j=1}^{\ell} h_j\times x_j\times \bpdf(i,N,x_j) \leq (1+\epsilon_{3,i})(\frac{(i+1)F^S_{i+1}}{N-i}+\epsilon_{2,i}) + \bpdf(i,N,\hat{x}_{\epsilon_{3,i}})$
    \end{enumerate}
    \item $ \frac{1-p}{q} \leq \sum_{j=1}^{\ell} h_j\times x_j \leq 1 - p$ %$ \sum_{j=1}^{\added{\ell}} h_j\times x_j = 1 - p$
    \item $0 \leq c \leq h_{idx}$
\end{enumerate}
(\textbf{Note:} we consider $idx=1,2,...,\ell+1$. When $idx=\ell+1$, we define $h_{\ell+1}=G$ and $x_{\ell+1} = 0$.)
\end{minipage}}

%%%%%%%%%%This is where LP3 originally located%%%%%%%%%

Theorem \ref{thm:bound3final-lower} shows that the upper/lower bounds we obtain hold with high confidence --- \fullversion{due to space limitations the formal proof is deferred to the full version \cite{fullversion}.}{see Appendix~\ref{app:FinalLPsProofs} for the complete proof.}

\begin{theorem}\label{thm:bound3final-lower}
Given an unknown password distribution $\mathcal{P}$ 
%let $S$ be a sample set containing $N$ items independently randomly sampled from $\mathcal{P}$, and let $F$ be its frequency list. 
for any $G\geq 0$, integer $i'\geq 0$, $\mathbf{\epsilon_2}=\{\epsilon_{2,0},\ldots,\epsilon_{2,i'}\}\in [0,1]^{i'+1}$, $\mathbf{\hat{x}_{\epsilon_3}}=\{\hat{x}_{\epsilon_{3,0}},\ldots,\hat{x}_{\epsilon_{3,i'}}\}$, $\mathbf{\epsilon_3}=\{\epsilon_{3,i}=\frac{1}{q^{i+1}}(\frac{1-\hat{x}_{\epsilon_{3,i}}}{1-q\hat{x}_{\epsilon_{3,i}}})^{N-i}-1, 0\leq i\leq i'\}$, we have:
\begingroup\makeatletter\def\f@size{9.5}\check@mathfonts
\begin{align*}
    &\Pr\left[ \lambda_G \geq \min\limits_{1\leq idx\leq l+1}\mathtt{LPlower}(X_{l,q},F^S,idx,G,i',\mathbf{\epsilon_2},\mathbf{\epsilon_3},\mathbf{\hat{x}_{\epsilon_3}})\right] \\ &\geq 1-\delta \\
    &\Pr\left[\lambda_G \leq \max\limits_{1\leq idx\leq l+1}\mathtt{LPupper}(X_{l,q},F^S,idx,G,i',\mathbf{\epsilon_2},\mathbf{\epsilon_3}, \mathbf{\hat{x}_{\epsilon_3}}) \right] \\
    & \geq 1-\delta
\end{align*}
\endgroup
%\begingroup\makeatletter\def\f@size{8.5}\check@mathfonts
%$$\Pr\left[ \lambda_G \geq \min\limits_{1\leq idx\leq l+1}\mathtt{LPlower}(X_{l,q},F^S,idx,G,i',\mathbf{\epsilon_2},\mathbf{\epsilon_3}, \mathbf{\hat{x}_{\epsilon_3}})\right] 
%    \geq 1-\delta $$
%    $$\Pr\left[\lambda_G \leq \max\limits_{1\leq idx\leq l+1}\mathtt{LPupper}(X_{l,q},F^S,idx,G,i',\mathbf{\epsilon_2},\mathbf{\epsilon_3}, \mathbf{\hat{x}_{\epsilon_3}}) \right] 
%    \geq 1-\delta$$
%\endgroup
where $\delta =  2\times\sum_{0\leq i\leq i'}\exp\left(\frac{-2(N-i)^2\epsilon_{2,i}^2}{N(i+1)^2}\right)$ and the randomness is taken over the sample $S \leftarrow \mathcal{P}^N$ of size $N$.

\end{theorem}

As a corollary of Theorem \ref{thm:bound3final-lower} we can also lower bound $\lambda(S,G)$ --- see \fullversion{the full version \cite{fullversion}.}{Corollary \ref{crl:bound3final-lower} in Appendix~\ref{app:lambdaSG}. }

\subsection{A Lower Bound for G $\geq$ N Derived from Existing Work}\label{sec:lowerbound-priorwork}
In this section we prove a lower bound using an idea from prior work. We use this bound as a comparison in Section~\ref{sec:empiricalanalysis} showing that our bounds perform very well.

Fixing arbitrary parameters $L\geq 1$ and $j\geq 2$ Blocki et al.~\cite{SP:BloHarZho18} proposed the formula $f(S,L)\doteq \frac{1}{N} \sum_{i: f_i^S \geq j} f_i^S - \frac{N}{(j-1)! L^{j-1}}$ as a lower bound for the expected number of passwords cracked by a rational attacker\footnote{Intuitively, a rational password cracker will continue checking passwords as long as the marginal reward $p_iv$ of checking the next password $pwd_i$ exceeds the marginal guessing cost. Assuming the cost to check each password guess is at most $1$ then as $p_iv \geq 1$ the attacker will continue guessing and check the next password $pwd_i$.} when the value of a cracked password is at least $v> NL$. In Theorem \ref{thm:lowerbound-priorwork} we prove that the same formula $f(S,L)$ can be used to derive high confidence lower bounds for $\lambda_{NL}$ (fixing $G=NL$) as below:

\newcommand{\thmpriorlb}{Given a sample set $S\leftarrow \mathcal{P}^N$ with size $N$, for any $L\geq 1, t\geq 0, 0\leq\epsilon\leq 1$, any integer $j\geq 2$, the expected percentage of passwords cracked by a perfect knowledge attacker making $G=N\times L$ guesses is bounded as below:
$$\Pr[\lambda_G \geq f(S,L)-t/N - \epsilon] \geq 1- \delta$$
where $\delta =  \exp\left(-2t^2/(Nj^2)\right) + \exp\left(-2N\epsilon^2\right)$,  $f(S,L)\doteq \frac{1}{N} \sum_{i: f_i^S \geq j} f_i^S - \frac{N}{(j-1)! L^{j-1}}$ and the randomness is taken over the sample set $S\leftarrow \mathcal{P}^N$.}

\begin{theorem}\label{thm:lowerbound-priorwork}
\thmpriorlb
\end{theorem}

We leave the proofs of Theorem \ref{thm:lowerbound-priorwork} to \fullversion{the full version~\cite{fullversion}}{Appendix \ref{app:lb1}} as the analysis closely follows Blocki et al.~\cite{SP:BloHarZho18} and McDiarmid's inequality, and is not the main focus of this paper. 

This lower bound is derived only as a comparison to the new approaches we proposed above. Empirical analysis in Section~\ref{sec:empiricalanalysis} demonstrates our new techniques (Corollary \ref{crl:bound2} and Theorems \ref{thm:bound1} and \ref{thm:bound3final-lower} ) yield superior lower bounds. 

%While the lower bound is useful in some cases we note that Theorem \ref{thm:lowerbound-priorwork} and Corollary \ref{crl:lowerbound-priorwork} will apply when the guessing number $G=NL$ is larger than the sample size $N$ since $L\geq 1$. Furthermore, empirical analysis demonstrates this lower bound is much worse than the other lower bounds we propose.
\section{Empirical Analysis}\label{sec:empiricalanalysis}

% !TEX root = main.tex

%%%%all guessing curves
\begin{figure*}
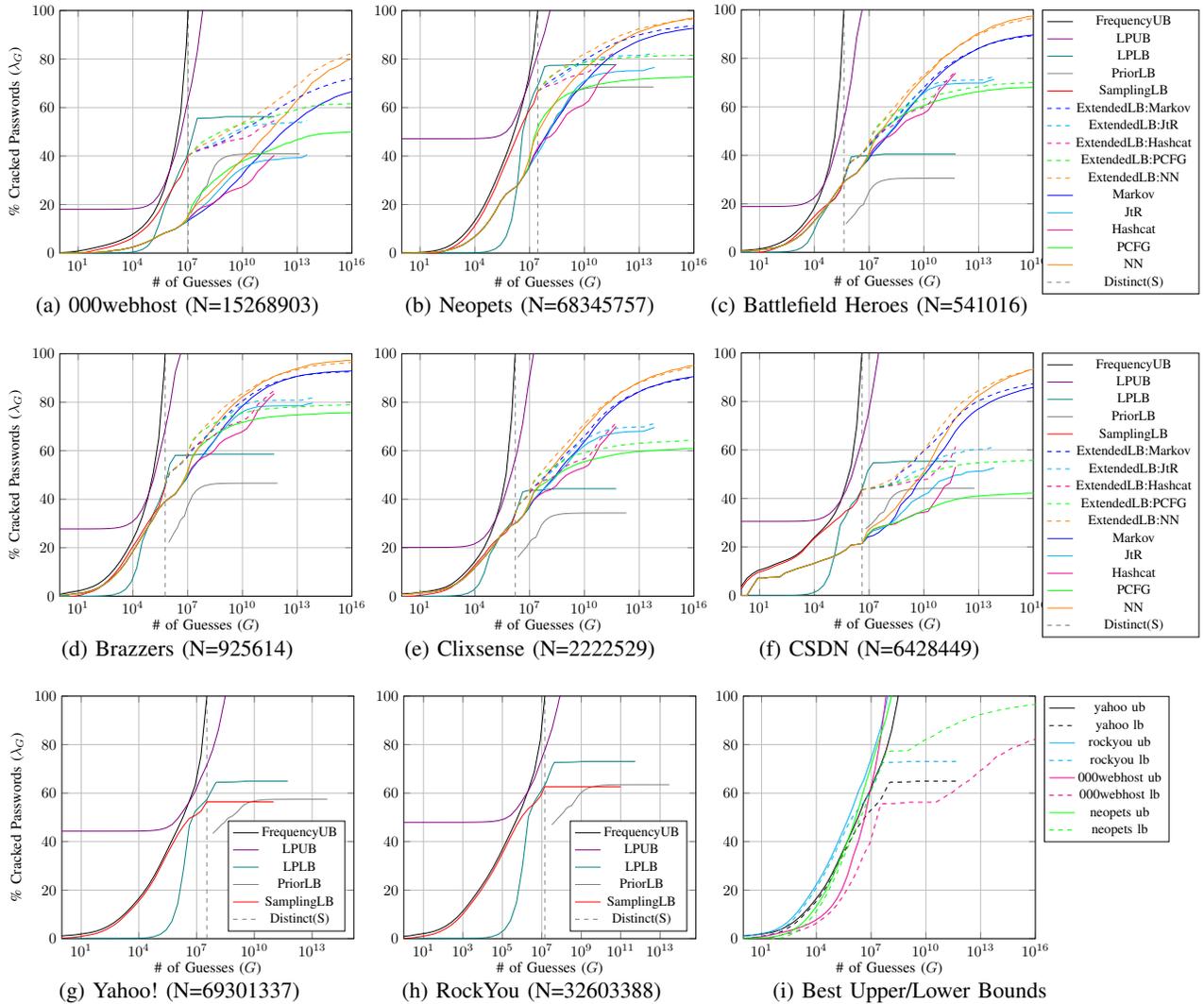
\centering
%000webhost
\input{Plots/000webhost}
\vspace{0.5em}
%\hfill
%neopets
\input{Plots/neopets}
\vspace{0.5em}
%\hfill
%\newline
%bfield
\input{Plots/bfield}
\vspace{0.5em}
%\hfill
\newline
%brazzers
\input{Plots/brazzers}
\vspace{0.5em}
%\hfill
%\newline
%clixsense
\input{Plots/clixsense}
\vspace{0.5em}
%\hfill
%csdn
\input{Plots/csdn}
\vspace{0.5em}
\newline
%Yahoo
%\begin{figure}\centering
\begin{subfigure}[b]{0.27\textwidth}
\vspace{-0.2cm}
      \begin{tikzpicture}[scale=0.6]
      \begin{axis}[
        title style={align=center},
        xlabel={\# of Guesses ($G$)},
        xmin = {1},
        xmode = log,
        log basis x={10},
        xlabel shift = -3pt,
        ylabel={\% Cracked Passwords ($\lambda_G$)},
        ymax={100},
        ymin = {0},
        %ymode = log,
        %log basis y={2},
        %ylabel shift = -3pt,
        grid=major,
        %small,
        %cycle list = {{red, mark=triangle},  {red, dashed,mark=triangle}, {blue, mark=square},{blue,dashed, mark=square},  {green, mark=diamond}, {green, dashed,mark=diamond}}
        cycle list = { {gray, dashed}, {black}, {violet}, {teal}, {gray}, {red}},
        legend style = {font=\small},%{font=\tiny},
        legend pos = south east, %north west, %outer north east,
        legend entries = {FrequencyUB, LPUB, LPLB, PriorLB, SamplingLB, Distinct(S)}
      ]
      \addlegendimage{no markers, black}
      \addlegendimage{no markers, violet}
      \addlegendimage{no markers, teal}
      \addlegendimage{no markers, gray}
      \addlegendimage{no markers, red}
      %\addlegendimage{no markers, blue}
      
      \addlegendimage{dashed, gray}
      
      %\addlegendimage{brown}

      %\addlegendimage{no markers, violet}
      %\addlegendimage{no markers, teal}

      %\addlegendimage{no markers, olive}
      
      %\addlegendimage{dashed, red}
      %\addlegendimage{dashed, blue}

       %\addlegendimage{purple}
       
      %\addlegendimage{only marks, mark=triangle*}
      %\addlegendimage{only marks, mark=square*}
      %\addlegendimage{only marks, mark=diamond*}
       % \addlegendimage{no markers, black}
        %\addlegendimage{no markers, dashed, black}
        
      %  \addlegendimage{dashed, red, mark=triangle*}
        
      %  \addlegendimage{dashed, blue, mark=square*}
        %\addlegendimage{only marks, green, mark=diamond*}
      %  \addlegendimage{dashed, green, mark=diamond*}

     %   \addlegendimage{only marks, cyan, mark=square*}
     %  \addlegendimage{only marks, yellow, mark=square*}
     %    \addlegendimage{no markers, purple, mark=square*}
      %  \addlegendimage{no markers, red, mark=square*}
         % \addlegendimage{no markers}
        % \addlegendimage{no markers, dashed}
        
    % # Distinct
    \addplot coordinates { (33895873, 0) (33895873, 100)};

    %yahoo upper bound
    %yahoo
    \addplot coordinates {(1,1.1128) (2,1.32785) (4,1.4971) (8,1.79541) (16,2.13606) (32,2.54794) (64,3.13958) (128,3.99445) (256,5.14545) (512,6.60003) (1024,8.41074) (2048,10.5658) (4096,13.039) (8192,15.7479) (16384,18.7154) (32768,22.1076) (65536,26.1203) (131072,30.6748) (262144,35.55) (524288,40.4617) (1048576,45.3936) (2097152,50.6323) (4194304,56.8075) (8388608,63.2196) (16777216,75.3242) (33554432,99.5332) (33895873,100.026) };

    %yahoo  lambda_G  myLP2  max
    \addplot coordinates { (1, 44.2954)
     (2, 44.2955)
     (4, 44.2955)
     (8, 44.2955)
     (16, 44.2955)
     (32, 44.2956)
     (64, 44.2958)
     (128, 44.2961)
     (256, 44.2968)
     (512, 44.2981)
     (1024, 44.3009)
     (2048, 44.3063)
     (4096, 44.3172)
     (8192, 44.3389)
     (16384, 44.3823)
     (32768, 44.4688)
     (65536, 44.6413)
     (131072, 44.983)
     (262144, 45.6582)
     (524288, 46.9945)
     (1.04858e+06, 49.4918)
     (2.09715e+06, 52.8483)
     (4.1943e+06, 56.4314)
     (8.38861e+06, 60.6758)
     (1.67772e+07, 66.379)
     (3.35544e+07, 71.4755)
     (6.71089e+07, 77.7174)
     (1.34218e+08, 86.0332)
     (2.68435e+08, 97.505)
     (2.88435e+08, 98.9301)
     (3.08435e+08, 100.2)
     %(5.36871e+08, 100.233)
    };

    %yahoo  lambda_G  myLP2  min
    \addplot coordinates { (1, -9.08279e-09)
     (2, -9.08279e-09)
     (4, -9.08279e-09)
     (8, 0.000109291)
     (16, 0.000218691)
     (32, 0.000437291)
     (64, 0.000874591)
     (128, 0.00174929)
     (256, 0.00349849)
     (512, 0.00699699)
     (1024, 0.013994)
     (2048, 0.027988)
     (4096, 0.0559759)
     (8192, 0.111952)
     (16384, 0.223904)
     (32768, 0.447807)
     (65536, 0.895615)
     (131072, 1.79123)
     (262144, 3.58246)
     (524288, 7.16492)
     (1.04858e+06, 14.3298)
     (2.09715e+06, 28.6597)
     (4.1943e+06, 46.6602)
     (8.38861e+06, 52.5008)
     (1.67772e+07, 54.8429)
     (3.35544e+07, 57.3192)
     (6.71089e+07, 61.7635)
     (7.71089e+07, 62.7236)
     (8.71089e+07, 63.5193)
     (9.71089e+07, 64.1023)
     (1.07109e+08, 64.4276)
     (1.17109e+08, 64.5013)
     (1.27109e+08, 64.5027)
     (1.34218e+08, 64.5037)
     (2.68435e+08, 64.5231)
     (5.36871e+08, 64.5618)
     (1.07374e+09, 64.6393)
     (2.14748e+09, 64.7942)
     (4.29497e+09, 64.9779)
     (5.49756e+11, 64.9779)
    };
    
    %lambda_G
    % delta(t) = 0.00991  delta1 = 9e-05
    %yahoo, 2 <= j <= 1000
    \addplot coordinates {(69301337,43.3599) (138602674,45.7288) (207904011,47.1949) (277205348,47.7951) (346506685,48.9201) (415808022,49.5312) (485109359,49.8997) (554410696,50.1388) (623712033,50.3028) (693013370,50.4201) (1386026740,52.4646) (2079040110,54.1312) (2772053480,54.9646) (3465066850,55.4646) (4158080220,55.7979) (4851093590,56.036) (5544106960,56.2146) (6237120330,56.3535) (6930133700,56.4646) (13860267400,56.9646) (20790401100,57.1312) (27720534800,57.2146) (34650668500,57.2646) (41580802200,57.2979) (48510935900,57.3217) (55441069600,57.3396) (62371203300,57.3535) (69301337000,57.3646) (138602674000,57.4146) (207904011000,57.4312) (277205348000,57.4396) (346506685000,57.4446) (415808022000,57.4479) (485109359000,57.4503) (554410696000,57.4521) (623712033000,57.4535) (693013370000,57.4546) (1386026740000,57.4596) (2079040110000,57.4612) (2772053480000,57.4621) (3465066850000,57.4626) (4158080220000,57.4629) (4851093590000,57.4632) (5544106960000,57.4633) (6237120330000,57.4635) (6930133700000,57.4636) (13860267400000,57.4641) (20790401100000,57.4642) (27720534800000,57.4643) (34650668500000,57.4644) (41580802200000,57.4644) (48510935900000,57.4644) (55441069600000,57.4645) (62371203300000,57.4645) };

    %lambda_Gnewlowerbound-d=25000-yahoo_plots.txt
    %1 <= G <= distinct  delta = 0.01delta1 = 9e-05  delta2 = 0.00991
    \addplot coordinates {(1,0.115353) (2,0.347353) (4,0.539353) (8,0.843353) (16,1.19535) (32,1.64735) (64,2.17935) (128,3.04735) (256,4.22735) (512,5.58735) (1024,7.45535) (2048,9.76735) (4096,12.1114) (8192,14.7434) (16384,17.6194) (32768,20.9634) (65536,24.7874) (131072,29.4634) (262144,34.1434) (524288,38.5514) (1048576,42.9114) (2097152,46.5394) (4194304,49.5874) (8388608,50.7554) (16777216,52.5554) (33554432,56.3354) (33885218,56.4114) (10^11,56.4114) };

    \end{axis}
  \end{tikzpicture}

\vspace{-.2cm}
 \caption{Yahoo! (N=69301337)} \figlab{YahooNewGuessingCurves}
\end{subfigure}
%\vspace{-.25cm}
%\end{figure}
%\hfill
%rockyou
%\begin{figure}\centering
\begin{subfigure}[b]{0.25\textwidth}
\vspace{-0.2cm}
      \begin{tikzpicture}[scale=0.6]
      \begin{axis}[
        title style={align=center},
        xlabel={\# of Guesses ($G$)},
        xmin = {1},
        xmode = log,
        log basis x={10},
        xlabel shift = -3pt,
        ymax={100},
        ymin = {0},
        %ymode = log,
        %log basis y={2},
        %ylabel shift = -3pt,
        grid=major,
        %small,
        %cycle list = {{red, mark=triangle},  {red, dashed,mark=triangle}, {blue, mark=square},{blue,dashed, mark=square},  {green, mark=diamond}, {green, dashed,mark=diamond}}
        cycle list = { {gray, dashed}, {black}, {violet}, {teal}, {gray}, {red} },% {cyan}, {magenta}, {yellow}, {teal, dashed}, {cyan, dashed}, {magenta, dashed}, {yellow, dashed}
        legend style = {font=\small},%{font=\tiny},
        legend pos = south east, %north west,  %outer north east, 
        legend entries = {FrequencyUB, LPUB, LPLB, PriorLB, SamplingLB, Distinct(S)} 
      ]
      \addlegendimage{no markers, black}
      \addlegendimage{no markers, violet}
      \addlegendimage{no markers, teal}
      \addlegendimage{no markers, gray}
      \addlegendimage{no markers, red}
      \addlegendimage{dashed, gray}
      
      %\addlegendimage{no markers, pink}
      
      %\addlegendimage{no markers, dashed, yellow}
      
      %\addlegendimage{no markers, teal}
      %\addlegendimage{no markers, cyan}
      %\addlegendimage{no markers, magenta}
      %\addlegendimage{no markers, yellow}
      
      %\addlegendimage{dashed, olive}
      
      %\addlegendimage{dashed, gray}

      %\addlegendimage{no markers, olive}
      
      %\addlegendimage{dashed, red}
      %\addlegendimage{dashed, blue}
      %\addlegendimage{dashed, orange}
      %\addlegendimage{dashed, green}
      %\addlegendimage{only marks, mark=triangle*}
      %\addlegendimage{only marks, mark=square*}
      %\addlegendimage{only marks, mark=diamond*}
       % \addlegendimage{no markers, black}
        %\addlegendimage{no markers, dashed, black}
        
      %  \addlegendimage{dashed, red, mark=triangle*}
        
      %  \addlegendimage{dashed, blue, mark=square*}
        %\addlegendimage{only marks, green, mark=diamond*}
      %  \addlegendimage{dashed, green, mark=diamond*}

     %   \addlegendimage{only marks, cyan, mark=square*}
     %  \addlegendimage{only marks, yellow, mark=square*}
     %    \addlegendimage{no markers, purple, mark=square*}
      %  \addlegendimage{no markers, red, mark=square*}
         % \addlegendimage{no markers}
        % \addlegendimage{no markers, dashed}
        
    % # Distinct
    \addplot coordinates { (14344391, 0) (14344391, 100)};

    %rockyou frequency upper bound
    %rockyou
    \addplot coordinates {(1,0.929511) (2,1.17205) (4,1.58996) (8,1.97602) (16,2.36573) (32,2.96563) (64,3.8348) (128,5.06415) (256,6.71577) (512,8.83802) (1024,11.4333) (2048,14.4886) (4096,17.8038) (8192,21.2952) (16384,25.0103) (32768,29.1282) (65536,33.7334) (131072,38.6622) (262144,43.7835) (524288,49.0365) (1048576,54.6604) (2097152,61.3613) (4194304,68.9058) (8388608,81.7704) (14344391,100.038) };

    %rockyou  lambda_G  myLP2  max
    \addplot coordinates { (1, 47.9691)
     (2, 47.9691)
     (4, 47.9691)
     (8, 47.9691)
     (16, 47.9692)
     (32, 47.9694)
     (64, 47.9697)
     (128, 47.9704)
     (256, 47.9719)
     (512, 47.9747)
     (1024, 47.9804)
     (2048, 47.9917)
     (4096, 48.0144)
     (8192, 48.0596)
     (16384, 48.15)
     (32768, 48.33)
     (65536, 48.6889)
     (131072, 49.4036)
     (262144, 50.82)
     (524288, 53.4792)
     (1.04858e+06, 57.0391)
     (2.09715e+06, 61.0338)
     (4.1943e+06, 65.844)
     (8.38861e+06, 72.5909)
     (1.67772e+07, 79.3481)
     (3.35544e+07, 86.8013)
     (6.71089e+07, 96.6657)
     (7.71089e+07, 99.0158)
     (8.71089e+07, 100.201)
     %(1.34218e+08, 100.216)
    };

    %rockyou  lambda_G  myLP2  min
    \addplot coordinates { (1, -4.55117e-08)
     (2, -4.55117e-08)
     (4, 0.000114154)
     (8, 0.000228254)
     (16, 0.000456454)
     (32, 0.000912954)
     (64, 0.00182595)
     (128, 0.00365195)
     (256, 0.00730385)
     (512, 0.0146077)
     (1024, 0.0292153)
     (2048, 0.0584307)
     (4096, 0.116861)
     (8192, 0.233722)
     (16384, 0.467445)
     (32768, 0.93489)
     (65536, 1.86978)
     (131072, 3.73956)
     (262144, 7.47913)
     (524288, 14.9437)
     (1.04858e+06, 29.8721)
     (2.09715e+06, 50.2609)
     (4.1943e+06, 56.458)
     (8.38861e+06, 60.3043)
     (1.67772e+07, 64.4719)
     (3.35544e+07, 71.0155)
     (4.35544e+07, 72.6686)
     (5.35544e+07, 72.7036)
     (6.35544e+07, 72.7067)
     (6.71089e+07, 72.7078)
     (1.34218e+08, 72.7283)
     (2.68435e+08, 72.7695)
     (5.36871e+08, 72.8518)
     (1.07374e+09, 73.0165)
     (2.14748e+09, 73.075)
     (5.49756e+11, 73.075)
    };
    
    %lambda_G
    % delta(t) = 0.00991  delta1 = 9e-05
    %rockyou, 2 <= j <= 1000
    \addplot coordinates {(32603388,46.9235) (65206776,49.6511) (97810164,51.1172) (130413552,52.2233) (163016940,53.3483) (195620328,53.9594) (228223716,54.3279) (260827104,54.567) (293430492,54.731) (326033880,54.8483) (652067760,58.4943) (978101640,60.161) (1304135520,60.9943) (1630169400,61.4943) (1956203280,61.8277) (2282237160,62.0658) (2608271040,62.2443) (2934304920,62.3832) (3260338800,62.4943) (6520677600,62.9943) (9781016400,63.161) (13041355200,63.2443) (16301694000,63.2943) (19562032800,63.3277) (22822371600,63.3515) (26082710400,63.3693) (29343049200,63.3832) (32603388000,63.3943) (65206776000,63.4443) (97810164000,63.461) (130413552000,63.4693) (163016940000,63.4743) (195620328000,63.4777) (228223716000,63.48) (260827104000,63.4818) (293430492000,63.4832) (326033880000,63.4843) (652067760000,63.4893) (978101640000,63.491) (1304135520000,63.4918) (1630169400000,63.4923) (1956203280000,63.4927) (2282237160000,63.4929) (2608271040000,63.4931) (2934304920000,63.4932) (3260338800000,63.4933) (6520677600000,63.4938) (9781016400000,63.494) (13041355200000,63.4941) (16301694000000,63.4941) (19562032800000,63.4942) (22822371600000,63.4942) (26082710400000,63.4942) (29343049200000,63.4942) };

    %lambda_Gnewlowerbound-d=25000-rockyou_plots.txt
    %1 <= G <= distinct  delta = 0.01delta1 = 9e-05  delta2 = 0.00991
    \addplot coordinates {(1,-0.0766468) (2,0.207353) (4,0.547353) (8,0.915353) (16,1.32735) (32,1.88735) (64,2.80735) (128,4.01535) (256,5.60335) (512,7.68735) (1024,10.2714) (2048,13.3794) (4096,16.5794) (8192,19.9794) (16384,23.7954) (32768,27.7274) (65536,32.2514) (131072,36.9834) (262144,41.9554) (524288,46.6634) (1048576,50.6474) (2097152,53.6474) (4194304,55.6754) (8388608,58.6274) (14335294,62.6474) (10^11, 62.6474) };

    \end{axis}
  \end{tikzpicture}

\vspace{-.2cm}
 \caption{RockYou (N=32603388)} \figlab{RockYouNewGuessingCurves}
\end{subfigure}
%\vspace{-.25cm}
%\end{figure}
%\begin{figure}\centering
\begin{subfigure}[b]{0.31\textwidth}
\vspace{-0.2cm}
      \begin{tikzpicture}[scale=0.6]
      \begin{axis}[
        title style={align=center},
        xlabel={\# of Guesses ($G$)},
        xmin = {1},
        xmax = {10^16},
        xmode = log,
        log basis x={10},
        xlabel shift = -3pt,
        ymax={100},
        ymin = {0},
        %ymode = log,
        %log basis y={2},
        %ylabel shift = -3pt,
        grid=major,
        %small,
        %cycle list = {{red, mark=triangle},  {red, dashed,mark=triangle}, {blue, mark=square},{blue,dashed, mark=square},  {green, mark=diamond}, {green, dashed,mark=diamond}}
        cycle list = { {black}, {black,dashed}, {cyan}, {cyan, dashed}, {magenta}, {magenta,dashed}, {green}, {green, dashed}},
        legend style = {font=\small},%{font=\tiny},
        legend pos = outer north east,%north west,
        legend entries = {yahoo ub, yahoo lb, rockyou ub, rockyou lb, 000webhost ub, 000webhost lb, neopets ub, neopets lb}
      ]
      \addlegendimage{no markers, black}
      \addlegendimage{dashed, black}
      %\addlegendimage{no markers, gray}
      %\addlegendimage{no markers, blue}
      \addlegendimage{no markers, cyan}
      \addlegendimage{dashed, cyan}
      %\addlegendimage{no markers, violet}
      \addlegendimage{no markers, magenta}
      \addlegendimage{dashed, magenta}
      %\addlegendimage{no markers, teal}
      \addlegendimage{no markers, green}
      \addlegendimage{dashed, green}
      %\addlegendimage{no markers, olive}
      
      %\addlegendimage{dashed, red}
      %\addlegendimage{dashed, blue}
       %\addlegendimage{purple}

    %yahoo best upper bound
    \addplot coordinates {(1,1.1128) (2,1.32785) (4,1.4971) (8,1.79541) (16,2.13606) (32,2.54794) (64,3.13958) (128,3.99445) (256,5.14545) (512,6.60003) (1024,8.41074) (2048,10.5658) (4096,13.039) (8192,15.7479) (16384,18.7154) (32768,22.1076) (65536,26.1203) (131072,30.6748) (262144,35.55) (524288,40.4617) (1048576,45.3936) (2097152,50.6323) (4.1943e+06, 56.4314)
     (8.38861e+06, 60.6758)
     (1.67772e+07, 66.379)
     (3.35544e+07, 71.4755)
     (6.71089e+07, 77.7174)
     (1.34218e+08, 86.0332)
     (2.68435e+08, 97.505)
     (2.88435e+08, 98.9301)
     (3.08435e+08, 100.2) };
    
    %yahoo best lower bound
     \addplot coordinates {(1,0.115353) (2,0.347353) (4,0.539353) (8,0.843353) (16,1.19535) (32,1.64735) (64,2.17935) (128,3.04735) (256,4.22735) (512,5.58735) (1024,7.45535) (2048,9.76735) (4096,12.1114) (8192,14.7434) (16384,17.6194) (32768,20.9634) (65536,24.7874) (131072,29.4634) (262144,34.1434) (524288,38.5514) (1048576,42.9114) (2097152,46.5394) (4194304,49.5874) (8.38861e+06, 52.5008)
     (1.67772e+07, 54.8429)
     (3.35544e+07, 57.3192)
     (6.71089e+07, 61.7635)
     (7.71089e+07, 62.7236)
     (8.71089e+07, 63.5193)
     (9.71089e+07, 64.1023)
     (1.07109e+08, 64.4276)
     (1.17109e+08, 64.5013)
     (1.27109e+08, 64.5027)
     (1.34218e+08, 64.5037)
     (2.68435e+08, 64.5231)
     (5.36871e+08, 64.5618)
     (1.07374e+09, 64.6393)
     (2.14748e+09, 64.7942)
     (4.29497e+09, 64.9779)
     (5.49756e+11, 64.9779) };

    %rockyou best upper bound
    \addplot coordinates {(1,0.929511) (2,1.17205) (4,1.58996) (8,1.97602) (16,2.36573) (32,2.96563) (64,3.8348) (128,5.06415) (256,6.71577) (512,8.83802) (1024,11.4333) (2048,14.4886) (4096,17.8038) (8192,21.2952) (16384,25.0103) (32768,29.1282) (65536,33.7334) (131072,38.6622) (262144,43.7835) (524288,49.0365) (1048576,54.6604)  (2.09715e+06, 61.0338)
     (4.1943e+06, 65.844)
     (8.38861e+06, 72.5909)
     (1.67772e+07, 79.3481)
     (3.35544e+07, 86.8013)
     (6.71089e+07, 96.6657)
     (7.71089e+07, 99.0158)
     (8.71089e+07, 100.201) };

    %rockyou best lower bound
    \addplot coordinates {(1,-0.0766468) (2,0.207353) (4,0.547353) (8,0.915353) (16,1.32735) (32,1.88735) (64,2.80735) (128,4.01535) (256,5.60335) (512,7.68735) (1024,10.2714) (2048,13.3794) (4096,16.5794) (8192,19.9794) (16384,23.7954) (32768,27.7274) (65536,32.2514) (131072,36.9834) (262144,41.9554) (524288,46.6634) (1048576,50.6474) (2097152,53.6474)  (4.1943e+06, 56.458)
     (8.38861e+06, 60.3043)
     (1.67772e+07, 64.4719)
     (3.35544e+07, 71.0155)
     (4.35544e+07, 72.6686)
     (5.35544e+07, 72.7036)
     (6.35544e+07, 72.7067)
     (6.71089e+07, 72.7078)
     (1.34218e+08, 72.7283)
     (2.68435e+08, 72.7695)
     (5.36871e+08, 72.8518)
     (1.07374e+09, 73.0165)
     (2.14748e+09, 73.075)
     (5.49756e+11, 73.075)};

    %000webhost best bound
    \addplot coordinates {(1,0.218426) (2,0.317399) (4,0.470442) (8,0.72841) (16,1.09745) (32,1.53355) (64,2.04037) (128,2.54783) (256,3.08776) (512,3.69693) (1024,4.39047) (2048,5.21538) (4096,6.19832) (8192,7.3926) (16384,8.87285) (32768,10.7385) (65536,13.1469) (131072,16.3296) (262144,20.6116) (524288,26.4025) (1048576,34.0244)
     (2.09715e+06, 40.5463)
     (4.1943e+06, 48.9989)
     (8.38861e+06, 60.0133)
     (1.67772e+07, 70.1013)
     (3.35544e+07, 82.9886)
     (6.71089e+07, 100.201) };

    %000webhost best lower bound
    \addplot coordinates {(1,-0.858182) (2,-0.810182) (4,-0.634182) (8,-0.358182) (16,0.00581763) (32,0.461818) (64,0.973818) (128,1.51382) (256,2.00182) (512,2.51782) (1024,3.19382) (2048,3.99782) (4096,4.87382) (8192,5.93382) (16384,7.47382) (32768,9.19382) (65536,11.1938) (131072,13.8378) (262144,16.9898) (524288,20.7498) (1048576,24.7098)  (2.09715e+06, 29.7143)
     (4.1943e+06, 35.0216)
     (8.38861e+06, 39.1658)
     (1.67772e+07, 47.1249)
     (3.35544e+07, 55.5526)
     (6.71089e+07, 55.642)
     (1.34218e+08, 55.686)
     (2.68435e+08, 55.7739)
     (5.36871e+08, 55.9497)
     (1.07374e+09, 56.2309)
     (2.14748e+09, 56.2309)
     (4.29497e+09, 56.2309)
     (8.58993e+09, 56.2309)
     (1.71799e+10, 56.2309)
     (3.43597e+10, 56.2309)
     (6.87195e+10,57.4434) (1.37439e+11,58.9754) (2.74878e+11,60.2594) (5.49756e+11,62.0594) (1.09951e+12,63.7874) (2.19902e+12,65.0834) (4.39805e+12,67.0314) (8.79609e+12,69.1794) (1.75922e+13,70.5514) (3.51844e+13,72.7354) (7.03687e+13,74.8914) (1.40737e+14,76.0074) (2.81475e+14,77.3474) (5.6295e+14,78.7674) (1.1259e+15,79.4714) (2.2518e+15,80.2794) (4.5036e+15,81.5314) (9.0072e+15,82.1594) (1e+16,82.2434) };

    %neopets  best upper bound
    \addplot coordinates {(1,0.0957224) (2,0.13098) (4,0.195629) (8,0.291942) (16,0.417719) (32,0.635437) (64,0.983075) (128,1.51791) (256,2.31334) (512,3.46685) (1024,5.00121) (2048,6.96575) (4096,9.4566) (8192,12.3855) (16384,15.7543) (32768,19.6828) (65536,24.2589) (131072,29.3261) (262144,34.7341) (524288,40.4954) (1048576,46.7505) (2097152,53.885) (4194304,61.8731) 
     (8.38861e+06, 69.2588)
     (1.67772e+07, 76.8607)
     (3.35544e+07, 83.8409)
     (6.71089e+07, 91.0223)
     (1.34218e+08, 100.2) };

    %neopets best lower bound
    \addplot coordinates {(1,-0.917102) (2,-0.861102) (4,-0.825102) (8,-0.717102) (16,-0.573102) (32,-0.381102) (64,-0.0611016) (128,0.370898) (256,1.0509) (512,2.0989) (1024,3.4029) (2048,5.1429) (4096,7.5389) (8192,10.2989) (16384,13.6949) (32768,17.5549) (65536,22.0629) (131072,26.7669) (262144,32.2589) (524288,37.5029) (1048576,42.8029) (2097152,48.2829) (4194304,53.5029) (8.38861e+06, 59.2737)
     (1.67772e+07, 65.7034)
     (3.35544e+07, 70.2992)
     (6.71089e+07, 76.9127)
     (1.34218e+08, 77.3942)
     (2.68435e+08, 77.4139)
     (5.36871e+08, 77.4532)
     (1.07374e+09, 77.5317)
     (2.14748e+09,79.0634) (4.29497e+09,80.4874) (8.58993e+09,81.8794) (1.71799e+10,83.0154) (3.43597e+10,84.0954) (6.87195e+10,85.6114) (1.37439e+11,86.7754) (2.74878e+11,87.6554) (5.49756e+11,88.9034) (1.09951e+12,89.8994) (2.19902e+12,90.6394) (4.39805e+12,91.5754) (8.79609e+12,92.3274) (1.75922e+13,92.8194) (3.51844e+13,93.3834) (7.03687e+13,94.0234) (1.40737e+14,94.3474) (2.81475e+14,94.8634) (5.6295e+14,95.2634) (1.1259e+15,95.6154) (2.2518e+15,95.9594) (4.5036e+15,96.3034) (9.0072e+15,96.5274) (1e+16,96.5354)};

    \end{axis}
  \end{tikzpicture}

\vspace{-.2cm}
 \caption{Best Upper/Lower Bounds} 
 \label{fig:bestguessingcurves}
\end{subfigure}
%\vspace{-.25cm}
%\end{figure}
\vspace{-0.2cm}
\caption{000webhost, Neopets, Battlefield Heroes, Brazzers, Clixsense, CSDN, Yahoo! and RockYou Guessing Curves, and Best Bounds}
\label{fig:allcurves}
\vspace{-0.4cm}
\end{figure*}

%best bounds (yahoo,rockyou,000webhost,neopets)
%\input{Plots/bestbound}

In this section we apply our techniques to upper/lower bound $\lambda_G$ to analyze several empirical password datasets. We compare these bounds to guessing curves generated by state of the art password cracking models. 

\subsection{Datasets}

We use eight empirical password datasets in our analysis: Yahoo!, RockYou, 000webhost, Neopets, Battlefield Heroes, Brazzers, Clixsense and CSDN --- for the last six datasets we use the sanitized versions prepared by Liu et al.~\cite{SP:LNGCU19}. Table~\ref{table:BasicInformation} in Appendix~\ref{app:figures} provides basic information about each dataset. $N$ represents the total size of the dataset and \textbf{\# Distinct} represents the number of distinct passwords after eliminating duplicates (i.e. $\mathbf{Distinct}(S)$). Similarly,  \textbf{\# Unique} represents the number of passwords that appear exactly once in $S$ (i.e. $\mathbf{Unique}(S)$). In our analysis we view each dataset $S$ as representing $N$ iid samples from an unknown password distribution. Good-Turing frequency estimation tells us that the total probability mass of unseen passwords is approximately $1-\sum_{pwd \in S} \Pr[pwd] \approx \frac{\mathbf{Unique}(S)}{N}$ which means for $G=\mathbf{Distinct}(S)$ we have $\lambda_G \geq \sum_{pwd \in S} \Pr[pwd] \approx \frac{N-\mathbf{Unique}(S)}{N}$. Thus, for Yahoo! (resp. CSDN) we should have  $\lambda_G \geq 0.575$ (resp. $\lambda_G \geq 0.443$). 

One of the password datasets (Yahoo!~\cite{SP:Bonneau12,NDSS:BloDatBon16}) is actually a differentially private frequency list and does not include plaintext passwords. 
For this dataset we can still apply our techniques to upper/lower bound $\lambda_G$, but we cannot compare our bounds with empirical password cracking models. While the dataset was slightly perturbed to satisfy differential privacy, Blocki et al.~\cite{NDSS:BloDatBon16} showed that the L1 distortion is minimal i.e., the additional error term is $O(1/\sqrt{N})$.

{\bf The Independent Samples Assumption: }In our analysis we assume that the dataset $S$ was sampled iid from some unknown distribution $\mathcal{P}$. As we noted previously our linear program allows us to reject datasets $S$ which are clearly inconsistent with our iid assumption. In particular, if the linear program is infeasible this indicates that there is no password distributions $\mathcal{P}$ consistent with the dataset $S$. This might occur if a large fraction of the dataset was duplicated\footnote{When $s_1,s_2 \leftarrow \mathcal{P}$ are sampled independently it is always possible that $s_1 = s_2$. By contrast, if we sample $s_1 \leftarrow \mathcal{P}$ and then simply fix $s_2=s_1$ without re-sampling then we would say that the record $s_2$ was duplicated e.g., a user registers for an account with password $s_1$ and later registers for a second account with the same password.}
For example, the description of the LinkedIn frequency corpus~\cite{bicycleAttacks} indicates that that  the dataset contains $177,500,189$ passwords, but only $164,590,819$ unique e-mail addresses \cite{linkedinrevisited}. This means that there are more passwords than unique users and it is possible that many of the entries in the LinkedIn frequency corpus are duplicates. We confirmed that the LinkedIn frequency corpus is not iid using our linear program. The linear program was infeasible indicating that we can reject the independent samples hypothesis for this dataset. By contrast, with every other dataset our linear program found feasible solutions. While we cannot absolutely guarantee that our LP rejects every dataset $S$ which is not iid these results increase our confidence that this assumption is a reasonable approximation of reality.   

{\bf Ethical Considerations. } Many of the password datasets we analyze contain stolen passwords that were subsequently leaked on the internet and the usage of this data raises important ethical considerations. We did not crack any new passwords as part of our analysis and the breached datasets are already publicly available. Thus, our usage of the data does not pose any additional risk to users.

\subsection{Password Cracking Models}

We use the empirical attack results in Liu et al.~\cite{SP:LNGCU19} to compare with our bounds for 000webhost, Neopets, Battlefield Heroes, Brazzers, Clixsense and CSDN datasets and to generate an extended lower bound using Corollary~\ref{crl:bound2}. Liu et al.~\cite{SP:LNGCU19} evaluate 10 password cracking models on 
the six datasets based on Markov model, PCFG, neural networks, Hashcat, and JtR techniques. In our analysis, we focus on the best performing ones of each password cracking technique: the neural network model (denote as \textbf{NN}) in Melicher et al.~\cite{USENIX:MUSKBCC16}, the Markov 4-gram model (denote as \textbf{Markov}) in Dell’Amico and Filippone~\cite{CCS:DelFil15}, the original PCFG model~\cite{SP:WAMG09} (denote as \textbf{PCFG}), the extended JtR~\cite{SP:LNGCU19} (denote as \textbf{JtR}), and Hashcat (denote as \textbf{Hashcat}) implemented in Liu et al.~\cite{SP:LNGCU19}. Some other models (e.g. Markov backoff model~\cite{CCS:DelFil15} and the PCFG model in Komanduri ~\cite{komanduri2016modeling}) in Liu et al.~\cite{SP:LNGCU19} may outperform \textbf{Markov} or \textbf{PCFG} in some ranges of $G$, but for every value of $G$ their performance is worse than at least one of the models we selected.

\subsection{Evaluation}

For each password dataset $S$ we generate upper/lower bounds on $\lambda_G$ using our results from Section \ref{sec:theoreticalanalysis} and compare the upper/lower bounds with the guessing curves derived from password cracking models. Our upper/lower bounds and empirical attack results for 000webhost, Neopets, Battlefield Heros, Brazzers, Clixsense, and CSDN, Yahoo!, and RockYou are shown in Figures~\ref{fig:allcurves}(a)-(h). The vertical dashed gray line in each subfigure shows $\mathbf{Distinct}(S)$ of the dataset $S$.

In these figures we use $\mathtt{FrequencyUB}(S,G)$ (resp.  $\mathtt{LPUB}(S,G)$)  to denote the upper bound obtained from Corollary~\ref{crl:upperbound-priorwork} (resp. Theorem~\ref{thm:bound3final-lower}) with password dataset $S$.  Similarly, we use  $\mathtt{SamplingLB}(S,G)$ (resp. $\mathtt{LPLB}(S,G)$) to denote the lower bound obtained by applying Theorem~\ref{thm:bound1} (resp. Theorem~\ref{thm:bound3final-lower}).  For comparison we also plot $\mathtt{PriorLB(S,G,j)}$ which denotes the lower bound from Theorem~\ref{thm:lowerbound-priorwork} based on results of Blocki et al. \cite{SP:BloHarZho18} --- specifically we set $\mathtt{PriorLB}(S,G) = \max\limits_{2\leq j\leq 1000} \mathtt{PriorLB(S,G,j)}$ where $\mathtt{PriorLB(S,G,j)}$ is the lower bound we obtain after fixing the parameter $j$ in Theorem~\ref{thm:lowerbound-priorwork}. Two of the lower bounds $\mathtt{LBUB}$ and $\mathtt{LPLB}$ require us to solve linear programs as a subroutine. We used Gurobi 9.0.1~\cite{gurobi} as our linear programming solver. 

For each of the upper/lower bounds there is an error term $\delta$ which upper bounds the probability that our bound is wrong. The error term $\delta$ will depend on our choice of parameters. For example, in Theorem~\ref{thm:lowerbound-priorwork} (resp. Theorem~\ref{thm:bound1} ) the parameters $t,\epsilon$ (resp. $t, d$) determines $\delta=\exp(-2t^2/(Nj^2)) + \exp(2N\epsilon^2)$ (resp. $\delta=\exp(-2t^2/d)$). Briefly, we always select parameters such that $\delta \leq 0.01$. For example, to generate $\mathtt{SamplingLB}(S,G)$ we set $d=2.5\times 10^4$ and $t \geq \sqrt{d \ln(1/\delta)/2}$ in Theorem~\ref{thm:bound1} i.e., we randomly partition $S$ into $D_1 \in \mathcal{P}^{N-d}$ and $D_2 \in \mathcal{P}^{d}$ and return our lower bound $(h(D_1,D_2,G)-t)/d$ where $h(D_1,D_2,G)$ counts the number of passwords in $D_2$ that are top $G$ passwords in $D_1$. See Appendix~\ref{app:parametersetting} for concrete details on how we specify all of these relevant parameters for other upper/lower bounds. 

We compare our upper/lower bounds with empirical password guessing curves derived from state of the art password cracking models. Specifically, we compare with the guessing curves generated by Liu et al.~\cite{SP:LNGCU19}. In particular, for each password dataset $S$ they first subsample $25,000$ passwords to obtain a smaller dataset $D_2$. Then for each password $pwd \in D_2$ and password cracking model $M$ they compute a guessing number $g_{M,pwd}$ for that password (often using Monte Carlo strength estimation~\cite{CCS:DelFil15}). Finally, for a guessing bound $G$ we can estimate that the model will crack $\tilde{\lambda}_{G,M} = \left|\{pwd \in D_2 : g_{M,pwd} \leq G \right|/|D_2|$ passwords. We can also apply Corollary~\ref{crl:bound2} to derive a new lower bound on $\lambda_G$ by combining the results from our model $M$ with our sampling based lower bound   $\mathtt{SamplingLB}(S,G)$ --- we use $\mathtt{ExtendedLB}(S,G,M)$ to denote the extended lower bound and highlight these lower bounds using dotted lines in Figure \ref{fig:allcurves}. The guessing curves generated by Liu et al.~\cite{SP:LNGCU19} were generated with a subsample of size $d=25000$. We used the same value of $d$ when applying the lower bounds $\mathtt{SamplingLB}(S,G)$ and $\mathtt{ExtendedLB}(S,G,M)$ as this is already sufficient to achieve error bound $\delta \leq 0.01$.
\vspace{-0.15cm}
\subsection{Discussion}\label{subsec:empirical-discussion}
\vspace{-0.1cm}
Figure~\ref{fig:allcurves} shows that when $G$ is small our upper bound $\mathtt{FrequencyUB}(S,G)$ and lower bound $\mathtt{SamplingLB}(S,G)$ are very close to each other. For example, when $G\leq 262144$ (resp. $G \leq 1.048576\times 10^6$) the difference between the upper bound $\mathtt{FrequencyUB}$ and the lower bound $\mathtt{SamplingLB}$ for Yahoo! dataset is smaller than $1.407\%$ (resp. $2.483\%$). As long as $\mathtt{FrequencyUB}(S,G)- \mathtt{SamplingLB}(S,G)$ is small the empirical distribution $\hat{\lambda}_G$ will give us an accurate approximation of the guessing curve $\lambda_G$ from the real (unknown) password distribution. However, as $G$ becomes large the gap $\mathtt{FrequencyUB}(S,G)- \mathtt{SamplingLB}(S,G)$ begins to widen e.g., for the Yahoo! dataset when $G=1.6777216\times 10^7$ we have $\mathtt{FrequencyUB}(S,G)- \mathtt{SamplingLB}(S,G) = 22.769\%$. While $\mathtt{FrequencyUB}(S,G)$ and lower bound $\mathtt{SamplingLB}(S,G)$ give the best upper/lower bounds for smaller $G$ we can see that the bounds reach plateau once $G \geq \mathbf{Distinct}(S)$ e.g., for the Yahoo! dataset $\mathtt{SamplingLB}(S,G)$ and $\mathtt{FrequencyUB}(S,G)$  remain constant for all $G \geq 3.3885218\times 10^7$. Once $G \geq \mathbf{Distinct}(S)$ we need new ideas upon the lower bound $\lambda_G \geq \frac{N- \mathbf{Unique}(S)}{N}$ or the upper bound $\lambda_G \leq 1$.  

Our linear programming bounds $\mathtt{LPUB}$ and $\mathtt{LPLB}$ push past the $G \leq \mathbf{Distinct}(S)$ barrier and allow us to obtain tighter upper and lower bounds even when  $G \geq \mathbf{Distinct}(S)$. The linear programming bounds ($\mathtt{LPUB}$ and $\mathtt{LPLB}$) are worse when $G$ is small e.g.,  for Yahoo! dataset when $G=262144$ $\mathtt{FrequencyLB}$ and $\mathtt{SamplingLB}$ tightly bound $\lambda_G$ as $34.14\% \leq \lambda_G \leq 35.55\%$ while the LP bounds are $3.58\% \leq \lambda_G \leq 45.66\%$. However, as $G$ increases we find that the linear programming approach yields significantly tighter bounds e.g., for the Yahoo! dataset when $G=6.7108864\times 10^7$ our LP bounds show that $61.76\% \leq \lambda_G \leq 77.72\%$ while our frequency and sampling based bounds $56.41 \% \leq \lambda_G \leq 100\%$ are much less tight. In such cases when $\hat{\lambda}_G>\mathtt{LPUB}(S,G)$ the empirical distribution {\em should not} be used to estimate the real guessing curve.

Similar to $\mathtt{SamplingLB}(S,G)$ our linear programming lower bound $\mathtt{LPLB}(S,G)$ also plateaus when $G$ is sufficiently large, but it plateaus at a higher value e.g., $64.97\%$ for Yahoo! dataset. We also note that both of our lower bounds $\mathtt{SamplingLB}$ and $\mathtt{LPLB}$ dramatically outperform the lower bound $\mathtt{PriorLB}$ based on prior work of Blocki et al.~\cite{SP:BloHarZho18}.

{\bf Are Password Cracking Models Guess Efficient?} Figure~\ref{fig:allcurves} demonstrates empirical cracking models are often much less guess efficient than an attacker who knows the distribution. In particular, if $\tilde{\lambda}_{G,M}$ denotes the percentage of passwords cracked by model $M$ withing $G$ guesses and $\tilde{\lambda}_{G,M} < \max\{ \mathtt{SamplingLB}(S,G), \mathtt{LPLB}(S,G)\}$ then we can be confident that a perfect knowledge attacker would crack more passwords. For example, Figure~\ref{fig:000webhostguessingcurves} (resp. Figure \ref{fig:brazzersguessingcurves}) shows that an attacker making $G=8390551$ (resp. $G=2097152$)  guesses would crack {\em at most} $14\%$ (resp. $42.14\%$) of 000webhost (resp. Brazzers) passwords using any password cracking model. By contrast, our lower bounds indicate that an attacker who knows the password distribution will crack {\em at least} $39.16\%$ (resp. $58.05\%$) of 000webhost (resp. Brazzers) passwords. These results indicate that even state of the art password cracking models can be improved substantially. One policy implication of this finding is that higher levels of key stretching may be necessary to protect hashed passwords against offline brute-force attacks. 

% NIST guidelines for password hashing requires the use of moderately expensive password hashing functions such as PBKDF2 or even memory hard functions such as scrypt~\cite{percival2009stronger}, Argon2~\cite{biryukov2016argon2} or DRSample~\cite{CCS:AlwBloHAr17}. As guessing costs increase attackers will have additional incentives to develop password cracking models that are guess efficient.  

{\bf Reducing the Uncertain Region:} The lower bound $\mathtt{ExtendedLB}$ extends our sampling based approach with empirical password cracking models. The lower bound eventually improves on $\mathtt{SamplingLB}$ and $\mathtt{LPLB}$ for sufficiently large $G$. For example, for Brazzers dataset when $G=10^{16}$ our model based lower bound using neural network attack results ($\mathtt{ExtendedLB:NN}$) shows $\lambda_G\geq 96.34\%$ while the best of our other lower bounds is only $58.62\%$.

Figure~\ref{fig:bestguessingcurves} plots the best upper/lower bounds (denoted as ub/lb in the figure) for the Yahoo!, RockYou, 000webhost, and Neopets datasets --- to avoid overcrowding we plot the best upper/lower bounds for Battlefield Heroes, Brazzers, Clixsense, and CSDN datasets in Appendix~\ref{app:figures}. In particular, we plot $\mathbb{UB}(S,G) = \min\{ \mathtt{LPUB}(S,G),$ $\mathtt{FrequencyUB}(S,G)\}$ (solid curves) and $\mathbb{LB}(S,G)=\max\{\mathtt{LPLB}(S,G),$  $\mathtt{SamplingLB}(S,G),$ $\mathtt{ExtendedLB}\mathtt{:NN}(S,G) \}$ (dotted curves). Notice that the lower bound appears to initially plateau before it starts to increase again e.g., for 000webhost dataset the lower bound plateaus at $55.55\%$ when $G=3.35544\times 10^7$, barely increases when $3.35544\times 10^7\leq G \leq 3.43597\times 10^{10}$ and then starts to significantly increase again when $G \geq 3.43597\times 10^{10}$. The initial plateau point is occurs when the lower $\mathtt{LPLB}(S,G)$ plateaus and once the empirical guessing curves ``catch up" the curve starts to increase again. Each point on the best bounds in Figure~\ref{fig:bestguessingcurves} hold with probability no less than 0.98 (the union bound of the probabilities that each bound holds).

The real (unknown) guessing curve $\lambda_G$ lies somewhere in between $\mathbb{LB}(S,G)$ and $\mathbb{UB}(S,G)$. While our work substantially tightens the gap $\mathbb{UB}(S,G)-\mathbb{LB}(S,G)$, there is still an uncertain region between the solid/dotted curves. We conjecture that improved password cracking models may be able to tighten this gap. 

{\bf The Impact of Sample Size:} We use Yahoo! dataset ($N=69301337$) as an example to analyze the impact of the sample size on the quality of our lower/upper bounds. We generate four subsampled Yahoo! datasets with sample size $N=10^4, 10^5, 10^6, 10^7$ respectively and generate the best upper/lower bounds for each subsampled dataset (as we did for the original Yahoo! dataset in Figure~\ref{fig:bestguessingcurves}). We use the same parameter settings as we use for the original Yahoo! dataset, except that for $N=10^4$ we set $d=2500$ instead of $25000$ when generating $\mathtt{SamplingLB}(S,G)$ using Theorem~\ref{thm:bound1}. We plot the best upper and lower bounds for different sample sizes in Figure~\ref{fig:samplesizeimpact}. As expected the upper/lower bounds improve as the sample size increases and when $N=10^7$ the bounds are reasonably close to those obtained from the original dataset. On the negative side when the sample size is just $N=10^4$ the upper/lower bounds diverge rapidly. Thus, for smaller password datasets such as those collected from a user study our upper/lower bounds may not be particularly useful. This may justify the continued use of password cracking models as a heuristic when analyzing smaller password datasets though we still need to be cautious when drawing conclusions as state-of-the-art password cracking models dramatically under-perform in comparison to an attacker who knows the password distribution. Developing statistical techniques to rigorous compare two password distributions with only a few samples remains an important open research challenge.

\begin{figure}\centering
%\begin{subfigure}[b]{0.3\textwidth}
%\vspace{-0.2cm}
      \begin{tikzpicture}[scale=0.6]
      \begin{axis}[
        title style={align=center},
        xlabel={\# of Guesses ($G$)},
        xmin = {1},
        xmode = log,
        log basis x={10},
        xlabel shift = -3pt,
        ylabel={\% Cracked Passwords ($\lambda_G$)},
        ymax={100},
        ymin = {0},
        %ymode = log,
        %log basis y={2},
        %ylabel shift = -3pt,
        grid=major,
        %small,
        %cycle list = {{red, mark=triangle},  {red, dashed,mark=triangle}, {blue, mark=square},{blue,dashed, mark=square},  {green, mark=diamond}, {green, dashed,mark=diamond}}
        cycle list = { {black}, {black,dashed}, {cyan}, {cyan, dashed}, {magenta}, {magenta,dashed}, {green}, {green, dashed}, {orange}, {orange,dashed} },
        legend style = {font=\small},%{font=\tiny},
        legend pos = outer north east,%north west,
        legend entries = {$N=69301337$ ub, $N=69301337$ lb, $N=10^7$ ub, $N=10^7$ lb, $N=10^6$ ub, $N=10^6$ lb, $N=10^5$ ub, $N=10^5$ lb, $N=10^4$ ub, $N=10^4$ lb }
      ]
      \addlegendimage{no markers, black}
      \addlegendimage{dashed, black}
      %\addlegendimage{no markers, gray}
      %\addlegendimage{no markers, blue}
      \addlegendimage{no markers, cyan}
      \addlegendimage{dashed, cyan}
      %\addlegendimage{no markers, violet}
      \addlegendimage{no markers, magenta}
      \addlegendimage{dashed, magenta}
      %\addlegendimage{no markers, teal}
      \addlegendimage{no markers, green}
      \addlegendimage{dashed, green}
      %\addlegendimage{no markers, olive}
      
      \addlegendimage{no markers, brown}
      \addlegendimage{dashed, brown}
      \addlegendimage{no markers, blue}
      \addlegendimage{dashed, blue}
      \addlegendimage{no markers, red}
      \addlegendimage{dashed, red}
      \addlegendimage{no markers, orange}
      \addlegendimage{dashed, orange}

    \addplot coordinates {(1,1.1128) (2,1.32785) (4,1.4971) (8,1.79541) (16,2.13606) (32,2.54794) (64,3.13958) (128,3.99445) (256,5.14545) (512,6.60003) (1024,8.41074) (2048,10.5658) (4096,13.039) (8192,15.7479) (16384,18.7154) (32768,22.1076) (65536,26.1203) (131072,30.6748) (262144,35.55) (524288,40.4617) (1048576,45.3936) (2097152,50.6323) (4.1943e+06, 56.4314)
     (8.38861e+06, 60.6758)
     (1.67772e+07, 66.379)
     (3.35544e+07, 71.4755)
     (6.71089e+07, 77.7174)
     (1.34218e+08, 86.0332)
     (2.68435e+08, 97.505)
     (2.88435e+08, 98.9301)
     (3.08435e+08, 100.2) };

    %yahoo best lower bound N=69301337
     \addplot coordinates {(1,0.115353) (2,0.347353) (4,0.539353) (8,0.843353) (16,1.19535) (32,1.64735) (64,2.17935) (128,3.04735) (256,4.22735) (512,5.58735) (1024,7.45535) (2048,9.76735) (4096,12.1114) (8192,14.7434) (16384,17.6194) (32768,20.9634) (65536,24.7874) (131072,29.4634) (262144,34.1434) (524288,38.5514) (1048576,42.9114) (2097152,46.5394) (4194304,49.5874) (8.38861e+06, 52.5008)
     (1.67772e+07, 54.8429)
     (3.35544e+07, 57.3192)
     (6.71089e+07, 61.7635)
     (7.71089e+07, 62.7236)
     (8.71089e+07, 63.5193)
     (9.71089e+07, 64.1023)
     (1.07109e+08, 64.4276)
     (1.17109e+08, 64.5013)
     (1.27109e+08, 64.5027)
     (1.34218e+08, 64.5037)
     (2.68435e+08, 64.5231)
     (5.36871e+08, 64.5618)
     (1.07374e+09, 64.6393)
     (2.14748e+09, 64.7942)
     (4.29497e+09, 64.9779)
     (5.49756e+11, 64.9779) };
     
     %yahoo N=10^7 upper bound
    %yahoo10000000
    \addplot coordinates {(1,1.15179) (2,1.36632) (4,1.53646) (8,1.83324) (16,2.17147) (32,2.58228) (64,3.17296) (128,4.02557) (256,5.17753) (512,6.63454) (1024,8.44413) (2048,10.6037) (4096,13.0899) (8192,15.8206) (16384,18.8374) (32768,22.334) (65536,26.4984) (131072,31.3334) (262144,36.6514) (524288,42.6427) (1.04858e+06, 48.0255)
 (2.09715e+06, 53.338)
 (4.1943e+06, 58.0962)
 (8.38861e+06, 64.1026)
 (1.67772e+07, 72.105)
 (3.35544e+07, 83.1578)
 (6.71089e+07, 98.6216)
 (7.71089e+07, 100.206) };

    %yahoo N=10^7 lower bound
    \addplot coordinates {(1,0.0953532) (2,0.331353) (4,0.487353) (8,0.807353) (16,1.13135) (32,1.51135) (64,2.07935) (128,2.95935) (256,4.17935) (512,5.54335) (1024,7.37135) (2048,9.47535) (4096,11.8514) (8192,14.4794) (16384,17.2594) (32768,20.6154) (65536,24.2994) (131072,28.2834) (262144,32.5314) (524288,35.9914) (1.04858e+06, 39.3887)
 (2.09715e+06, 41.6089)
 (4.1943e+06, 43.2881)
 (8.38861e+06, 46.5087)
 (1.67772e+07, 50.0055)
 (3.35544e+07, 50.0438)
 (6.71089e+07, 50.0774)
 (1.34218e+08, 50.1445)
 (2.68435e+08, 50.2787)
 (5.36871e+08, 50.5471)
 (1.07374e+09, 50.7045)
 (5.49756e+11, 50.7045) };

    %yahoo N=10^6 upper bound
    %yahoo1000000
\addplot coordinates {(1,1.29432) (2,1.51292) (4,1.68432) (8,1.97432) (16,2.31532) (32,2.74312) (64,3.34152) (128,4.19942) (256,5.36992) (512,6.83872) (1024,8.66962) (2048,10.8633) (4096,13.414) (8192,16.3776) (16384,19.9093) (32768,24.2258) (65536,30.3387) (131072, 35.3422)
 (262144, 40.2296)
 (524288, 45.832)
 (1.04858e+06, 53.3938)
 (2.09715e+06, 63.7871)
 (4.1943e+06, 78.2089)
 (8.38861e+06, 98.3266)
 (9.38861e+06, 100.205)
 };

    %yahoo N=10^6 lower bound
    \addplot coordinates {(1,0.0473532) (2,0.279353) (4,0.443353) (8,0.719353) (16,1.09935) (32,1.50335) (64,2.03935) (128,2.92335) (256,3.98735) (512,5.39535) (1024,7.13935) (2048,9.23535) (4096,11.4994) (8192,14.0274) (16384,16.3714) (32768,18.7794) (65536,21.5154) (131072, 23.3405)
 (262144, 24.4326)
 (524288, 26.6146)
 (1.04858e+06, 30.6417)
 (2.09715e+06, 33.6203)
 (4.1943e+06, 33.6413)
 (8.38861e+06, 33.6832)
 (1.67772e+07, 33.7671)
 (3.35544e+07, 33.9349)
 (6.71089e+07, 34.2704)
 (1.34218e+08, 34.5214)
 (5.49756e+11, 34.5214) };
 
    %yahoo N=10^5 upper bound
    %yahoo100000
\addplot coordinates {(1,1.74848) (2,1.98548) (4,2.17348) (8,2.46548) (16,2.86048) (32,3.29348) (64,3.93548) (128,4.83648) (256,6.06948) (512,7.70448) (1024,9.83848) (2048,12.5015) (4096,16.5975) (8192,20.7645) 
(16384, 24.0423)
 (32768, 28.7672)
 (65536, 34.9691)
 (131072, 43.5638)
 (262144, 55.63)
 (524288, 72.6451)
 (1.04858e+06, 96.514)
 (1.14858e+06, 100.201)
 };

    %yahoo N=10^5 lower bound
    \addplot coordinates {(1,0.0193532) (2,0.219353) (4,0.375353) (8,0.707353) (16,1.04735) (32,1.43535) (64,2.01135) (128,2.74735) (256,3.77135) (512,4.93935) (1024,6.51135) (2048,7.76335) (4096,8.73135) (8192,9.08735) (16384,9.60335) (32768,10.8874) (65536,13.4194) (66895,13.5514) (131072, 15.5616)
 (262144, 15.8722)
 (524288, 15.8984)
 (1.04858e+06, 15.9508)
 (2.09715e+06, 16.0557)
 (4.1943e+06, 16.2654)
 (8.38861e+06, 16.6848)
 (1.67772e+07, 17.014)
 (5.49756e+11, 17.014) };
 
    %yahoo N=10^4 upper bound
    %yahoo10000
\addplot coordinates {(1,3.21821) (2,3.38821) (4,3.58821) (8,3.88821) (16,4.32821) (32,4.89821) (64,5.72821) (128,7.00821) (256,8.88821) (512,11.4482) (1024, 13.3608)
 (2048, 17.0069)
 (4096, 22.4688)
 (8192, 29.645)
 (16384, 39.6305)
 (32768, 53.6672)
 (65536, 73.4711)
 (131072, 100.204) };

    %yahoo N=10^4 lower bound
    %yahoo10000  lambda_G  myLP2  min
\addplot coordinates { (1, 0.000245376)
 (2, 0.000495376)
 (4, 0.000995376)
 (8, 0.00198538)
 (16, 0.00396538)
 (32, 0.00792538)
 (64, 0.0158554)
 (128,0.482168) (256,0.522168) (512,0.682168) (1024,0.922168) (2048,1.20217) (4096,1.92217) (7244,3.00217) (524288, 3.00217)
 (1.04858e+06, 3.47866)
 (2.09715e+06, 3.78471)
 (5.49756e+11, 3.78471)
};
 
    \end{axis}
  \end{tikzpicture}

\vspace{-.2cm}
 \caption{The Impact of Sample Size on Yahoo! Dataset} 
 \label{fig:samplesizeimpact}
%\end{subfigure}
\vspace{-.45cm}
\end{figure}
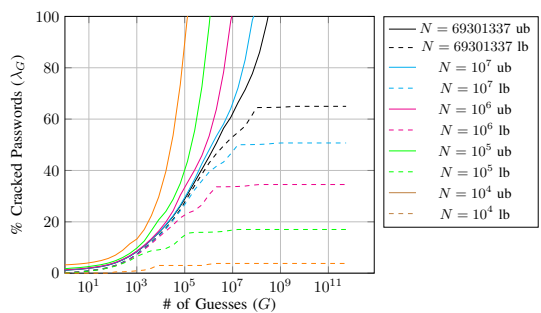

\textbf{Attacker with Partial Knowledge of the Password Distribution:} The results in Figure~\ref{fig:allcurves} indicate that a perfect-knowledge attacker will often substantially outperform state-of-the-art password models. However, the password cracking models do not require perfect knowledge of the password distribution and it is possible that our lower-bounds on $\lambda_G$ might overestimate the performance of a real-world attacker. As discussed previously we can also view the lower bound from Theorem \ref{thm:bound1} as lower bounding the performance of an attacker with partial knowledge of the distribution. Supposing that the attacker has obtained $k$ independent samples $D_1 = \{s_1,\ldots,s_k\} \leftarrow \mathcal{P}$ (e.g., $s_1,\ldots,s_k$ might represent passwords that the attacker has already cracked) from our unknown distribution $\mathcal{P}$ the attacker can build a dictionary $T(D_1,G)$ containing the top $G$ passwords from $D_1$ and use this dictionary to help crack any remaining passwords. In  Figure~\ref{fig:neopetsattackerknowledge} we evaluate the performance of this partial knowledge attacker for $k \in \{10^4, 10^5, 10^6, 10^7, N-d\}$ where $d=25,000$ denotes the size of our holdout test set $D_2$. Figure~\ref{fig:neopetsattackerknowledge} plots the results for the Neopets dataset with similar plots for 000webhost, Yahoo! and RockYou are deferred to \fullversion{the full version \cite{fullversion} of the paper}{Appendix~\ref{app:figures} Figure~\ref{fig:attackerknowledge-app}}. The figure shows that as the attacker obtains (cracks) more password samples ($k$) the performance of the hybrid (partial knowledge) attacker continues to improve and gradually approaches our theoretical lower bounds e.g., fixing the guessing number $G=524288$ a hybrid attacker building a dictionary from $10^7$ samples will crack $10.696\%$ more user passwords than an attacker with $10^6$ samples. For comparison, we also plot the guessing curve derived using the minimum guessing number (\textbf{min-guess number}) heuristic. Intuitively, we heuristically assume that each particular password will be cracked within $G$ guesses if the password was cracked by {\em any} of the five password cracking models (see Figure~\ref{fig:neopetsguessingcurves})  within $G$ guesses. The minimum guessing number was proposed as a heuristic to model the performance of a real-world attacker \cite{USENIX:USBCCKKMMS15} who may have more sophisticated dictionaries and rule-lists than publicly available models. Our experimental results indicate that our hybrid partial knowledge attacker will quickly start to outperform the minimum guessing number heuristic and to approach our lower bound from Theorem \ref{thm:bound1}.

\begin{figure}\centering
%\begin{subfigure}[htb]{0.45\textwidth}%{0.25\textwidth}
%\vspace{-0.2cm}
      \begin{tikzpicture}[scale=0.6]
      \begin{axis}[
        title style={align=center},
        xlabel={\# of Guesses ($G$)},
        xmin = {1},
        xmax = {10^9},
        xmode = log,
        log basis x={10},
        xlabel shift = -3pt,
        ylabel={\% Cracked Passwords},
        ymax={100},
        ymin = {0},
        %ymode = log,
        %log basis y={2},
        %ylabel shift = -3pt,
        grid=major,
        %small,
        cycle list = { 
        {black}, {black,dashed}, 
        {cyan}, %{cyan, dashed}, 
        {magenta}, %{magenta,dashed}, 
        {green}, %{green, dashed}, 
        {orange}, %{orange, dashed}, 
        %{olive}, {olive, dashed}, 
        %{violet}, {violet, dashed}, 
        %{yellow}, {yellow, dashed}, 
        {red}, %{red, dashed},
        {blue}, {blue, dashed},
        {teal}, {teal, dashed}
        },
        legend style = {font=\small}, %{font=\fontsize{3}{3}\selectfont},
        %{font=\tiny},
        legend pos = north west,
        legend entries = { 
        best ub, best lb, 
        k=$10^4$,
        k=$10^5$,
        k=$10^6$,
        k=$10^7$,
        k=$68320757$,
        min-guess number,
        }
      ]
      \addlegendimage{no markers, black}
      \addlegendimage{dashed, black}
      \addlegendimage{no markers, cyan}
      %\addlegendimage{dashed, cyan}
      \addlegendimage{no markers, magenta}
      %\addlegendimage{dashed, magenta}
      \addlegendimage{no markers, green}
      %\addlegendimage{dashed, green}
      \addlegendimage{no markers, orange}
      %\addlegendimage{dashed, orange}
      %\addlegendimage{no markers, olive}
      %\addlegendimage{dashed, olive}
      %\addlegendimage{no markers, violet}
      %\addlegendimage{dashed, violet}
      %\addlegendimage{no markers, yellow}
      %\addlegendimage{dashed, yellow}
      \addlegendimage{no markers, red}
      %\addlegendimage{dashed, red}
      \addlegendimage{no markers, blue}
      %\addlegendimage{dashed, blue}
      %\addlegendimage{no markers, teal}
      %\addlegendimage{dashed, teal}
    
    %neopets  best upper bound
    \addplot coordinates {(1,0.0957224) (2,0.13098) (4,0.195629) (8,0.291942) (16,0.417719) (32,0.635437) (64,0.983075) (128,1.51791) (256,2.31334) (512,3.46685) (1024,5.00121) (2048,6.96575) (4096,9.4566) (8192,12.3855) (16384,15.7543) (32768,19.6828) (65536,24.2589) (131072,29.3261) (262144,34.7341) (524288,40.4954) (1048576,46.7505) (2097152,53.885) (4194304,61.8731) 
     (8.38861e+06, 69.2588)
     (1.67772e+07, 76.8607)
     (3.35544e+07, 83.8409)
     (6.71089e+07, 91.0223)
     (1.34218e+08, 100.2) };

    %neopets best lower bound
    \addplot coordinates {(1,-0.917102) (2,-0.861102) (4,-0.825102) (8,-0.717102) (16,-0.573102) (32,-0.381102) (64,-0.0611016) (128,0.370898) (256,1.0509) (512,2.0989) (1024,3.4029) (2048,5.1429) (4096,7.5389) (8192,10.2989) (16384,13.6949) (32768,17.5549) (65536,22.0629) (131072,26.7669) (262144,32.2589) (524288,37.5029) (1048576,42.8029) (2097152,48.2829) (4194304,53.5029) (8.38861e+06, 59.2737)
     (1.67772e+07, 65.7034)
     (3.35544e+07, 70.2992)
     (6.71089e+07, 76.9127)
     (1.34218e+08, 77.3942)
     (2.68435e+08, 77.4139)
     (5.36871e+08, 77.4532)
     (1.07374e+09, 77.5317)
     (2.14748e+09,79.0634) (4.29497e+09,80.4874) (8.58993e+09,81.8794) (1.71799e+10,83.0154) (3.43597e+10,84.0954) (6.87195e+10,85.6114) (1.37439e+11,86.7754) (2.74878e+11,87.6554) (5.49756e+11,88.9034) (1.09951e+12,89.8994) (2.19902e+12,90.6394) (4.39805e+12,91.5754) (8.79609e+12,92.3274) (1.75922e+13,92.8194) (3.51844e+13,93.3834) (7.03687e+13,94.0234) (1.40737e+14,94.3474) (2.81475e+14,94.8634) (5.6295e+14,95.2634) (1.1259e+15,95.6154) (2.2518e+15,95.9594) (4.5036e+15,96.3034) (9.0072e+15,96.5274) (1e+16,96.5354)};

    %AttackerKnowledgeComp-d=25000-neopets-freq_plots.txtAttackerKnowledgeComparison
%N = 68345757  d = 25000  distinct = 27987227
%k = 10000
\addplot coordinates {(1,0.08) (2,0.096) (4,0.14) (8,0.224) (16,0.324) (32,0.404) (64,0.496) (128,0.748) (256,0.952) (512,1.04) (1024,1.136) (2048,1.432) (4096,2.096) (8192,3.164) (9814,3.648) };
%k = 100000
\addplot coordinates {(1,0.08) (2,0.116) (4,0.172) (8,0.256) (16,0.348) (32,0.576) (64,0.924) (128,1.376) (256,2.004) (512,3.016) (1024,4.112) (2048,5.504) (4096,7.184) (8192,7.78) (16384,8.32) (32768,9.432) (65536,11.668) (91251,13.356) };
%k = 1000000
\addplot coordinates {(1,0.08) (2,0.104) (4,0.172) (8,0.252) (16,0.396) (32,0.624) (64,0.956) (128,1.52) (256,2.288) (512,3.388) (1024,4.812) (2048,6.82) (4096,9.212) (8192,11.792) (16384,14.492) (32768,17.604) (65536,20.74) (131072,21.996) (262144,23.436) (524288,26.504) (774794,29.472) };
%k = 10000000
\addplot coordinates {(1,0.08) (2,0.104) (4,0.172) (8,0.252) (16,0.424) (32,0.632) (64,1) (128,1.536) (256,2.296) (512,3.44) (1024,4.956) (2048,6.888) (4096,9.328) (8192,12.172) (16384,15.516) (32768,19.3) (65536,23.752) (131072,28.676) (262144,33.416) (524288,37.2) (1048576,41.028) (2097152,43.06) (4194304,46.704) (5884500,49.876) };
%k = 68320757
\addplot coordinates {(1,0.08) (2,0.104) (4,0.172) (8,0.252) (16,0.428) (32,0.628) (64,0.996) (128,1.552) (256,2.324) (512,3.468) (1024,4.948) (2048,6.884) (4096,9.368) (8192,12.344) (16384,15.584) (32768,19.56) (65536,24.208) (131072,29.42) (262144,34.268) (524288,39.444) (1048576,45.02) (2097152,50.204) (4194304,55.216) (8388608,59.856) (16777216,63.292) (27979346,68.472) };

%MinGuessingNumber
%./minguessingnumber/neopets_minguessingnumber.txt
\addplot coordinates {(1,0) (2,0) (4,0.08) (8,0.084) (16,0.14) (32,0.224) (64,0.364) (128,0.636) (256,1.016) (512,1.592) (1024,2.332) (2048,3.232) (4096,4.66) (8192,6.26) (16384,8.256) (32768,10.624) (65536,13.616) (131072,17.18) (262144,21.336) (524288,24.352) (1.04858e+06,26.06) (2.09715e+06,27.676) (4.1943e+06,31.124) (8.38861e+06,39.008) (1.67772e+07,51.392) (3.35544e+07,58.356) (6.71089e+07,63.236) (1.34218e+08,66.64) (2.68435e+08,70.116) (5.36871e+08,73.096) (1.07374e+09,75.596) (2.14748e+09,77.692) (4.29497e+09,79.528) (8.58993e+09,81.364) (1.71799e+10,83.076) (3.43597e+10,84.556) (6.87195e+10,86.252) (1.37439e+11,87.544) (2.74878e+11,88.856) (5.49756e+11,90.064) (1.09951e+12,91.028) (2.19902e+12,91.804) (4.39805e+12,92.712) (8.79609e+12,93.452) (1.75922e+13,94.012) (3.51844e+13,94.564) (7.03687e+13,95.232) (1.40737e+14,95.716) (2.81475e+14,96.196) (5.6295e+14,96.544) (1.1259e+15,96.828) (2.2518e+15,97.16) (4.5036e+15,97.528) (9.0072e+15,97.744) (1e+16,97.768) };
    \end{axis}
  \end{tikzpicture}

\vspace{-.2cm}
 \caption{Attacker's Knowledge (Neopets)} \label{fig:neopetsattackerknowledge}
%%\vspace{-.15cm}
%\end{subfigure}
\vspace{-.45cm}
\end{figure}
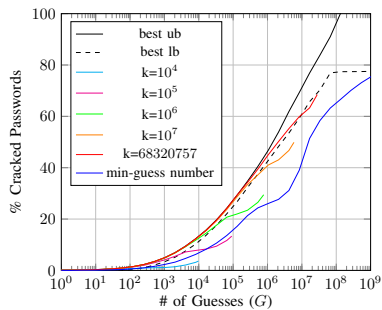

\fullversion{}{{\bf Zipf's Law in Passwords:}
Zipf's  law~\cite{malone2012investigating,wang2016implications,wang2017zipf} has been proposed as reasonable model for the password distribution $\lambda_G$. For example, CDF Zipf's law estimates that $\lambda_G \approx y G^r$ where the constant $y,r>0$ are tuned based on the sample $S$. Several papers \cite{wang2016implications,wang2017zipf,SP:BloHarZho18} find that CDF Zipf's law closely fits all known empirical password distributions. However, there is no theoretical guarantee that the estimate $y G^r$ is close to $\lambda_G$. \fullversion{we compare our upper/lower bounds with the CDF-Zipf curves using the parameters $y,r$ from \cite{wang2016implications} for the datasets RockYou, Battlefield Heroes, 000webhost, CSDN and \cite{SP:BloHarZho18} for the more recent Yahoo! dataset~\cite{SP:Bonneau12,NDSS:BloDatBon16}. The plots and parameter settings for Rockyou, Yahoo! and CSDN datasets are shown in Figure \ref{fig:zipf} and Table \ref{table:zipfparameter}. The comparison results for 000webhost and Battlefield Heroes are similar to the other three datasets. We leave their plots and parameter settings in the full version~\cite{fullversion}.
}{In Appendix~\ref{app:figures} Figure \ref{fig:zipf} we compare our upper/lower bounds with the CDF-Zipf curves using the parameters $y,r$ from \cite{wang2016implications} for the datasets RockYou, Battlefield Heroes, 000webhost, CSDN and \cite{SP:BloHarZho18} for the more recent Yahoo! dataset~\cite{SP:Bonneau12,NDSS:BloDatBon16} --- see Appendix~\ref{app:figures} Table \ref{table:zipfparameter}.} In all of the plots the CDF-Zipf plot (green) is close to our best upper bound (red). For the Battlefield, CSDN and 000webhost datasets the CDF-Zipf plot (green) lies in between our best upper bound (red) and our best lower bound (blue) indicating that the curve $yG^r$ is consistent with our statistical bounds. For the RockYou and Yahoo! datasets the CDF-Zipf plots (green) often lie above the red upper bound e.g., for the RockYou (resp. Yahoo!) dataset when $G=33554432$ (resp. $G=134217728$) we have $yG^r \geq 0.96$ (resp. $yG^r \geq 0.98$) while our upper bounds imply that $\lambda_G \leq 0.87$ (resp. $\lambda_G \leq 0.86$). In such cases we can confidently state that the CDF-Zipf curve overestimates $\lambda_G$. }

\section{Applications to Password Policies}
\label{sec:applicationtopasswordpolicies}
In this section we illustrate how our statistical techniques can be used to help guide password policies.   

\subsection{Tuning Password Hash Cost Parameters}\label{sec:application_hashcost}
We first consider the problem of tuning the cost parameter of a password hash function.  An offline attacker who has stolen the (salted) cryptographic hash of the user’s password can check as many passwords as s/he wants by computing the (salted) cryptographic hashes of likely password guesses to see if they match the stolen hash value. An offline attacker is limited only by the resources s/he is willing to invest cracking and by the cost of repeatedly evaluating the password hash function. Ideally a password hash function should be moderately expensive to compute so that it is impractical for an offline attacker to check millions or billions of password guesses. However, the function cannot be too expensive as the organization must also evaluate the hash function every time a user attempts to login.

Suppose that our organization is considering doubling the cost of the hash function. Assuming that the attacker's resources remain constant this would reduce the number of passwords that the offline attacker can check from $G$ to $G/2$ and the probability that the attacker cracks the user's password would decrease from $\lambda_G$ to $\lambda_{G/2}$. However, doubling the hash cost parameter will require the organization to invest additional computing resources to handle user authentications. Thus, the organization will only make the change if the security benefits are substantial enough. 

We can use our statistical techniques to upper/lower bound the security gain $\lambda_G - \lambda_{G/2}$ (resp. $\lambda_G - \lambda_{G/2^x}$) when increasing hashing costs by a factor of $2$ (resp. $2^x$). In particular, if $\mathbb{UB}(S,G)$ (resp. $\mathbb{LB}(S,G)$) denotes our upper/lower bounds on $\lambda_G$ then we know that 
\begingroup\makeatletter\def\f@size{8.5}\check@mathfonts
\[\mathbb{LB}(S,G) - \mathbb{UB}(S,G/2^x) \leq  \lambda_G - \lambda_{G/2^x} \leq \mathbb{UB}(S,G)- \mathbb{LB}(S,G/2^x) \ . \] \endgroup

Figure~\ref{fig:hashcostneopets} plots our upper/lower bounds on $\lambda_G - \lambda_{G/2^x}$ for $x \in \{1,2,3,4,5\}$ for the Neopets dataset --- similar results for Yahoo!, RockYou and 000webhost are deferred to \fullversion{the full version \cite{fullversion}}{Figure~\ref{fig:hashcost-app} in Appendix~\ref{app:figures}}. For example, suppose that the attacker was originally able to attempt $G=1.67 \times 10^7$ password guesses before we increase our guessing costs by $2^5$. The figure indicates that this policy change will reduce the number of cracked passwords by {\em at least} $25.2\%$ and {\em at most} $37.7\%$. If reducing the potential damage of an offline attack by $25.2\%$ (resp. $37.7\%$) is (resp. is not) worth the additional costs then the organization should (resp. should not) increase the hash cost parameter. We remark that an organization that transitions away from a password hash function like PBKDF2 or BCRYPT to a modern memory hard function like scrypt~\cite{percival2009stronger}, Argon2~\cite{biryukov2016argon2} or DRSample~\cite{CCS:AlwBloHAr17} will substantially increase hashing costs. 

We remark that upper/lower bounding $ \lambda_G - \lambda_{G/b} $ can be useful when defending against online password spraying attacks. Suppose that our organization is considering adopting a stricter version of its lockout policy where an account is locked whenever there are at least $G/b$ (as opposed to $G$) incorrect login attempts within a $30$ day window. For reference, NIST Authentication Guidelines \cite{NISTAUTHGUIDE} suggests limiting the number of incorrect login attempts to $G=100$ within a 30 day window. Adopting a stricter lockout policy comes with a usability cost so our organization would only adopt the policy if the security benefits ($ \lambda_G - \lambda_{G/b} $) are substantial enough e.g., with $G=128$ in Figure ~\ref{fig:hashcostneopets} we have $0.491\%\leq \lambda_{G}-\lambda_{G/2} \leq 0.579\%$.

%tune hash cost parameters

\begin{figure}\centering
\includegraphics[scale=0.6]{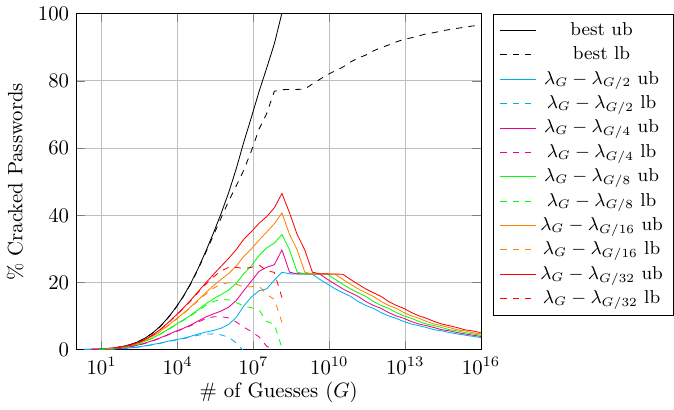}
%\begin{subfigure}[htb]{0.25\textwidth}
%\vspace{-0.2cm}

\vspace{-.2cm}
 %\caption{Neopets} 
 \caption{Tuning Password Hash Cost Parameters (Neopets)}
 \label{fig:hashcostneopets}
%\vspace{-.15cm}
%subfigure}
\vspace{-.45cm}
\end{figure}

\subsection{Comparing Password Distributions}
\label{sec:application_yahoocorpus}

The Yahoo! frequency corpus \cite{SP:Bonneau12,NDSS:BloDatBon16} allows us to compare password distributions along several dimensions: gender, account tenure and composition policy (before/after) the adoption of a six character minimum restriction. The comparison results are shown in Figure~\ref{fig:yahoocorpus} in Appendix~\ref{app:figures}.

\noindent {\bf Measuring Shifts in The Password Distributions over Time}
We generated upper and lower bounds for Yahoo! passwords with account tenure of 5-10 years and account tenure below 1 year as shown in Figure~\ref{fig:yahootenure}. We found statistically significant evidence for Bonneau’s claim~\cite{SP:Bonneau12} that there is ``a weak trend towards improvement over time, with more recent accounts having slightly stronger passwords''. For example, when our attacker makes $G=8192$ guesses, at least $16.86\%$ old account passwords (5-10 years) will be guessed while at most $14.40\%$ new account passwords (0-1 year) will be guessed (confidence: $98\%$). One caveat is that the weak trend towards stronger passwords may be partially explained by Yahoo!’s adoption of the six-character minimum requirement.

\noindent {\bf Measuring the Effect of Gender on Password Strength}
Bonneau~\cite{SP:Bonneau12} previously found that (self-reported) gender had a small/split effect on password security with ``male-chosen passwords being slightly more vulnerable to online attack and slightly stronger against offline attack.'' We are able to provide statically significant justification for these claims using our upper/lower bounds as shown in Figure~\ref{fig:yahoogender}. First, we can provide statistically significant evidence that female passwords are more resistant to an attacker making $G \leq 128$ guesses. For example, when $G=32$ we find that the attacker will crack at least $2.776\%$ (resp. at most $2.283\%$) of male (resp. female) passwords. For $512 \leq G \leq 1048576$ the conclusions are reversed e.g., when $G=131072$ we find that the attacker will crack at most $30.050\%$ (resp. at least $32.219\%$) of male (resp. female) passwords. (Note: When $G > 1048576$ the upper/lower bounds start to diverge preventing us from making definitive comparisons. Also, when $G=256$ the upper/lower bounds for male/female passwords are very close).

\noindent {\bf Measuring The Effect of Password Requirements At Registration}
We generated upper/lower bounds for Yahoo! passwords that were selected with (and without) a six character minimum restriction as shown in Figure~\ref{fig:yahooreq}. Bonneau~\cite{SP:Bonneau12} previously concluded that the change made ``almost no difference'' in security against online guessing while slightly increasing the resistance to offline attacks. By contrast, we find statistically significant evidence that an online attacker making at most $G=8$ guesses will crack more passwords picked under the six character minimum restriction. For example, when an online attacker makes $G=8$ guesses, we can say (with $98\%$ confidence) that the attacker will crack between $1.393\%$ to $1.610\%$ (resp. $1.645\%$ to $1.880\%$) of passwords picked without the restriction (resp. with the restriction). The findings are reversed for an attacker making $256 \leq G \leq 524288$ guesses yielding statistically significant evidence for Bonneau’s finding that passwords picked under the six-character minimum restriction are more resistant to offline attacks. (Note: for larger values of G the upper/lower bounds begin to diverge preventing us from making a definitive comparison).

\subsection{Password Composition Policies}
%%%%PCP comparison
\begin{figure*}
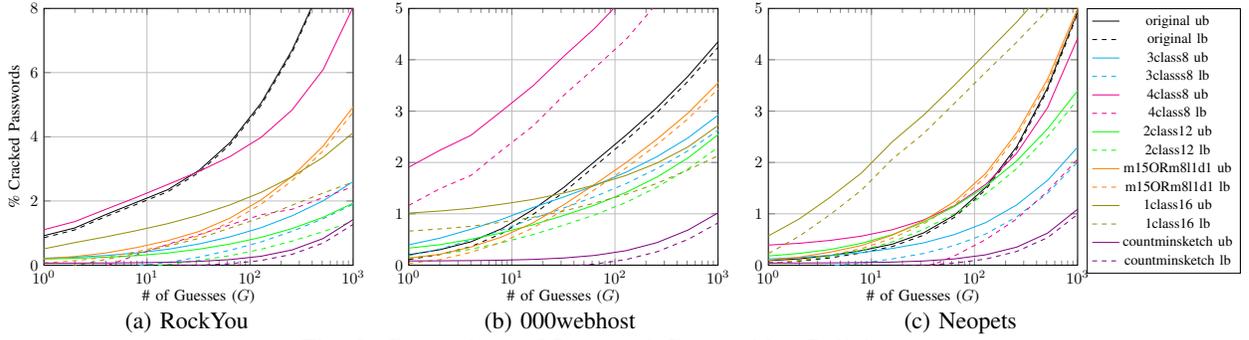
\centering
%rockyou
\input{Plots/PCP-rockyou}
%\vspace{0.5em}
%000webhost
\input{Plots/PCP-000webhost}
%\vspace{0.5em}
%\hfill
%neopets
\input{Plots/PCP-neopets}
%\vspace{0.5em}
%\hfill
%\newline
\vspace{-0.2cm}
\caption{Comparison of Password Composition Policies}
\label{fig:PCP}
\vspace{-0.4cm}
\end{figure*}

We now turn our attention to the problem of identifying secure password composition policies which yield password distributions that are more resistant to online (resp. offline) password cracking attacks. One immediate challenge we face is that, with the exception of the Yahoo! frequency corpus, none of the available password datasets record the passwords that each user would have selected in response to (additional) restrictions. One potential remedy would be to run a large user study to ask users what password they would select under a variety of restrictions i.e., obtaining samples from each of the resulting distributions for each password composition policy. However, as we previously discussed one limitation of our statistical techniques is that we need a reasonably large number of samples from each distribution before we would be able to draw meaningful comparisons e.g., see Figure \ref{fig:samplesizeimpact}. In practice, the number of participants $N$ (e.g., $N \leq 10^4$) in any user study would be heavily constrained by practical budget considerations and would be too small to allow us to draw meaningful comparisons using our statistical techniques.

We address the challenges above by following \cite{EC:BKPS13} and making a heuristic assumption about the way that users respond to password restrictions. Suppose that we are given a predicate $\mathtt{Allowed}$ describing our password composition policy i.e., $\mathtt{Allowed}(pwd) = 1$ if and only if users are allowed to select the password $pwd$. If $\mathcal{P}_1$ (resp. $\mathcal{P}_2$) denotes the probability distribution before (resp. after) adopting the composition policy then the normalized probabilities model \cite{EC:BKPS13} says that for each password $pwd$ with $\mathtt{Allowed}(pwd)=1$ we have 
\[ \Pr_{x \leftarrow \mathcal{P}_2}\left[x=pwd \right] = \Pr_{x \leftarrow \mathcal{P}_1}\left[ x=pwd ~|~\mathtt{Allowed}(x)=1 \right] \ . \]
Intuitively, we can imagine that each user utilizes rejection sampling and repeatedly samples passwords from $\mathcal{P}_1$ until s/he finds one that is consistent with the composition policy. While this heuristic assumption may not perfectly model how users respond to password restrictions, the assumption allows us to apply our statistical techniques to compare candidate password composition policies and identify the most promising candidates for further evaluation. In particular, if the password dataset $S_1$ denotes $N$ iid samples from the distribution $\mathcal{P}_1$ then we can filter  $S_1$ by removing passwords that are incompatible with the composition policy to obtain $S_2 = \{pwd\in S~:~Allowed(pwd)=1\}$. Now $S_2$ can be viewed as $N'\leq N$ iid samples from $\mathcal{P}_2$. Thus, we can apply our statistical techniques to obtain upper/lower bounds ($\mathbb{UB}(S_2,G)$ and $\mathbb{LB}(S_2,G)$) for the new distribution. In our analysis we primarily focus on smaller guessing numbers $G$ to determine whether or not the new distribution is more/less resistant to online cracking.  

We analyze a wide variety of candidate PCPs. In general, we use $x$class$y$ to denote the policy which requires the user to pick a password with {\em at least} $y$ characters total coming from {\em at least} $x$ distinct character classes (lowercase letters, uppercase letters, digits, special characters). For example, 4class8 requires passwords that are at least 8 characters long and include at least one lowercase, one uppercase, one digit and one special character. We also consider the PCP used by Github.com (at least 15 characters OR at least 8 characters including a number and a lower case letter).  Finally, we consider a proposal of Schechter et al.~\cite{schechter2010popularity} to use a count-min-sketch data structure to estimate the frequency of each user password and ban passwords whose (estimated) frequency is above a threshold $r \times N$. We instantiated this policy with the differentially private count-min-sketch as described in \cite{blocki2020dalock} which adds laplace noise to each cell in the count-sketch data-structure to preserve $\epsilon$-differential privacy. Since the count-min-sketch is trained based on sampled passwords $S$ we adopted the following approach to ensure that the comparison with other PCPs is fair. 1) We randomly partition the dataset $S$ into two equal size datasets $S_{train}$ and $S_{test}$ of size $N/2$. 2) The dataset $S_{train}$ is used to construct the noisy count-min-sketch (depth=5, width =$10^6$, 20MB of space) using the differential privacy parameter $\epsilon=0.1$ and then discarded. 3) We define the password policy $\mathtt{Allowed}(pwd) =1$ if and only the estimated frequency of $pwd$ is {\em at most} $r \times N/2$ --- we fixed $r= 10^{-5}$ in our experiments below and explore different threshold parameters $r \in \{10^{-4}, 10^{-5}, 10^{-6}\}$ in \fullversion{the full version \cite{fullversion}}{Appendix~\ref{app:figures} Figure~\ref{fig:PCPfull}}. Because the count-sketch was trained only using passwords from the dataset $S_{train}$ we can filter $S_{test}$ to obtain fresh samples that are consistent with the composition policy.   

Figure~\ref{fig:PCP} shows the upper and lower bounds of different PCPs as well as the upper/lower bounds of the original dataset (\textbf{original ub}, \textbf{original lb}) of the three datasets. To avoid clutter we omitted PCPs (e.g. $1class8$, $2class8$, $1class12$) that performed similarly to others. Since we are focusing primarily on online attacks, we enlarge the plots and only show $1\leq G\leq 1000$ in this figure. Note: to get tighter bounds with small $G$, we set the parameter $d$ in $\mathtt{SamplingLB}$ to be $N/2$ and set the error rate $\epsilon$ of $\mathtt{FrequencyUB}$ to be $0.01$. We leave the full plots with $G> 1000$ in \fullversion{the full version \cite{fullversion}}{Appendix~\ref{app:figures} Figure~\ref{fig:PCPfull}}.  

An interesting finding is that no single PCP, excluding count-min sketch, performs well across all datasets and parameter ranges. For example, $1.6\%$ of the 000webhost passwords that comply with the policy are $P@ssw0rd$ so  applying the $4class8$ rule actually increases the percentage of cracked passwords as shown in Figure \ref{fig:000webhostPCPs-zoomin}. Similarly, the $1class16$ policy performs particularly poorly on the Neopets dataset. On the positive side the count-min sketch PCP universally performed well across all three datasets and for all guessing parameters $G$. For example, when $G=128$ applying the count-min sketch PCP to the 000webhost dataset reduces the percentage of cracked passwords from at least $2.434\%$ to {\em at most} $0.283\%$. For comparison, applying the  $4class8$ (resp. $1class16$) {\em increases} (resp. {\em decreases}) the percentage of cracked passwords to  at least $4.395\%$ (resp. at least $1.379\%$). Our analysis supports the proposal of Schechter et al. \cite{schechter2010popularity} to base password composition policies on (estimated) password frequency with the caveat that our analysis in this section does depend on a heuristic assumption about how password composition policies impact the distribution. We also remark that when the guessing number $G$ is low that it is more likely that a real world attacker will closely resemble a perfect knowledge attacker. In particular, the attacker only needs to know the top $G$ passwords and we expect that each of these passwords will occur relatively frequently making them the easiest to learn. 

\subsection{Discussion: Obtaining Password Samples} \label{subsec:securesamples} 
Before applying our statistical techniques we implicitly assume that our organization has been able to obtain independent samples $S=(s_1,\ldots, s_N)$ from the unknown distribution $\mathcal{P}$ over user passwords. While our organization can always rely on a breached password dataset such as 000webhost or RockYou, it is unlikely that the password samples will be perfectly representative of the current user password distribution e.g., due to varying demographic factors, language, culture, account value and password restrictions. Ideally, an organization would obtain the samples  $S=(s_1,\ldots, s_N)$ from {\em its own} users before applying our statistical framework. While the collection of such a sample raises security and privacy concerns, we stress that we do not require  plaintext passwords in order to apply our statistical techniques. Bonneau \cite{SP:Bonneau12} already developed an efficient framework to securely collect an anomymized password frequency list and Blocki et al.~\cite{NDSS:BloDatBon16} developed an efficient differentially private algorithm that was used to publish a differentially private version of Yahoo! frequency corpus that Bonneau collected ($N \approx 70$ million users). An organization could adapt this same framework to securely and privately collect a password frequency corpus from its own users and then apply our statistical techniques to characterize the attacker's guessing curve.

%\section{Discussion and Future Work}

\section{Conclusion}
We introduced several statistical techniques to upper and lower bound $\lambda_G$ the performance of a password cracker who knows the distribution from which passwords were sampled. Our upper/lower bounds hold with high confidence and can be derived from a password dataset even when the real password distribution is unknown to us. We applied our technique to analyze several large empirical password datasets. Our analysis demonstrates that the empirical guessing curve closely matches the real guessing curve as long as $G$ is not too large, and highlights that state-of-the-art password cracking models are often far less guess efficient than a perfect knowledge attacker. For example, with the 000webhost dataset the state-of-the-art password cracking models indicate the attacker making $8338608$ guesses would crack any password with probability $14\%$ while our lower bound shows the probability is at least $39.16\%$ with high confidence ($99\%$). This shows that even the most sophisticated password cracking models still have a large room for improvement. 
We also demonstrate how to apply our theoretical bounds to determine whether or not particular
password interventions (e.g., key-stretching, imposing restrictions on the passwords users can pick) yield effective defenses.
While our results significantly narrows the uncertain region for $\lambda_G$, there are regions where our best upper/lower bounds diverge significantly. Reducing this gap will help us to better understand the distribution over user chosen passwords and is an important challenge for future research.

\section*{Acknowledgments}
This research was supported in part by the National Science Foundation under awards CNS \#1755708 and CNS \#2047272, a gift from Protocol Labs, and by a Purdue Big Ideas award. We would like to thank anonymous reviewers for constructive feedback which helped us to improve the presentation of this paper. 

%\section*{References}
\bibliographystyle{IEEEtran}
\bibliography{abbrev3,crypto,references}

\appendices
\section{The Linear Program \texttt{LPupper}}\label{app:LPupper}

Below we show the linear program $\mathtt{LPupper}$ that is used to generate the upper bound ($\mathtt{LPUB}(S,G)$) in Theorem~\ref{thm:bound3final-lower}.

\fbox{\begin{minipage}{8cm}
\textbf{Linear Programming Task 3: \\ }$\mathtt{LPupper}(G,X_l,F^S,idx,i',\mathbf{\epsilon_2},\mathbf{\epsilon_3},\mathbf{\hat{x}_{\epsilon_3}})$ \\
\textbf{Input Parameters:} $G$, $X_l=\{x_1,...,x_l\}$, $F^S=\{F^S_1,...,F^S_N\}$, $idx$, $i'$, $\epsilon_2=\{\epsilon_{2,0},\ldots,\epsilon_{2,i'}\}$, $\mathbf{\hat{x}_{\epsilon_3}}=\{\hat{x}_{\epsilon_{3,0}},\ldots,\hat{x}_{\epsilon_{3,i'}}\}$, $\mathbf{\epsilon_3}=\{\epsilon_{3,i}=\frac{1}{q^{i+1}}\left(\frac{1-\hat{x}_{\epsilon_{3,i}}}{1-q\hat{x}_{\epsilon_{3,i}}}\right)^{N-i}-1, 0\leq i\leq i'\}$ \\
\textbf{Variables:} $h_1,...,h_l,c, p$ \\
%\textbf{Objective:} $\min\left(b\times (\sum_{j<idx}h_j+c)\right)$ \\ 
\textbf{Objective:} $\max\left(\sum_{j<idx}h_j\times x_j +c\times x_{idx}\right)$ \\ 
%$\min\left(b\times (\sum_{j<idx}h_j\times x_j +c\times x_{idx})\right)$ \\ 
\textbf{Constraints:} 
\begin{enumerate}
    \item $\sum_{j<idx}h_j+c = G$
    \item $\forall 0\leq i\leq i'$:
    \begin{enumerate}
        \item for $i=0$, $\frac{1}{1+\epsilon_{3,i}}(\frac{(i+1)F^S_{i+1}}{N-i}-\epsilon_{2,i} - \frac{i+1}{N-i} - p - \bpdf(i,N,q\hat{x}_{\epsilon_{3,i}})) \leq \sum_{j=1}^l h_j\times x_j\times \bpdf(i,N,x_j) \leq q^{i+1}(\frac{(i+1)F^S_{i+1}}{N-i}+\epsilon_{2,i} - p\times \bpdf(i,N,x_l))$ 
        \item for $1\leq i\leq i'$, $\frac{1}{1+\epsilon_{3,i}}(\frac{(i+1)F^S_{i+1}}{N-i}-\epsilon_{2,i} - \frac{i+1}{N-i} - p\times \bpdf(i,N,x_l) - \bpdf(i,N,q\hat{x}_{\epsilon_{3,i}})) \leq \sum_{j=1}^l h_j\times x_j\times \bpdf(i,N,x_j) \leq q^{i+1}(\frac{(i+1)F^S_{i+1}}{N-i}+\epsilon_{2,i})$
    \end{enumerate}
    \item $ 1-p\leq \sum_{j=1}^l h_j\times x_j \leq q \times (1 - p)$ 
    \item $0 \leq c \leq h_{idx}$
\end{enumerate}
(\textbf{Note:} we consider $idx=1,2,...,l+1$. When $idx=l+1$, we define $h_{l+1}=G$ and $x_{l+1} = x_l$.)
\end{minipage}}
\fullversion{\section{Zipf's Law in Passwords}
Zipf's  law~\cite{malone2012investigating,wang2016implications,wang2017zipf} has been proposed as reasonable model for the password distribution $\lambda_G$. For example, CDF Zipf's law estimates that $\lambda_G \approx y G^r$ where the constant $y,r>0$ are tuned based on the sample $S$. Several papers \cite{wang2016implications,wang2017zipf,SP:BloHarZho18} find that CDF Zipf's law closely fits all known empirical password distributions. However, there is no theoretical guarantee that the estimate $y G^r$ is close to $\lambda_G$. \fullversion{we compare our upper/lower bounds with the CDF-Zipf curves using the parameters $y,r$ from \cite{wang2016implications} for the datasets RockYou, Battlefield Heroes, 000webhost, CSDN and \cite{SP:BloHarZho18} for the more recent Yahoo! dataset~\cite{SP:Bonneau12,NDSS:BloDatBon16}. The plots and parameter settings for Rockyou, Yahoo! and CSDN datasets are shown in Figure \ref{fig:zipf} and Table \ref{table:zipfparameter}. The comparison results for 000webhost and Battlefield Heroes are similar to the other three datasets. We leave their plots and parameter settings in the full version~\cite{fullversion}.
}{In Figure \ref{fig:zipf} we compare our upper/lower bounds with the CDF-Zipf curves using the parameters $y,r$ from \cite{wang2016implications} for the datasets RockYou, Battlefield Heroes, 000webhost, CSDN and \cite{SP:BloHarZho18} for the more recent Yahoo! dataset~\cite{SP:Bonneau12,NDSS:BloDatBon16} --- see Table \ref{table:zipfparameter}.} In all of the plots the CDF-Zipf plot (green) is close to our best upper bound (red). For the Battlefield, CSDN and 000webhost datasets the CDF-Zipf plot (green) lies in between our best upper bound (red) and our best lower bound (blue) indicating that the curve $yG^r$ is consistent with our statistical bounds. For the RockYou and Yahoo! datasets the CDF-Zipf plots (green) often lie above the red upper bound e.g., for the RockYou (resp. Yahoo!) dataset when $G=33554432$ (resp. $G=134217728$) we have $yG^r \geq 0.96$ (resp. $yG^r \geq 0.98$) while our upper bounds imply that $\lambda_G \leq 0.87$ (resp. $\lambda_G \leq 0.86$). In such cases we can confidently state that the CDF-Zipf curve overestimates $\lambda_G$. 

%zipf law
\input{Plots/zipf}

\begin{table}[!t]
  \caption{CDF Zipf's Law Parameters~\cite{wang2016implications,SP:BloHarZho18}}
  \label{table:zipfparameter}
  \centering
  \begin{tabular}{ccc}
    \toprule
    \textbf{Dataset ($S$)} & y & r \\
    \midrule
    Yahoo!~\cite{NDSS:BloDatBon16} & 0.03315 & 0.1811 \\ 
    RockYou~\cite{rockyoudataset}  & 0.037433 & 0.187227 \\ 
    \fullversion{}{000webhost~\cite{000webhostdataset} & 0.005858& 0.281557 \\ }
    \fullversion{}{Battlefield Heroes~\cite{bfieldsdataset} & 0.010298 & 0.294932  \\} 
    CSDN~\cite{csdndataset} & 0.058799 & 0.148573  \\    
  \bottomrule
\end{tabular}
\end{table}}{}
\section{Missing Figures And Tables}\label{app:figures}

The basic information of the eight datasets used in this paper is in Table~\ref{table:BasicInformation}.

\begin{table}[!t]
  \vspace{-0.1cm}
  \caption{Basic Information for Password Datasets}
  \vspace{-0.15cm}
  \label{table:BasicInformation}
  \centering
  \begin{tabular}{crrc}
    \toprule
    \textbf{Dataset ($S$)} & \textbf{\# Passwords ($N$)} & \textbf{\# Distinct} & \textbf{\# Unique} \\
    \midrule
    Yahoo!~\cite{NDSS:BloDatBon16} & 69301337 & 33895873 & 29452171 \\ 
    RockYou~\cite{rockyoudataset}  & 32603388 & 14344391 & 11884632 \\ 
    %LinkedIn~\cite{bicycleAttacks} & 174292189 & 57431283 & 21424510 \\ 
    000webhost~\cite{000webhostdataset} & 15268903 & 10592935 & 9006529 \\ 
    Neopets~\cite{neopetsdataset} & 68345757 & 27987227 & 21509860 \\ 
    Battlefield Heroes~\cite{bfieldsdataset} & 541016 & 416130  & 373549  \\ 
    Brazzers~\cite{brazzersdataset} & 925614 & 587934 & 491136 \\ 
    Clixsense~\cite{clixsensedataset} & 2222529 & 1628577 & 1455585 \\ 
    CSDN~\cite{csdndataset} & 6428449 & 4037749 & 3581824  \\    
  \bottomrule
\end{tabular}
%\vspace{-.70cm}
\end{table}

Figure~\ref{fig:yahoocorpus} shows the comparison results of Yahoo! frequency corpus discussed in Section~\ref{sec:application_yahoocorpus}.
%%%%yahoo corpus
\begin{figure*}\centering
%yahoo req
\begin{subfigure}[htb]{0.30\textwidth}
%\vspace{-0.2cm}
      \begin{tikzpicture}[scale=0.6]
      \begin{axis}[
        title style={align=center},
        xlabel={\# of Guesses ($G$)},
        xmin = {1},
        xmode = log,
        log basis x={10},
        xlabel shift = -3pt,
        ylabel={\% Cracked Passwords},
        ymax={100},
        ymin = {0},
        %ymode = log,
        %log basis y={2},
        %ylabel shift = -3pt,
        grid=major,
        %small,
        cycle list = { {black}, {black,dashed}, {cyan}, {cyan, dashed}, {magenta}, {magenta,dashed}, {green}, {green, dashed}, {orange}, {orange, dashed}, {olive}, {olive, dashed}},
        legend style = {font=\small}, %{font=\fontsize{3}{3}\selectfont},
        %{font=\tiny},
        legend pos = north west, %outer north east,
        legend entries = { 5-10y ub, 5-10y lb, 0-1y ub, 0-1y lb }
        %, 1-2y ub, 1-2y lb, 2-3y ub, 2-3y lb, 3-4y ub, 3-4y lb, 4-5y ub, 4-5y lb 
      ]
      \addlegendimage{no markers, black}
      \addlegendimage{dashed, black}
      \addlegendimage{no markers, cyan}
      \addlegendimage{dashed, cyan}
      %\addlegendimage{no markers, magenta}
      %\addlegendimage{dashed, magenta}
      %\addlegendimage{no markers, green}
      %\addlegendimage{dashed, green}
      %\addlegendimage{no markers, orange}
      %\addlegendimage{dashed, orange}
      %\addlegendimage{no markers, olive}
      %\addlegendimage{dashed, olive}

    %yahoo-5-10y best ub
    \addplot coordinates {(1,0.960356) (2,1.07466) (4,1.24574) (8,1.54023) (16,1.86362) (32,2.31153) (64,2.97732) (128,3.98035) (256,5.28204) (512,6.95081) (1024,8.99671) (2048,11.4163) (4096,14.1633) (8192,17.1249) (16384,20.3033) (32768,23.9129) (65536,28.0729) (131072,32.7666) (262144,37.817) (524288,42.9915) (1048576,48.4528) 
    (2.09715e+06, 54.2552)
     (4.1943e+06, 58.9121)
     (8.38861e+06, 64.2257)
     (1.67772e+07, 69.3899)
     (3.35544e+07, 76.0345)
     (6.71089e+07, 85.0845)
     (1.34218e+08, 97.6359)
     (1.54218e+08, 100.201)
     %(2.68435e+08, 100.24)
    };

    %yahoo-5-10yfrequencyub
    %\addplot coordinates {(1,0.960356) (2,1.07466) (4,1.24574) (8,1.54023) (16,1.86362) (32,2.31153) (64,2.97732) (128,3.98035) (256,5.28204) (512,6.95081) (1024,8.99671) (2048,11.4163) (4096,14.1633) (8192,17.1249) (16384,20.3033) (32768,23.9129) (65536,28.0729) (131072,32.7666) (262144,37.817) (524288,42.9915) (1048576,48.4528) (2097152,54.8945) (4194304,62.1) (8388608,76.511) (15236685,100.04) };

    %yahoo-5-10y best lb
    \addplot coordinates {(1,0.770908) (2,0.878608) (4,1.05541) (8,1.35641) (16,1.68091) (32,2.11921) (64,2.77771) (128,3.79401) (256,5.07651) (512,6.75521) (1024,8.81711) (2048,11.21) (4096,13.9461) (8192,16.8669) (16384,19.985) (32768,23.5146) (65536,27.5421) (131072,32.0191) (262144,36.6819) (524288,41.0112) (1048576,44.6497) (2097152,47.4627) 
     (4.1943e+06, 50.2425)
     (8.38861e+06, 51.7157)
     (1.67772e+07, 54.3966)
     (3.35544e+07, 58.7042)
     (6.71089e+07, 59.9971)
     (1.34218e+08, 60.0202)
     (2.68435e+08, 60.0663)
     (5.36871e+08, 60.1585)
     (1.07374e+09, 60.343)
     (2.14748e+09, 60.53) };

    %lambda_Gnewlowerbound-d=1000000-yahoo-5-10y_plots.txt
    %1 <= G <= distinct  delta = 0.01delta1 = 9e-05  delta2 = 0.00991
    %\addplot coordinates {(1,0.770908) (2,0.878608) (4,1.05541) (8,1.35641) (16,1.68091) (32,2.11921) (64,2.77771) (128,3.79401) (256,5.07651) (512,6.75521) (1024,8.81711) (2048,11.21) (4096,13.9461) (8192,16.8669) (16384,19.985) (32768,23.5146) (65536,27.5421) (131072,32.0191) (262144,36.6819) (524288,41.0112) (1048576,44.6497) (2097152,47.4627) (4194304,48.5008) (8388608,50.5621) (14776534,53.7209) };

     %yahoo-0y upper bound

     %yahoo-0y  lambda_G  myLP2  max
     
     \addplot coordinates {(1,0.935976) (2,1.26971) (4,1.50496) (8,1.75888) (16,2.07438) (32,2.38515) (64,2.79711) (128,3.40556) (256,4.27826) (512,5.40385) (1024,6.77994) (2048,8.77412) (4096,11.4657) (8192,14.3999) (16384,17.5038) (32768,20.7755) (65536,24.5145) (131072,28.9406) (262144,34.4322) 
     (524288, 39.5969)
     (1.04858e+06, 44.9943)
     (2.09715e+06, 49.9288)
     (4.1943e+06, 56.1358)
     (8.38861e+06, 64.4596)
     (1.67772e+07, 75.874)
     (3.35544e+07, 91.825)
     (4.35544e+07, 99.3798)
     (5.35544e+07, 100.217)
     %(6.71089e+07, 100.243)
     };

    %yahoo-0y frequency ub
    %\addplot coordinates {(1,0.935976) (2,1.26971) (4,1.50496) (8,1.75888) (16,2.07438) (32,2.38515) (64,2.79711) (128,3.40556) (256,4.27826) (512,5.40385) (1024,6.77994) (2048,8.77412) (4096,11.4657) (8192,14.3999) (16384,17.5038) (32768,20.7755) (65536,24.5145) (131072,28.9406) (262144,34.4322) (524288,40.7281) (1048576,50.8446) (2097152,71.0775) (3600983,100.095) };

   %yahoo-0y best lowerbound
   \addplot coordinates {(1,0.693708) (2,1.02451) (4,1.26371) (8,1.51611) (16,1.82851) (32,2.14111) (64,2.54481) (128,3.15121) (256,4.01231) (512,5.12921) (1024,6.46651) (2048,8.40431) (4096,11.085) (8192,13.9391) (16384,16.8767) (32768,19.7255) (65536,22.5843) (131072,25.5363) (262144,28.6116) 
   (524288, 30.5887)
     (1.04858e+06, 32.5729)
     (2.09715e+06, 34.1114)
     (4.1943e+06, 37.1502)
     (8.38861e+06, 41.2632)
     (1.67772e+07, 41.6906)
     (3.35544e+07, 41.723)
     (6.71089e+07, 41.7878)
     (1.34218e+08, 41.9172)
     (2.68435e+08, 42.1762)
     (5.36871e+08, 42.4665)
     };

    \end{axis}
  \end{tikzpicture}

\vspace{-.2cm}
 \caption{Account Tenure} \label{fig:yahootenure}
%\vspace{-0.15cm}
\end{subfigure}
%\vspace{0.5em}
%yahoo req
\begin{subfigure}[htb]{0.3\textwidth}%{0.45\textwidth}
%\vspace{-0.2cm}
      \begin{tikzpicture}[scale=0.6]
      \begin{axis}[
        title style={align=center},
        xlabel={\# of Guesses ($G$)},
        xmin = {1},
        xmode = log,
        log basis x={10},
        xlabel shift = -3pt,
        ylabel={\% Cracked Passwords},
        ymax={100},
        ymin = {0},
        %ymode = log,
        %log basis y={2},
        %ylabel shift = -3pt,
        grid=major,
        %small,
        cycle list = { {black}, {black,dashed}, {cyan}, {cyan, dashed}, {magenta}, {magenta,dashed}, {green}, {green, dashed}},
        legend style = {font=\small}, %{font=\fontsize{3}{3}\selectfont},
        %{font=\tiny},
        legend pos = north west, %outer north east,
        legend entries = {none ub, none lb, 6 char minimum ub, 6 char minimum lb}
      ]
      \addlegendimage{no markers, black}
      \addlegendimage{dashed, black}
      \addlegendimage{no markers, cyan}
      \addlegendimage{dashed, cyan}
      %\addlegendimage{no markers, magenta}
      %\addlegendimage{dashed, magenta}
      %\addlegendimage{no markers, green}
      %\addlegendimage{dashed, green}
      
     %yahoo-nonereq best upper bound
      \addplot coordinates {(1,1.06855) (2,1.17568) (4,1.3355) (8,1.60976) (16,1.9526) (32,2.43202) (64,3.14171) (128,4.19755) (256,5.60167) (512,7.38374) (1024,9.57744) (2048,12.1665) (4096,15.0884) (8192,18.2268) (16384,21.6069) (32768,25.5258) (65536,29.9898) (131072,34.9148) (262144,40.0981) (524288,45.4643) (1048576,51.3438) (2.09715e+06, 56.4277) (4.1943e+06, 61.7983) (8.38861e+06, 66.5364) (1.67772e+07, 72.3536) (3.35544e+07, 80.0998) (6.71089e+07, 90.794) (1.34218e+08, 100.213) };
      
      %yahoo-nonereq frequency lambda_G upper bound
    %yahoo-nonereq
    %\addplot coordinates {(1,1.06855) (2,1.17568) (4,1.3355) (8,1.60976) (16,1.9526) (32,2.43202) (64,3.14171) (128,4.19755) (256,5.60167) (512,7.38374) (1024,9.57744) (2048,12.1665) (4096,15.0884) (8192,18.2268) (16384,21.6069) (32768,25.5258) (65536,29.9898) (131072,34.9148) (262144,40.0981) (524288,45.4643) (1048576,51.3438) (2097152,57.8094) (4194304,68.072) (8388608,88.5972) (10728510,100.048) };

    %yahoo-nonereq  lambda_G  myLP2  max
    %\addplot coordinates { (1, 40.9842) (2, 40.9843) (4, 40.9843) (8, 40.9843) (16, 40.9845) (32, 40.9847) (64, 40.9853) (128, 40.9864) (256, 40.9887) (512, 40.9932) (1024, 41.0022) (2048, 41.0203) (4096, 41.0564) (8192, 41.1286) (16384, 41.2727) (32768, 41.5605) (65536, 42.1336) (131072, 43.2696) (262144, 45.4725) (524288, 48.7401) (1.04858e+06, 52.4912) (2.09715e+06, 56.4277) (4.1943e+06, 61.7983) (8.38861e+06, 66.5364) (1.67772e+07, 72.3536) (3.35544e+07, 80.0998) (6.71089e+07, 90.794) (1.34218e+08, 100.213) };

    %yahoo-nonereq best lower bound
    \addplot coordinates {(1,0.850908) (2,0.963608) (4,1.12501) (8,1.39301) (16,1.73931) (32,2.21661) (64,2.93211) (128,3.99051) (256,5.38911) (512,7.16051) (1024,9.34931) (2048,11.9435) (4096,14.8634) (8192,18.0016) (16384,21.3342) (32768,25.1382) (65536,29.4131) (131072,34.0704) (262144,38.6845) (524288,42.7746) (1048576,46.135) (2097152,47.7919) (4.1943e+06, 50.7657) (8.38861e+06, 52.6671) (1.67772e+07, 56.2644) (3.35544e+07, 59.4388) (6.71089e+07, 59.4552) (1.34218e+08, 59.4881) (2.68435e+08, 59.5538) (5.36871e+08, 59.6851) (1.07374e+09, 59.9479) (2.14748e+09, 59.9859) (4.29497e+09, 59.9859) };

    %lambda_Gnewlowerbound-d=1000000-yahoo-nonereq_plots.txt
    %1 <= G <= distinct  delta = 0.01delta1 = 9e-05  delta2 = 0.00991
    %\addplot coordinates {(1,0.850908) (2,0.963608) (4,1.12501) (8,1.39301) (16,1.73931) (32,2.21661) (64,2.93211) (128,3.99051) (256,5.38911) (512,7.16051) (1024,9.34931) (2048,11.9435) (4096,14.8634) (8192,18.0016) (16384,21.3342) (32768,25.1382) (65536,29.4131) (131072,34.0704) (262144,38.6845) (524288,42.7746) (1048576,46.135) (2097152,47.7919) (4194304,49.2687) (8388608,52.205) (10266814,53.5132) };

    %yahoo-nonereq  lambda_G  myLP2  min
    %\addplot coordinates { (1, -2.18618e-08) (2, -2.18618e-08) (4, 0.000173878) (8, 0.000347778) (16, 0.000695578) (32, 0.00139118) (64, 0.00278238) (128, 0.00556488) (256, 0.0111297) (512, 0.0222594) (1024, 0.0445188) (2048, 0.0890377) (4096, 0.178076) (8192, 0.356151) (16384, 0.712302) (32768, 1.4246) (65536, 2.8492) (131072, 5.69841) (262144, 11.3968) (524288, 22.7936) (1.04858e+06, 41.737) (2.09715e+06, 47.7917) (4.1943e+06, 50.7657) (8.38861e+06, 52.6671) (1.67772e+07, 56.2644) (3.35544e+07, 59.4388) (6.71089e+07, 59.4552) (1.34218e+08, 59.4881) (2.68435e+08, 59.5538) (5.36871e+08, 59.6851) (1.07374e+09, 59.9479) (2.14748e+09, 59.9859) (4.29497e+09, 59.9859) (8.58993e+09, 59.9859) (1.71799e+10, 59.9859) (5.49756e+11, 59.9859) };

    %yahoo-6charmin  best upper bound
    \addplot coordinates {(1,1.20171) (2,1.42078) (4,1.58828) (8,1.8797) (16,2.19691) (32,2.60122) (64,3.18202) (128,4.00213) (256,5.12775) (512,6.54083) (1024,8.31144) (2048,10.4245) (4096,12.8903) (8192,15.5856) (16384,18.6236) (32768,22.2381) (65536,26.5104) (131072,31.3486) (262144,36.5253) (524288,42.0775)  (1.04858e+06, 48.0369) (2.09715e+06, 52.9919) (4.1943e+06, 58.0863) (8.38861e+06, 63.4835) (1.67772e+07, 70.5496) (3.35544e+07, 80.2073) (6.71089e+07, 93.6606) 
    (8.71089e+07, 100.025)
    %(1.34218e+08, 100.235) 
    };
    
    %yahoo-6charmin  frequency upper bound
    %yahoo-6charmin
    %\addplot coordinates {(1,1.20171) (2,1.42078) (4,1.58828) (8,1.8797) (16,2.19691) (32,2.60122) (64,3.18202) (128,4.00213) (256,5.12775) (512,6.54083) (1024,8.31144) (2048,10.4245) (4096,12.8903) (8192,15.5856) (16384,18.6236) (32768,22.2381) (65536,26.5104) (131072,31.3486) (262144,36.5253) (524288,42.0775) (1048576,48.5977) (2097152,56.4626) (4194304,72.1924) (7909580,100.059) };
    
    %yahoo-6charmin  lambda_G  myLP2  max
    %\addplot coordinates { (1, 34.2924) (2, 34.2924) (4, 34.2924) (8, 34.2925) (16, 34.2928) (32, 34.2932) (64, 34.2941) (128, 34.2958) (256, 34.2992) (512, 34.3061) (1024, 34.32) (2048, 34.3476) (4096, 34.4028) (8192, 34.5131) (16384, 34.7335) (32768, 35.1728) (65536, 36.0455) (131072, 37.7662) (262144, 40.6771) (524288, 44.3017) (1.04858e+06, 48.0369) (2.09715e+06, 52.9919) (4.1943e+06, 58.0863) (8.38861e+06, 63.4835) (1.67772e+07, 70.5496) (3.35544e+07, 80.2073) (6.71089e+07, 93.6606) (1.34218e+08, 100.235) };

    %yahoo-6charmin best lower bound
    \addplot coordinates {(1,0.982808) (2,1.19791) (4,1.36311) (8,1.64591) (16,1.95821) (32,2.36121) (64,2.93201) (128,3.75811) (256,4.88411) (512,6.29401) (1024,8.05171) (2048,10.1677) (4096,12.624) (8192,15.2804) (16384,18.2258) (32768,21.7162) (65536,25.7222) (131072,30.1033) (262144,34.3584) (524288,37.8512) (1048576,40.3666) (2.09715e+06, 43.24) (4.1943e+06, 44.5585) (8.38861e+06, 47.1905) (1.67772e+07, 51.3871) (3.35544e+07, 52.5011) (6.71089e+07, 52.5262) (1.34218e+08, 52.5766) (2.68435e+08, 52.6772) (5.36871e+08, 52.8786) (1.07374e+09, 53.136) (2.14748e+09, 53.136) (4.29497e+09, 53.136) };

    %lambda_Gnewlowerbound-d=1000000-yahoo-6charmin_plots.txt
    %1 <= G <= distinct  delta = 0.01delta1 = 9e-05  delta2 = 0.00991
    %\addplot coordinates {(1,0.982808) (2,1.19791) (4,1.36311) (8,1.64591) (16,1.95821) (32,2.36121) (64,2.93201) (128,3.75811) (256,4.88411) (512,6.29401) (1024,8.05171) (2048,10.1677) (4096,12.624) (8192,15.2804) (16384,18.2258) (32768,21.7162) (65536,25.7222) (131072,30.1033) (262144,34.3584) (524288,37.8512) (1048576,40.3666) (2097152,41.3765) (4194304,43.457) (7379213,46.5519) };

    %yahoo-6charmin  lambda_G  myLP2  min
    %\addplot coordinates { (1, 3.48094e-09) (2, 0.000125803) (4, 0.000251503) (8, 0.000503103) (16, 0.0010062) (32, 0.0020123) (64, 0.0040246) (128, 0.0080493) (256, 0.0160985) (512, 0.032197) (1024, 0.0643941) (2048, 0.128788) (4096, 0.257576) (8192, 0.515152) (16384, 1.03031) (32768, 2.06061) (65536, 4.12122) (131072, 8.24244) (262144, 16.4849) (524288, 32.9464) (1.04858e+06, 38.9696) (2.09715e+06, 43.24) (4.1943e+06, 44.5585) (8.38861e+06, 47.1905) (1.67772e+07, 51.3871) (3.35544e+07, 52.5011) (6.71089e+07, 52.5262) (1.34218e+08, 52.5766) (2.68435e+08, 52.6772) (5.36871e+08, 52.8786) (1.07374e+09, 53.136) (2.14748e+09, 53.136) (4.29497e+09, 53.136) (8.58993e+09, 53.136) (1.71799e+10, 53.136) (3.43597e+10, 53.136) (6.87195e+10, 53.136) (1.37439e+11, 53.136) (2.74878e+11, 53.136) (5.49756e+11, 53.136) };    
    
    \end{axis}
  \end{tikzpicture}
\vspace{-.2cm}
 \caption{Password Requirements At Registration} \label{fig:yahooreq}
%\vspace{-.15cm}
\end{subfigure}
%\vspace{0.5em}
%\hfill
%yahoo gender
\begin{subfigure}[htb]{0.3\textwidth}
%\vspace{-0.2cm}
      \begin{tikzpicture}[scale=0.6]
      \begin{axis}[
        title style={align=center},
        xlabel={\# of Guesses ($G$)},
        xmin = {1},
        xmode = log,
        log basis x={10},
        xlabel shift = -3pt,
        ylabel={\% Cracked Passwords},
        ymax={100},
        ymin = {0},
        %ymode = log,
        %log basis y={2},
        %ylabel shift = -3pt,
        grid=major,
        %small,
        cycle list = { {black}, {black,dashed}, {cyan}, {cyan, dashed}, {magenta}, {magenta,dashed}, {green}, {green, dashed}},
        legend style = {font=\small}, %{font=\fontsize{3}{3}\selectfont},
        %{font=\tiny},
        legend pos = north west, %outer north east,
        legend entries = {female ub, female lb, male ub, male lb}
      ]
      \addlegendimage{no markers, black}
      \addlegendimage{dashed, black}
      \addlegendimage{no markers, cyan}
      \addlegendimage{dashed, cyan}
      %\addlegendimage{no markers, magenta}
      %\addlegendimage{dashed, magenta}
      %\addlegendimage{no markers, green}
      %\addlegendimage{dashed, green}
      
      %yahoo-female best upper bound
      \addplot coordinates {(1,0.896962) (2,1.05908) (4,1.2408) (8,1.49424) (16,1.81866) (32,2.28213) (64,2.96363) (128,3.9945) (256,5.33269) (512,7.0052) (1024,9.05464) (2048,11.4158) (4096,14.0415) (8192,16.9146) (16384,20.0788) (32768,23.722) (65536,28.0029) (131072,32.8467) (262144,38.0222) (524288,43.3395) (1048576,48.8387) (2.09715e+06, 54.7016) (4.1943e+06, 59.2246) (8.38861e+06, 64.7467) (1.67772e+07, 69.9519) (3.35544e+07, 76.577) (6.71089e+07, 85.5687) (1.34218e+08, 98.0337)
      (1.54218e+08, 100.202)
      %(2.68435e+08, 100.239) 
      };
    %yahoo-female frequency
    %\addplot coordinates {(1,0.896962) (2,1.05908) (4,1.2408) (8,1.49424) (16,1.81866) (32,2.28213) (64,2.96363) (128,3.9945) (256,5.33269) (512,7.0052) (1024,9.05464) (2048,11.4158) (4096,14.0415) (8192,16.9146) (16384,20.0788) (32768,23.722) (65536,28.0029) (131072,32.8467) (262144,38.0222) (524288,43.3395) (1048576,48.8387) (2097152,55.3319) (4194304,62.1975) (8388608,75.9288) (15753280,100.039) };

%yahoo female best lowerbounds
\addplot coordinates {(1,0.711408) (2,0.872708) (4,1.06291) (8,1.32061) (16,1.64781) (32,2.10761) (64,2.79941) (128,3.82981) (256,5.16451) (512,6.84031) (1024,8.87211) (2048,11.2571) (4096,13.8917) (8192,16.7482) (16384,19.8802) (32768,23.4373) (65536,27.5913) (131072,32.2193) (262144,36.9941) (524288,41.5056) (1048576,45.2746) (2097152,48.1767)  (4.1943e+06, 50.682)
 (8.38861e+06, 52.3257)
 (1.67772e+07, 55.017)
 (3.35544e+07, 59.4992)
 (6.71089e+07, 61.2041)
 (1.34218e+08, 61.226)
 (2.68435e+08, 61.27)
 (5.36871e+08, 61.3579)
 (1.07374e+09, 61.5336)
 (2.14748e+09, 61.7214) };

    %lambda_Gnewlowerbound-d=1000000-yahoo-female_plots.txt
    %1 <= G <= distinct  delta = 0.01delta1 = 9e-05  delta2 = 0.00991
    %\addplot coordinates {(1,0.711408) (2,0.872708) (4,1.06291) (8,1.32061) (16,1.64781) (32,2.10761) (64,2.79941) (128,3.82981) (256,5.16451) (512,6.84031) (1024,8.87211) (2048,11.2571) (4096,13.8917) (8192,16.7482) (16384,19.8802) (32768,23.4373) (65536,27.5913) (131072,32.2193) (262144,36.9941) (524288,41.5056) (1048576,45.2746) (2097152,48.1767) (4194304,49.1947) (8388608,51.2149) (15301438,54.5535) };

    %yahoo-male best upper bound
    \addplot coordinates {(1,1.30442) (2,1.56192) (4,1.76338) (8,2.09232) (16,2.48907) (32,2.95068) (64,3.52595) (128,4.30686) (256,5.35404) (512,6.69782) (1024,8.38167) (2048,10.4399) (4096,12.8333) (8192,15.4842) (16384,18.3854) (32768,21.6833) (65536,25.6022) (131072,30.0498) (262144,34.817) (524288,39.7489) (1048576,44.8917) (2097152,50.83) 
    (4.1943e+06, 55.4863)
 (8.38861e+06, 61.1937)
 (1.67772e+07, 66.7338)
 (3.35544e+07, 73.0044)
 (6.71089e+07, 81.3248)
 (1.34218e+08, 92.7084)
 (1.54218e+08, 95.4756)
 (1.74218e+08, 98.0643)
 (1.94218e+08, 100.2)
 %(2.68435e+08, 100.22) 
 };
    
    %yahoo-male frequencyub
    %\addplot coordinates {(1,1.30442) (2,1.56192) (4,1.76338) (8,2.09232) (16,2.48907) (32,2.95068) (64,3.52595) (128,4.30686) (256,5.35404) (512,6.69782) (1024,8.38167) (2048,10.4399) (4096,12.8333) (8192,15.4842) (16384,18.3854) (32768,21.6833) (65536,25.6022) (131072,30.0498) (262144,34.817) (524288,39.7489) (1048576,44.8917) (2097152,50.83) (4194304,57.4877) (8388608,68.3468) (16777216,90.0652) (20627909,100.035) };

    %yahoo male best lowerbound
    \addplot coordinates {(1,1.12351) (2,1.37641) (4,1.57911) (8,1.90851) (16,2.30681) (32,2.77641) (64,3.34591) (128,4.13321) (256,5.19171) (512,6.53801) (1024,8.22301) (2048,10.2794) (4096,12.6511) (8192,15.2985) (16384,18.1775) (32768,21.3749) (65536,25.1783) (131072,29.4928) (262144,33.9305) (524288,38.2237) (1048576,42.0357) (2097152,45.0785) 
    (4.1943e+06, 47.0195)
     (8.38861e+06, 49.9273)
     (1.67772e+07, 52.102)
     (3.35544e+07, 56.2448)
     (6.71089e+07, 60.3747)
     (1.34218e+08, 60.3965)
     (2.68435e+08, 60.4313)
     (5.36871e+08, 60.5008)
     (1.07374e+09, 60.6398)
     (2.14748e+09, 60.9127)
     };

    %lambda_Gnewlowerbound-d=1000000-yahoo-male_plots.txt
    %1 <= G <= distinct  delta = 0.01delta1 = 9e-05  delta2 = 0.00991
    %\addplot coordinates {(1,1.12351) (2,1.37641) (4,1.57911) (8,1.90851) (16,2.30681) (32,2.77641) (64,3.34591) (128,4.13321) (256,5.19171) (512,6.53801) (1024,8.22301) (2048,10.2794) (4096,12.6511) (8192,15.2985) (16384,18.1775) (32768,21.3749) (65536,25.1783) (131072,29.4928) (262144,33.9305) (524288,38.2237) (1048576,42.0357) (2097152,45.0785) (4194304,46.7547) (8388608,48.3692) (16777216,51.6244) (20159687,52.9316) };

    \end{axis}
  \end{tikzpicture}

\vspace{-.2cm}
 \caption{Gender (self-reported)} \label{fig:yahoogender}
%\vspace{-.15cm}
%\vspace{0.3cm}
\end{subfigure}
%\vspace{0.5em}
%\hfill
%\newline
\vspace{-0.2cm}
\caption{Analysis on Yahoo! Corpus}
\label{fig:yahoocorpus}
%\vspace{-0.4cm}
\end{figure*}

%linkedin
%\input{Plots/linkedin}

We have shown the best upper/lower bounds for Yahoo!, RockYou, 000webhost, and Neopets in Figure~\ref{fig:bestguessingcurves} in the main body of the paper. Here we plot the best upper/lower bounds for  Battlefield Heroes, Brazzers, Clixsense, and CSDN datasets in Figure~\ref{fig:app-bestguessingcurves}.
%remaining best bounds
\begin{figure}\centering
%\begin{subfigure}[b]{0.45\textwidth}
\vspace{-0.2cm}
      \begin{tikzpicture}[scale=0.6]
      \begin{axis}[
        title style={align=center},
        xlabel={\# of Guesses ($G$)},
        xmin = {1},
        xmax = {10^16},
        xmode = log,
        log basis x={10},
        xlabel shift = -3pt,
        ylabel={\% Cracked Passwords ($\lambda_G$)},
        ymax={100},
        ymin = {0},
        %ymode = log,
        %log basis y={2},
        %ylabel shift = -3pt,
        grid=major,
        %small,
        %cycle list = {{red, mark=triangle},  {red, dashed,mark=triangle}, {blue, mark=square},{blue,dashed, mark=square},  {green, mark=diamond}, {green, dashed,mark=diamond}}
        cycle list = { {black}, {black,dashed}, {cyan}, {cyan, dashed}, {magenta}, {magenta,dashed}, {green}, {green, dashed}},
        legend style = {font=\small},%{font=\tiny},
        legend pos = outer north east,%south east,%outer north east,%north west,
        legend entries = {BattlefieldHero ub, BattlefieldHero lb, Brazzer ub, Brazzer lb, Clixsense ub, Clixsense lb, CSDN ub, CSDN lb}
      ]
      \addlegendimage{no markers, black}
      \addlegendimage{dashed, black}
      %\addlegendimage{no markers, gray}
      %\addlegendimage{no markers, blue}
      \addlegendimage{no markers, cyan}
      \addlegendimage{dashed, cyan}
      %\addlegendimage{no markers, violet}
      \addlegendimage{no markers, magenta}
      \addlegendimage{dashed, magenta}
      %\addlegendimage{no markers, teal}
      \addlegendimage{no markers, green}
      \addlegendimage{dashed, green}
      
     %bfield best upper bound
    \addplot coordinates {(1,0.771778) (2,0.903012) (4,1.09506) (8,1.3298) (16,1.70151) (32,2.22294) (64,2.94362) (128,3.88167) (256,5.06721) (512,6.58584) (1024,8.51129) (2048,10.8785) (4096,13.7274) (8192,17.1678) (16384,21.5632) (32768,27.62) 
    %(32768, 27.7765)
     (65536, 32.9858)
     (131072, 40.7454)
     (262144, 48.7366)
     (524288, 58.1881)
     (1.04858e+06, 70.9568)
     (2.09715e+06, 88.7054)
     (4.1943e+06, 100.223) };

    %bfield best lower bound
    \addplot coordinates {(1,-0.727288) (2,-0.583288) (4,-0.367288) (8,-0.179288) (16,0.160712) (32,0.708712) (64,1.46071) (128,2.36071) (256,3.52071) (512,4.97671) (1024,6.98471) (2048,9.16471) (4096,11.7087) (8192,14.1607) (16384,16.5407) (32768,18.5967) (65536,20.3327) (131072, 22.1847) (262144,25.6807) (524288, 32.0409)
     (1.04858e+06, 39.3487)
     (2.09715e+06, 39.6827)
     (4.1943e+06,40.8274) (8.38861e+06,44.1634) (1.67772e+07,49.6914) (3.35544e+07,52.7874) (6.71089e+07,54.9234) (1.34218e+08,57.6074) (2.68435e+08,60.1234) (5.36871e+08,62.9234) (1.07374e+09,65.7274) (2.14748e+09,68.2234) (4.29497e+09,70.4714) (8.58993e+09,72.7194) (1.71799e+10,74.5954) (3.43597e+10,76.5034) (6.87195e+10,79.1114) (1.37439e+11,80.9474) (2.74878e+11,82.2394) (5.49756e+11,84.0674) (1.09951e+12,85.6754) (2.19902e+12,86.7634) (4.39805e+12,88.3594) (8.79609e+12,89.6514) (1.75922e+13,90.4634) (3.51844e+13,91.6474) (7.03687e+13,92.6874) (1.40737e+14,93.2234) (2.81475e+14,93.9874) (5.6295e+14,94.7034) (1.1259e+15,95.0954) (2.2518e+15,95.5554) (4.5036e+15,96.0914) (9.0072e+15,96.3554) (1e+16,96.4114) 
     };

    %brazzers best upperbound
    %brazzers
    \addplot coordinates {(1,0.80286) (2,1.26493) (4,1.72203) (8,2.06278) (16,2.52474) (32,3.18701) (64,3.98453) (128,5.14074) (256,6.75901) (512,8.82759) (1024,11.3939) (2048,14.4578) (4096,18.0367) (8192,22.0835) (16384,26.7884) (32768,32.706) (65536,40.4089) (131072, 47.8269)
     (262144, 57.2776)
     (524288, 67.038)
     (1.04858e+06, 77.886)
     (2.09715e+06, 92.3372)
     (4.1943e+06, 100.215) };
    
    %brazzers  best lower bound
    \addplot coordinates {(1,-0.695791) (2,-0.695791) (4,-0.223791) (8,0.148209) (16,0.624209) (32,1.39621) (64,2.08021) (128,3.20421) (256,4.96821) (512,6.92821) (1024,9.33621) (2048,12.2402) (4096,15.7882) (8192,19.4642) (16384,22.8962) (32768,26.5442) (65536,30.1962) (131072, 33.9423)
     (262144, 38.8023)
     (524288, 44.6133)
     (1.04858e+06, 54.4467)
     (2.09715e+06, 58.0515)
     (4.1943e+06, 58.0741)
     (8.38861e+06, 58.1194)
     (1.67772e+07,63.7674) (3.35544e+07,66.5754) (6.71089e+07,68.8314) (1.34218e+08,71.3394) (2.68435e+08,73.6554) (5.36871e+08,76.0434) (1.07374e+09,78.1074) (2.14748e+09,80.0394) (4.29497e+09,81.6754) (8.58993e+09,83.3034) (1.71799e+10,84.5794) (3.43597e+10,85.8154) (6.87195e+10,87.4634) (1.37439e+11,88.7634) (2.74878e+11,89.5194) (5.49756e+11,90.6754) (1.09951e+12,91.4114) (2.19902e+12,91.9914) (4.39805e+12,92.6634) (8.79609e+12,93.2594) (1.75922e+13,93.7274) (3.51844e+13,94.3634) (7.03687e+13,94.9074) (1.40737e+14,95.1274) (2.81475e+14,95.4114) (5.6295e+14,95.6914) (1.1259e+15,95.7994) (2.2518e+15,95.9914) (4.5036e+15,96.2274) (9.0072e+15,96.3354) (1e+16,96.3434) 
     };

     %clixsense best upper bound
    \addplot coordinates {(1,0.949166) (2,1.09738) (4,1.28014) (8,1.53422) (16,1.84247) (32,2.14141) (64,2.53623) (128,3.07044) (256,3.79345) (512,4.94164) (1024,6.58414) (2048,8.57746) (4096,10.8924) (8192,13.5206) (16384,16.5411) (32768,20.1487) (65536,24.8842) (131072,30.8802) (262144, 36.2333)
     (524288, 43.3275)
     (1.04858e+06, 51.0357)
     (2.09715e+06, 59.6372)
     (4.1943e+06, 71.1096)
     (8.38861e+06, 86.7998)
     (1.67772e+07, 100.216) };

    %clixsense best lower bound
    \addplot coordinates {(1,-0.299282) (2,-0.0752823) (4,0.0767177) (8,0.312718) (16,0.628718) (32,0.916718) (64,1.32872) (128,1.79672) (256,2.57272) (512,3.70072) (1024,5.35272) (2048,7.26072) (4096,9.48472) (8192,11.9487) (16384,14.6727) (32768,17.2167) (65536,20.3127) (131072,22.7487) (262144, 25.1386)
     (524288, 27.7971)
     (1.04858e+06, 30.6247)
     (2.09715e+06, 36.1184)
     (4.1943e+06, 43.0037)
     (8.38861e+06, 43.5821)
     (1.67772e+07,47.3154) (3.35544e+07,50.1674) (6.71089e+07,52.8394) (1.34218e+08,55.5474) (2.68435e+08,58.0434) (5.36871e+08,60.5994) (1.07374e+09,63.0554) (2.14748e+09,65.4514) (4.29497e+09,68.0994) (8.58993e+09,70.6114) (1.71799e+10,72.7914) (3.43597e+10,74.9514) (6.87195e+10,77.4754) (1.37439e+11,79.4434) (2.74878e+11,80.8994) (5.49756e+11,82.8434) (1.09951e+12,84.3434) (2.19902e+12,85.3514) (4.39805e+12,86.6594) (8.79609e+12,87.8034) (1.75922e+13,88.6314) (3.51844e+13,89.6154) (7.03687e+13,90.5594) (1.40737e+14,91.0834) (2.81475e+14,91.7674) (5.6295e+14,92.5394) (1.1259e+15,92.9154) (2.2518e+15,93.3714) (4.5036e+15,94.0794) (9.0072e+15,94.3594) (1e+16,94.3914) };

    %csdn best upperbound
    \addplot coordinates {(1,3.74093) (2,7.05043) (4,8.95445) (8,10.3196) (16,10.9835) (32,11.7707) (64,12.775) (128,13.6396) (256,14.4731) (512,15.4077) (1024,16.6743) (2048,18.415) (4096,20.7587) (8192,23.3233) (16384,25.5672) (32768,27.8101) (65536,30.4126) (131072,33.8073) (262144,38.3379) (524288, 43.0939)
     (1.04858e+06, 48.549)
     (2.09715e+06, 56.4268)
     (4.1943e+06, 64.1956)
     (8.38861e+06, 73.5898)
     (1.67772e+07, 86.1397)
     (3.35544e+07, 100.206) };

    %csdn  best lower bound
    \addplot coordinates {(1,2.65608) (2,6.04408) (4,8.09608) (8,9.48808) (16,10.1681) (32,10.9881) (64,12.0441) (128,12.9361) (256,13.7961) (512,14.7841) (1024,16.1041) (2048,17.8361) (4096,20.4641) (8192,22.9521) (16384,25.0201) (32768,26.9241) (65536,28.7481) (131072,30.7841) (262144,32.5081) (524288,34.7921) (1.04858e+06, 37.3967)
     (2.09715e+06, 39.8581)
     (4.1943e+06, 44.1232)
     (8.38861e+06, 51.3936)
     (1.67772e+07, 54.7118)
     (3.35544e+07, 54.7379)
     (6.71089e+07, 54.7901)
     (1.34218e+08, 54.8945)
     (2.68435e+08, 55.1032)
     (5.36871e+08, 55.3261)
     (1.07374e+09, 55.3261)
     (2.14748e+09,55.5234) (4.29497e+09,57.3994) (8.58993e+09,59.4634) (1.71799e+10,62.2474) (3.43597e+10,65.5074) (6.87195e+10,68.9954) (1.37439e+11,71.4354) (2.74878e+11,74.5394) (5.49756e+11,77.1994) (1.09951e+12,79.2714) (2.19902e+12,81.0154) (4.39805e+12,82.6234) (8.79609e+12,84.2834) (1.75922e+13,85.3834) (3.51844e+13,86.5834) (7.03687e+13,87.9154) (1.40737e+14,88.8634) (2.81475e+14,89.7954) (5.6295e+14,90.7394) (1.1259e+15,91.2994) (2.2518e+15,91.9594) (4.5036e+15,92.7274) (9.0072e+15,93.2034) (1e+16,93.2554) };

    \end{axis}
  \end{tikzpicture}
\vspace{-.2cm}
 \caption{Best Upper/Lower Bounds for Battlefield Heroes, Brazzers, Clixsense, and CSDN} 
 \label{fig:app-bestguessingcurves}
%\end{subfigure}
%\vspace{-.25cm}
\end{figure}

\fullversion{}{Figure~\ref{fig:attackerknowledge-app} shows the performance of the partial knowledge attacker for 000webhost, Yahoo! and RockYou datasets similar to Figure~\ref{fig:neopetsattackerknowledge}.
%attacker's knowledge
%%%%attacker's knowledge
\begin{figure*}\centering
%neopets
%\input{Plots/attackerknowledge-neopets}
%\vspace{0.5em}
%000wehost
\begin{subfigure}[htb]{0.29\textwidth}
%\vspace{-0.2cm}
      \begin{tikzpicture}[scale=0.6]
      \begin{axis}[
        title style={align=center},
        xlabel={\# of Guesses ($G$)},
        xmin = {1},
        xmax = {10^9},
        xmode = log,
        log basis x={10},
        xlabel shift = -3pt,
        ylabel={\% Cracked Passwords ($\lambda_G$)},
        ymax={100},
        ymin = {0},
        %ymode = log,
        %log basis y={2},
        %ylabel shift = -3pt,
        grid=major,
        %small,
        cycle list = { 
        {black}, {black,dashed}, 
        {cyan}, %{cyan, dashed}, 
        {magenta}, %{magenta,dashed}, 
        {green}, %{green, dashed}, 
        {orange}, %{orange, dashed}, 
        %{olive}, {olive, dashed}, 
        %{violet}, {violet, dashed}, 
        %{yellow}, {yellow, dashed}, 
        {red}, %{red, dashed},
        {blue}, {blue, dashed},
        {teal}, {teal, dashed}
        },
        legend style = {font=\small}, %{font=\fontsize{3}{3}\selectfont},
        %{font=\tiny},
        legend pos = north west,
        legend entries = { 
        best ub, best lb, 
        k=$10^4$,
        k=$10^5$,
        k=$10^6$,
        k=$10^7$,
        k=$15243903$,
        MinGuessingNumber,
        }
      ]
      \addlegendimage{no markers, black}
      \addlegendimage{dashed, black}
      \addlegendimage{no markers, cyan}
      %\addlegendimage{dashed, cyan}
      \addlegendimage{no markers, magenta}
      %\addlegendimage{dashed, magenta}
      \addlegendimage{no markers, green}
      %\addlegendimage{dashed, green}
      \addlegendimage{no markers, orange}
      %\addlegendimage{dashed, orange}
      %\addlegendimage{no markers, olive}
      %\addlegendimage{dashed, olive}
      %\addlegendimage{no markers, violet}
      %\addlegendimage{dashed, violet}
      %\addlegendimage{no markers, yellow}
      %\addlegendimage{dashed, yellow}
      \addlegendimage{no markers, red}
      %\addlegendimage{dashed, red}
      \addlegendimage{no markers, blue}
      %\addlegendimage{dashed, blue}
      %\addlegendimage{no markers, teal}
      %\addlegendimage{dashed, teal}
    
     %000webhost best bound
    \addplot coordinates {(1,0.218426) (2,0.317399) (4,0.470442) (8,0.72841) (16,1.09745) (32,1.53355) (64,2.04037) (128,2.54783) (256,3.08776) (512,3.69693) (1024,4.39047) (2048,5.21538) (4096,6.19832) (8192,7.3926) (16384,8.87285) (32768,10.7385) (65536,13.1469) (131072,16.3296) (262144,20.6116) (524288,26.4025) (1048576,34.0244)
     (2.09715e+06, 40.5463)
     (4.1943e+06, 48.9989)
     (8.38861e+06, 60.0133)
     (1.67772e+07, 70.1013)
     (3.35544e+07, 82.9886)
     (6.71089e+07, 100.201) };

    %000webhost best lower bound
    \addplot coordinates {(1,-0.858182) (2,-0.810182) (4,-0.634182) (8,-0.358182) (16,0.00581763) (32,0.461818) (64,0.973818) (128,1.51382) (256,2.00182) (512,2.51782) (1024,3.19382) (2048,3.99782) (4096,4.87382) (8192,5.93382) (16384,7.47382) (32768,9.19382) (65536,11.1938) (131072,13.8378) (262144,16.9898) (524288,20.7498) (1048576,24.7098)  (2.09715e+06, 29.7143)
     (4.1943e+06, 35.0216)
     (8.38861e+06, 39.1658)
     (1.67772e+07, 47.1249)
     (3.35544e+07, 55.5526)
     (6.71089e+07, 55.642)
     (1.34218e+08, 55.686)
     (2.68435e+08, 55.7739)
     (5.36871e+08, 55.9497)
     (1.07374e+09, 56.2309)
     (2.14748e+09, 56.2309)
     (4.29497e+09, 56.2309)
     (8.58993e+09, 56.2309)
     (1.71799e+10, 56.2309)
     (3.43597e+10, 56.2309)
     (6.87195e+10,57.4434) (1.37439e+11,58.9754) (2.74878e+11,60.2594) (5.49756e+11,62.0594) (1.09951e+12,63.7874) (2.19902e+12,65.0834) (4.39805e+12,67.0314) (8.79609e+12,69.1794) (1.75922e+13,70.5514) (3.51844e+13,72.7354) (7.03687e+13,74.8914) (1.40737e+14,76.0074) (2.81475e+14,77.3474) (5.6295e+14,78.7674) (1.1259e+15,79.4714) (2.2518e+15,80.2794) (4.5036e+15,81.5314) (9.0072e+15,82.1594) (1e+16,82.2434) };
    
    %AttackerKnowledgeComp-d=25000-000webhost-freq_plots.txtAttackerKnowledgeComparison
%N = 15268903  d = 25000  distinct = 10592935
%k = 10000
\addplot coordinates {(1,0.112) (2,0.172) (4,0.3) (8,0.532) (16,0.828) (32,1.184) (64,1.416) (128,1.668) (256,1.688) (512,1.752) (1024,1.796) (2048,1.884) (4096,2.112) (8192,2.52) (9826,2.66) };
%k = 100000
\addplot coordinates {(1,0.112) (2,0.172) (4,0.288) (8,0.564) (16,0.908) (32,1.38) (64,1.928) (128,2.38) (256,2.82) (512,3.348) (1024,3.676) (2048,4.02) (4096,4.092) (8192,4.184) (16384,4.444) (32768,4.884) (65536,5.688) (95470,6.392) };
%k = 1000000
\addplot coordinates {(1,0.112) (2,0.172) (4,0.296) (8,0.584) (16,0.972) (32,1.388) (64,1.904) (128,2.448) (256,3.056) (512,3.624) (1024,4.256) (2048,5.092) (4096,5.92) (8192,6.776) (16384,7.716) (32768,8.748) (65536,9.428) (131072,9.976) (262144,10.964) (524288,12.596) (891409,15.256) };
%k = 10000000
\addplot coordinates {(1,0.112) (2,0.172) (4,0.296) (8,0.584) (16,0.956) (32,1.384) (64,1.888) (128,2.48) (256,3.004) (512,3.652) (1024,4.384) (2048,5.212) (4096,6.172) (8192,7.404) (16384,8.768) (32768,10.424) (65536,12.352) (131072,14.884) (262144,17.648) (524288,20.908) (1048576,24.716) (2097152,26.7) (4194304,30.48) (7357744,36.3) };
%k = 15243903
\addplot coordinates {(1,0.112) (2,0.172) (4,0.296) (8,0.584) (16,0.956) (32,1.388) (64,1.884) (128,2.484) (256,3.008) (512,3.68) (1024,4.372) (2048,5.228) (4096,6.236) (8192,7.42) (16384,8.76) (32768,10.444) (65536,12.476) (131072,15.16) (262144,18.296) (524288,21.9) (1048576,25.64) (2097152,29.692) (4194304,32.524) (8388608,38.26) (10578265,41.304) };

%MinGuessingNumber
%./minguessingnumber/000webhost_minguessingnumber.txt
\addplot coordinates {(1,0) (2,0) (4,0) (8,0) (16,0.16) (32,0.188) (64,0.188) (128,0.32) (256,0.436) (512,0.788) (1024,1.112) (2048,1.348) (4096,1.82) (8192,2.32) (16384,2.98) (32768,3.788) (65536,4.704) (131072,5.784) (262144,7.12) (524288,8.224) (1.04858e+06,9.036) (2.09715e+06,9.916) (4.1943e+06,11.404) (8.38861e+06,14.712) (1.67772e+07,20.928) (3.35544e+07,24.86) (6.71089e+07,27.608) (1.34218e+08,30.084) (2.68435e+08,32.64) (5.36871e+08,34.86) (1.07374e+09,37.388) (2.14748e+09,39.94) (4.29497e+09,42.156) (8.58993e+09,44.352) (1.71799e+10,46.416) (3.43597e+10,48.228) (6.87195e+10,50.524) (1.37439e+11,52.768) (2.74878e+11,54.86) (5.49756e+11,57.276) (1.09951e+12,59.812) (2.19902e+12,62.084) (4.39805e+12,64.292) (8.79609e+12,66.924) (1.75922e+13,68.848) (3.51844e+13,71.236) (7.03687e+13,73.26) (1.40737e+14,74.624) (2.81475e+14,76.008) (5.6295e+14,77.424) (1.1259e+15,78.384) (2.2518e+15,79.428) (4.5036e+15,80.712) (9.0072e+15,81.5) (1e+16,81.62) };
    \end{axis}
  \end{tikzpicture}
\vspace{-.2cm}
 \caption{000webhost} \label{fig:000webhostattackerknowledge}
%\vspace{-.15cm}
\end{subfigure}
%\vspace{0.5em}
%\hfill
%\newline
%yahoo
\begin{subfigure}[htb]{0.29\textwidth}
%\vspace{-0.2cm}
      \begin{tikzpicture}[scale=0.6]
      \begin{axis}[
        title style={align=center},
        xlabel={\# of Guesses ($G$)},
        xmin = {1},
        xmax = {10^9},
        xmode = log,
        log basis x={10},
        xlabel shift = -3pt,
        ylabel={\% Cracked Passwords ($\lambda_G$)},
        ymax={100},
        ymin = {0},
        %ymode = log,
        %log basis y={2},
        %ylabel shift = -3pt,
        grid=major,
        %small,
        cycle list = { 
        {black}, {black,dashed}, 
        {cyan}, %{cyan, dashed}, 
        {magenta}, %{magenta,dashed}, 
        {green}, %{green, dashed}, 
        {orange}, %{orange, dashed}, 
        %{olive}, {olive, dashed}, 
        %{violet}, {violet, dashed}, 
        %{yellow}, {yellow, dashed}, 
        {red}, {red, dashed},
        {blue}, {blue, dashed},
        {teal}, {teal, dashed}
        },
        legend style = {font=\small}, %{font=\fontsize{3}{3}\selectfont},
        %{font=\tiny},
        legend pos = north west,
        legend entries = { 
        best ub, best lb, 
        k=$10^4$,
        k=$10^5$,
        k=$10^6$,
        k=$10^7$,
        k=$69276337$,
        }
      ]
      \addlegendimage{no markers, black}
      \addlegendimage{dashed, black}
      \addlegendimage{no markers, cyan}
      %\addlegendimage{dashed, cyan}
      \addlegendimage{no markers, magenta}
      %\addlegendimage{dashed, magenta}
      \addlegendimage{no markers, green}
      %\addlegendimage{dashed, green}
      \addlegendimage{no markers, orange}
      %\addlegendimage{dashed, orange}
      %\addlegendimage{no markers, olive}
      %\addlegendimage{dashed, olive}
      %\addlegendimage{no markers, violet}
      %\addlegendimage{dashed, violet}
      %\addlegendimage{no markers, yellow}
      %\addlegendimage{dashed, yellow}
      \addlegendimage{no markers, red}
      %\addlegendimage{dashed, red}
      %\addlegendimage{no markers, blue}
      %\addlegendimage{dashed, blue}
      %\addlegendimage{no markers, teal}
      %\addlegendimage{dashed, teal}
    
    %yahoo best upper bound
    \addplot coordinates {(1,1.1128) (2,1.32785) (4,1.4971) (8,1.79541) (16,2.13606) (32,2.54794) (64,3.13958) (128,3.99445) (256,5.14545) (512,6.60003) (1024,8.41074) (2048,10.5658) (4096,13.039) (8192,15.7479) (16384,18.7154) (32768,22.1076) (65536,26.1203) (131072,30.6748) (262144,35.55) (524288,40.4617) (1048576,45.3936) (2097152,50.6323) (4.1943e+06, 56.4314)
     (8.38861e+06, 60.6758)
     (1.67772e+07, 66.379)
     (3.35544e+07, 71.4755)
     (6.71089e+07, 77.7174)
     (1.34218e+08, 86.0332)
     (2.68435e+08, 97.505)
     (2.88435e+08, 98.9301)
     (3.08435e+08, 100.2) };
    
    %yahoo best lower bound
     \addplot coordinates {(1,0.115353) (2,0.347353) (4,0.539353) (8,0.843353) (16,1.19535) (32,1.64735) (64,2.17935) (128,3.04735) (256,4.22735) (512,5.58735) (1024,7.45535) (2048,9.76735) (4096,12.1114) (8192,14.7434) (16384,17.6194) (32768,20.9634) (65536,24.7874) (131072,29.4634) (262144,34.1434) (524288,38.5514) (1048576,42.9114) (2097152,46.5394) (4194304,49.5874) (8.38861e+06, 52.5008)
     (1.67772e+07, 54.8429)
     (3.35544e+07, 57.3192)
     (6.71089e+07, 61.7635)
     (7.71089e+07, 62.7236)
     (8.71089e+07, 63.5193)
     (9.71089e+07, 64.1023)
     (1.07109e+08, 64.4276)
     (1.17109e+08, 64.5013)
     (1.27109e+08, 64.5027)
     (1.34218e+08, 64.5037)
     (2.68435e+08, 64.5231)
     (5.36871e+08, 64.5618)
     (1.07374e+09, 64.6393)
     (2.14748e+09, 64.7942)
     (4.29497e+09, 64.9779)
     (5.49756e+11, 64.9779) };
    
    %AttackerKnowledgeComp-d=25000-yahoo_plots.txtAttackerKnowledgeComparison
%N = 69301337  d = 25000  distinct = 33895873
%k = 10000
\addplot coordinates {(1,1.076) (2,1.268) (4,1.424) (8,1.708) (16,1.972) (32,2.28) (64,2.68) (128,3.14) (256,3.72) (512,3.8) (1024,4.028) (2048,4.428) (4096,5.036) (8192,6.408) (9508,6.864) };
%k = 100000
\addplot coordinates {(1,1.076) (2,1.268) (4,1.404) (8,1.704) (16,2.048) (32,2.516) (64,3.06) (128,3.932) (256,4.972) (512,6.364) (1024,7.844) (2048,9.4) (4096,11.188) (8192,11.424) (16384,11.92) (32768,12.924) (65536,14.808) (88072,16.1) };
%k = 1000000
\addplot coordinates {(1,1.076) (2,1.268) (4,1.44) (8,1.704) (16,2.06) (32,2.496) (64,3.08) (128,3.956) (256,5.18) (512,6.624) (1024,8.428) (2048,10.64) (4096,12.868) (8192,15.44) (16384,17.84) (32768,20.164) (65536,22.948) (131072,23.552) (262144,24.792) (524288,27.348) (763247,29.608) };
%k = 10000000
\addplot coordinates {(1,1.076) (2,1.268) (4,1.436) (8,1.704) (16,2.06) (32,2.484) (64,3.116) (128,3.908) (256,5.156) (512,6.684) (1024,8.532) (2048,10.72) (4096,13.18) (8192,15.808) (16384,18.604) (32768,21.788) (65536,25.764) (131072,29.884) (262144,34.248) (524288,37.784) (1048576,39.364) (2097152,40.464) (4194304,42.808) (6135480,45.116) };
%k = 69276337
\addplot coordinates {(1,1.076) (2,1.268) (4,1.4) (8,1.704) (16,2.06) (32,2.488) (64,3.084) (128,3.92) (256,5.16) (512,6.692) (1024,8.584) (2048,10.76) (4096,13.168) (8192,15.78) (16384,18.8) (32768,22.116) (65536,26.052) (131072,30.364) (262144,35.312) (524288,39.836) (1048576,43.952) (2097152,47.712) (4194304,50.596) (8388608,51.824) (16777216,53.724) (33554432,57.452) (33885254,57.524) };

    \end{axis}
  \end{tikzpicture}

\vspace{-.2cm}
 \caption{Yahoo} \label{fig:yahooattackerknowledge}
%\vspace{-.15cm}
\end{subfigure}
%\vspace{0.5em}
%\hfill
%\
%rockyou
\begin{subfigure}[htb]{0.33\textwidth}
%\vspace{-0.2cm}
      \begin{tikzpicture}[scale=0.6]
      \begin{axis}[
        title style={align=center},
        xlabel={\# of Guesses ($G$)},
        xmin = {1},
        xmax = {10^9},
        xmode = log,
        log basis x={10},
        xlabel shift = -3pt,
        ylabel={\% Cracked Passwords ($\lambda_G$)},
        ymax={100},
        ymin = {0},
        %ymode = log,
        %log basis y={2},
        %ylabel shift = -3pt,
        grid=major,
        %small,
        cycle list = { 
        {black}, {black,dashed}, 
        {cyan}, %{cyan, dashed}, 
        {magenta}, %{magenta,dashed}, 
        {green}, %{green, dashed}, 
        {orange}, %{orange, dashed}, 
        %{olive}, {olive, dashed}, 
        %{violet}, {violet, dashed}, 
        %{yellow}, {yellow, dashed}, 
        {red}, {red, dashed},
        {blue}, {blue, dashed},
        {teal}, {teal, dashed}
        },
        legend style = {font=\small}, %{font=\fontsize{3}{3}\selectfont},
        %{font=\tiny},
        legend pos = north west,
        legend entries = { 
        best ub, best lb, 
        k=$10^4$,
        k=$10^5$,
        k=$10^6$,
        k=$10^7$,
        k=$32578388$,
        }
      ]
      \addlegendimage{no markers, black}
      \addlegendimage{dashed, black}
      \addlegendimage{no markers, cyan}
      %\addlegendimage{dashed, cyan}
      \addlegendimage{no markers, magenta}
      %\addlegendimage{dashed, magenta}
      \addlegendimage{no markers, green}
      %\addlegendimage{dashed, green}
      \addlegendimage{no markers, orange}
      %\addlegendimage{dashed, orange}
      %\addlegendimage{no markers, olive}
      %\addlegendimage{dashed, olive}
      %\addlegendimage{no markers, violet}
      %\addlegendimage{dashed, violet}
      %\addlegendimage{no markers, yellow}
      %\addlegendimage{dashed, yellow}
      \addlegendimage{no markers, red}
      %\addlegendimage{dashed, red}
      %\addlegendimage{no markers, blue}
      %\addlegendimage{dashed, blue}
      %\addlegendimage{no markers, teal}
      %\addlegendimage{dashed, teal}
    
    %rockyou best upper bound
    \addplot coordinates {(1,0.929511) (2,1.17205) (4,1.58996) (8,1.97602) (16,2.36573) (32,2.96563) (64,3.8348) (128,5.06415) (256,6.71577) (512,8.83802) (1024,11.4333) (2048,14.4886) (4096,17.8038) (8192,21.2952) (16384,25.0103) (32768,29.1282) (65536,33.7334) (131072,38.6622) (262144,43.7835) (524288,49.0365) (1048576,54.6604)  (2.09715e+06, 61.0338)
     (4.1943e+06, 65.844)
     (8.38861e+06, 72.5909)
     (1.67772e+07, 79.3481)
     (3.35544e+07, 86.8013)
     (6.71089e+07, 96.6657)
     (7.71089e+07, 99.0158)
     (8.71089e+07, 100.201) };

    %rockyou best lower bound
    \addplot coordinates {(1,-0.0766468) (2,0.207353) (4,0.547353) (8,0.915353) (16,1.32735) (32,1.88735) (64,2.80735) (128,4.01535) (256,5.60335) (512,7.68735) (1024,10.2714) (2048,13.3794) (4096,16.5794) (8192,19.9794) (16384,23.7954) (32768,27.7274) (65536,32.2514) (131072,36.9834) (262144,41.9554) (524288,46.6634) (1048576,50.6474) (2097152,53.6474)  (4.1943e+06, 56.458)
     (8.38861e+06, 60.3043)
     (1.67772e+07, 64.4719)
     (3.35544e+07, 71.0155)
     (4.35544e+07, 72.6686)
     (5.35544e+07, 72.7036)
     (6.35544e+07, 72.7067)
     (6.71089e+07, 72.7078)
     (1.34218e+08, 72.7283)
     (2.68435e+08, 72.7695)
     (5.36871e+08, 72.8518)
     (1.07374e+09, 73.0165)
     (2.14748e+09, 73.075)};
   
    %AttackerKnowledgeComp-d=25000-rockyou_plots.txtAttackerKnowledgeComparison
%N = 32603388  d = 25000  distinct = 14344391
%k = 10000
\addplot coordinates {(1,0.844) (2,1.088) (4,1.484) (8,1.892) (16,2.204) (32,2.62) (64,3.352) (128,4.3) (256,5.152) (512,6.008) (1024,6.22) (2048,6.776) (4096,7.672) (8192,9.728) (9261,10.24) };
%k = 100000
\addplot coordinates {(1,0.844) (2,1.108) (4,1.588) (8,1.94) (16,2.392) (32,3.056) (64,3.86) (128,5.132) (256,6.432) (512,8.42) (1024,10.644) (2048,13.176) (4096,15.336) (8192,16.74) (16384,17.248) (32768,18.628) (65536,21.128) (82579,22.492) };
%k = 1000000
\addplot coordinates {(1,0.844) (2,1.108) (4,1.588) (8,1.992) (16,2.396) (32,3.052) (64,4.008) (128,5.116) (256,6.636) (512,8.72) (1024,11.252) (2048,14.16) (4096,17.588) (8192,20.8) (16384,24.116) (32768,27.348) (65536,30.556) (131072,31.576) (262144,32.996) (524288,36.1) (688035,37.9) };
%k = 10000000
\addplot coordinates {(1,0.844) (2,1.108) (4,1.588) (8,1.992) (16,2.456) (32,3.068) (64,4) (128,5.148) (256,6.624) (512,8.696) (1024,11.26) (2048,14.336) (4096,17.632) (8192,21.116) (16384,24.96) (32768,28.86) (65536,33.232) (131072,37.848) (262144,42.08) (524288,45.548) (1048576,47.964) (2097152,49.624) (4194304,52.848) (5264363,54.74) };
%k = 32578388
\addplot coordinates {(1,0.844) (2,1.108) (4,1.588) (8,1.992) (16,2.452) (32,3.084) (64,4.028) (128,5.152) (256,6.628) (512,8.724) (1024,11.272) (2048,14.304) (4096,17.652) (8192,21.22) (16384,24.98) (32768,28.952) (65536,33.408) (131072,38.264) (262144,43.096) (524288,47.668) (1048576,51.524) (2097152,54.528) (4194304,56.776) (8388608,59.712) (14335347,63.82) };

    \end{axis}
  \end{tikzpicture}

\vspace{-.2cm}
 \caption{RockYou} \label{fig:rockyouattackerknowledge}
%\vspace{-.15cm}
\end{subfigure}
%\vspace{-0.2cm}
\caption{Attacker's Knowledge (000webhost, Yahoo!, RockYou)}
\label{fig:attackerknowledge-app}
%\vspace{-0.4cm}
\end{figure*}
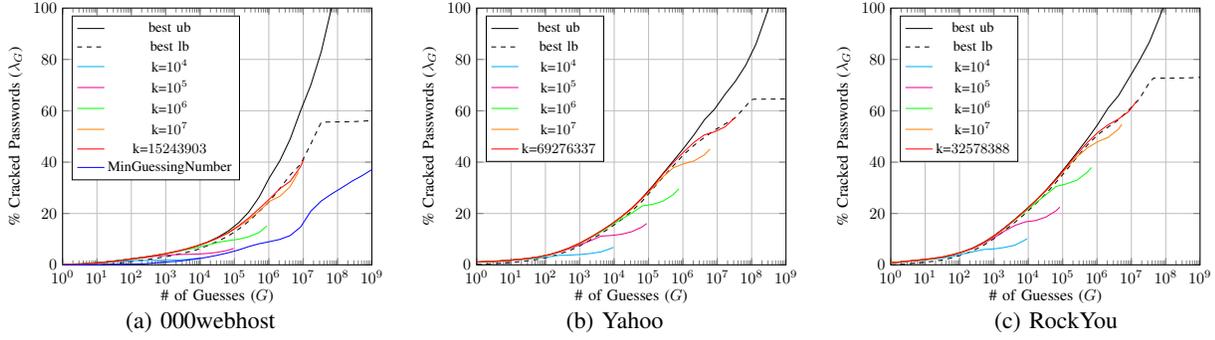
}

\fullversion{}{
Table~\ref{table:zipfparameter} and Figure~\ref{fig:zipf} show the Zipf's Law parameters and the compraison results.

\begin{table}[!t]
  \caption{CDF Zipf's Law Parameters~\cite{wang2016implications,SP:BloHarZho18}}
  \label{table:zipfparameter}
  \centering
  \begin{tabular}{ccc}
    \toprule
    \textbf{Dataset ($S$)} & y & r \\
    \midrule
    Yahoo!~\cite{NDSS:BloDatBon16} & 0.03315 & 0.1811 \\ 
    RockYou~\cite{rockyoudataset}  & 0.037433 & 0.187227 \\ 
    \fullversion{}{000webhost~\cite{000webhostdataset} & 0.005858& 0.281557 \\ }
    \fullversion{}{Battlefield Heroes~\cite{bfieldsdataset} & 0.010298 & 0.294932  \\} 
    CSDN~\cite{csdndataset} & 0.058799 & 0.148573  \\    
  \bottomrule
\end{tabular}
\end{table}

\input{Plots/zipf}
}

\fullversion{}{
Figure~\ref{fig:hashcost-app} shows the upper and lower bounds of tuning password hash cost parameters of 000webhost, Yahoo!, and RockYou datasets, similar to Figure~\ref{fig:hashcostneopets} in Section~\ref{sec:application_hashcost}.

%%%%attacker's knowledge
\begin{figure*}
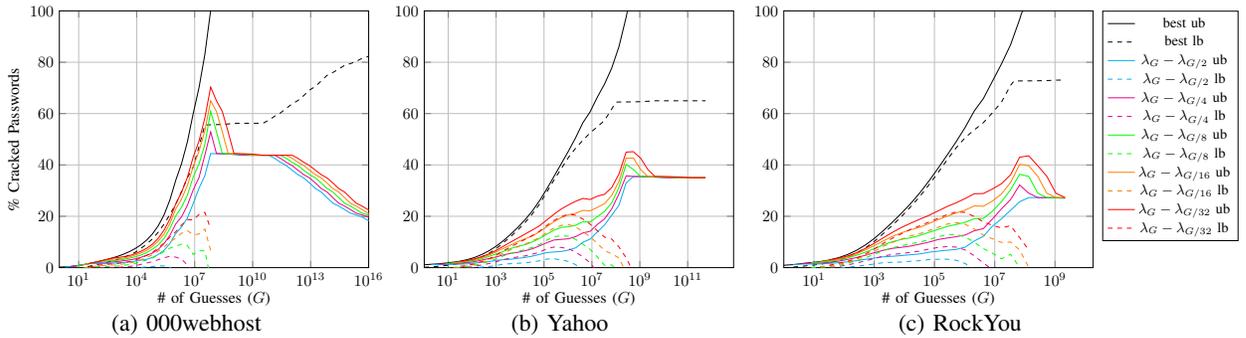
\centering
%neopets
%\input{Plots/attackerknowledge-neopets}
%\vspace{0.5em}
%000wehost
\input{Plots/hashcost-000webhost}
%\vspace{0.5em}
%\hfill
%\newline
%yahoo
\input{Plots/hashcost-yahoo}
%\vspace{0.5em}
%\hfill
%\
%rockyou
\begin{subfigure}[htb]{0.31\textwidth}
%\vspace{-0.2cm}
      \begin{tikzpicture}[scale=0.6]
      \begin{axis}[
        title style={align=center},
        xlabel={\# of Guesses ($G$)},
        xmin = {1},
        xmode = log,
        log basis x={10},
        xlabel shift = -3pt,
        ymax={100},
        ymin = {0},
        %ymode = log,
        %log basis y={2},
        %ylabel shift = -3pt,
        grid=major,
        %small,
        cycle list = { 
        {black}, {black,dashed}, 
        {cyan}, {cyan, dashed}, 
        {magenta}, {magenta,dashed}, 
        {green}, {green, dashed}, 
        {orange}, {orange, dashed}, 
        %{olive}, {olive, dashed}, 
        %{violet}, {violet, dashed}, 
        %{yellow}, {yellow, dashed}, 
        {red}, {red, dashed},
        %{blue}, {blue, dashed},
        %{teal}, {teal, dashed}
        },
        legend style = {font=\small}, %{font=\fontsize{3}{3}\selectfont},
        %{font=\tiny},
        legend pos = outer north east,
        legend entries = { 
        best ub, best lb, 
        $\lambda_G-\lambda_{G/2}$ ub, $\lambda_G-\lambda_{G/2}$ lb,
        $\lambda_G-\lambda_{G/4}$ ub, 
        $\lambda_G-\lambda_{G/4}$ lb,
        $\lambda_G-\lambda_{G/8}$ ub,
        $\lambda_G-\lambda_{G/8}$ lb,
        $\lambda_G-\lambda_{G/16}$ ub,
        $\lambda_G-\lambda_{G/16}$ lb,
        $\lambda_G-\lambda_{G/32}$ ub,
        $\lambda_G-\lambda_{G/32}$ lb,
        }
      ]
      \addlegendimage{no markers, black}
      \addlegendimage{dashed, black}
      \addlegendimage{no markers, cyan}
      \addlegendimage{dashed, cyan}
      \addlegendimage{no markers, magenta}
      \addlegendimage{dashed, magenta}
      \addlegendimage{no markers, green}
      \addlegendimage{dashed, green}
      \addlegendimage{no markers, orange}
      \addlegendimage{dashed, orange}
      %\addlegendimage{no markers, olive}
      %\addlegendimage{dashed, olive}
      %\addlegendimage{no markers, violet}
      %\addlegendimage{dashed, violet}
      %\addlegendimage{no markers, yellow}
      %\addlegendimage{dashed, yellow}
      \addlegendimage{no markers, red}
      \addlegendimage{dashed, red}
      %\addlegendimage{no markers, blue}
      %\addlegendimage{dashed, blue}
      %\addlegendimage{no markers, teal}
      %\addlegendimage{dashed, teal}
    
    %rockyou best upper bound
    \addplot coordinates {(1,0.929511) (2,1.17205) (4,1.58996) (8,1.97602) (16,2.36573) (32,2.96563) (64,3.8348) (128,5.06415) (256,6.71577) (512,8.83802) (1024,11.4333) (2048,14.4886) (4096,17.8038) (8192,21.2952) (16384,25.0103) (32768,29.1282) (65536,33.7334) (131072,38.6622) (262144,43.7835) (524288,49.0365) (1048576,54.6604)  (2.09715e+06, 61.0338)
     (4.1943e+06, 65.844)
     (8.38861e+06, 72.5909)
     (1.67772e+07, 79.3481)
     (3.35544e+07, 86.8013)
     (6.71089e+07, 96.6657)
     (7.71089e+07, 99.0158)
     (8.71089e+07, 100.201) };

    %rockyou best lower bound
    \addplot coordinates {(1,-0.0766468) (2,0.207353) (4,0.547353) (8,0.915353) (16,1.32735) (32,1.88735) (64,2.80735) (128,4.01535) (256,5.60335) (512,7.68735) (1024,10.2714) (2048,13.3794) (4096,16.5794) (8192,19.9794) (16384,23.7954) (32768,27.7274) (65536,32.2514) (131072,36.9834) (262144,41.9554) (524288,46.6634) (1048576,50.6474) (2097152,53.6474)  (4.1943e+06, 56.458)
     (8.38861e+06, 60.3043)
     (1.67772e+07, 64.4719)
     (3.35544e+07, 71.0155)
     (4.35544e+07, 72.6686)
     (5.35544e+07, 72.7036)
     (6.35544e+07, 72.7067)
     (6.71089e+07, 72.7078)
     (1.34218e+08, 72.7283)
     (2.68435e+08, 72.7695)
     (5.36871e+08, 72.8518)
     (1.07374e+09, 73.0165)
     (2.14748e+09, 73.075)};
   
    %upperbound of lambda_G-lambda_{G/b}, b = 2
\addplot coordinates { (2,1.17205) (4,1.38261) (8,1.42867) (16,1.45038) (32,1.63828) (64,1.94745) (128,2.2568) (256,2.70042) (512,3.23467) (1024,3.74595) (2048,4.2172) (4096,4.4244) (8192,4.7158) (16384,5.0309) (32768,5.3328) (65536,6.006) (131072,6.4108) (262144,6.8001) (524288,7.0811) (1.04858e+06,7.997) (2.09715e+06,10.3864) (4.1943e+06,12.1966) (8.38861e+06,16.1329) (1.67772e+07,19.0438) (3.35544e+07,22.3294) (6.71089e+07,25.6502) (1.34218e+08,27.2922) (2.68435e+08,27.2717) (5.36871e+08,27.2305) (1.07374e+09,27.1482) (2.14748e+09,26.9835)};
%lowerbound of lambda_G-lambda_{G/b}, b = 2
\addplot coordinates { (2,-0.722158) (4,-0.624697) (8,-0.674607) (16,-0.64867) (32,-0.47838) (64,-0.15828) (128,0.18055) (256,0.5392) (512,0.97158) (1024,1.43338) (2048,1.9461) (4096,2.0908) (8192,2.1756) (16384,2.5002) (32768,2.7171) (65536,3.1232) (131072,3.25) (262144,3.2932) (524288,2.8799) (1.04858e+06,1.6109) (2.09715e+06,-1.013) (4.1943e+06,-4.5758) (8.38861e+06,-5.5397) (1.67772e+07,-8.119) (3.35544e+07,-8.3326) (6.71089e+07,-14.0935) (1.34218e+08,-23.9374)};
%upperbound of lambda_G-lambda_{G/b}, b = 4
\addplot coordinates { (4,1.58996) (8,1.76867) (16,1.81838) (32,2.05028) (64,2.50745) (128,3.1768) (256,3.90842) (512,4.82267) (1024,5.82995) (2048,6.80125) (4096,7.5324) (8192,7.9158) (16384,8.4309) (32768,9.1488) (65536,9.938) (131072,10.9348) (262144,11.5321) (524288,12.0531) (1.04858e+06,12.705) (2.09715e+06,14.3704) (4.1943e+06,15.1966) (8.38861e+06,18.9435) (1.67772e+07,22.8901) (3.35544e+07,26.497) (6.71089e+07,32.1938) (1.34218e+08,28.9845) (2.68435e+08,27.2922) (5.36871e+08,27.2717) (1.07374e+09,27.2305) (2.14748e+09,27.1482)};
%lowerbound of lambda_G-lambda_{G/b}, b = 4
\addplot coordinates { (4,-0.382158) (8,-0.256697) (16,-0.26261) (32,-0.08867) (64,0.44162) (128,1.04972) (256,1.76855) (512,2.6232) (1024,3.55563) (2048,4.54138) (4096,5.1461) (8192,5.4908) (16384,5.9916) (32768,6.4322) (65536,7.2411) (131072,7.8552) (262144,8.222) (524288,8.0012) (1.04858e+06,6.8639) (2.09715e+06,4.6109) (4.1943e+06,1.7976) (8.38861e+06,-0.7295) (1.67772e+07,-1.3721) (3.35544e+07,-1.5754) (6.71089e+07,-6.6403) (1.34218e+08,-14.073)};
%upperbound of lambda_G-lambda_{G/b}, b = 8
\addplot coordinates { (8,1.97602) (16,2.15838) (32,2.41828) (64,2.91945) (128,3.7368) (256,4.82842) (512,6.03067) (1024,7.41795) (2048,8.88525) (4096,10.1164) (8192,11.0238) (16384,11.6309) (32768,12.5488) (65536,13.754) (131072,14.8668) (262144,16.0561) (524288,16.7851) (1.04858e+06,17.677) (2.09715e+06,19.0784) (4.1943e+06,19.1806) (8.38861e+06,21.9435) (1.67772e+07,25.7007) (3.35544e+07,30.3433) (6.71089e+07,36.3614) (1.34218e+08,35.5281) (2.68435e+08,28.9845) (5.36871e+08,27.2922) (1.07374e+09,27.2717) (2.14748e+09,27.2305)};
%lowerbound of lambda_G-lambda_{G/b}, b = 8
\addplot coordinates { (8,-0.014158) (16,0.1553) (32,0.29739) (64,0.83133) (128,1.64962) (256,2.63772) (512,3.85255) (1024,5.20725) (2048,6.66363) (4096,7.74138) (8192,8.5461) (16384,9.3068) (32768,9.9236) (65536,10.9562) (131072,11.9731) (262144,12.8272) (524288,12.93) (1.04858e+06,11.9852) (2.09715e+06,9.8639) (4.1943e+06,7.4215) (8.38861e+06,5.6439) (1.67772e+07,3.4381) (3.35544e+07,5.1715) (6.71089e+07,0.1169) (1.34218e+08,-6.6198)};
%upperbound of lambda_G-lambda_{G/b}, b = 16
\addplot coordinates { (16,2.36573) (32,2.75828) (64,3.28745) (128,4.1488) (256,5.38842) (512,6.95067) (1024,8.62595) (2048,10.4733) (4096,12.2005) (8192,13.6079) (16384,14.7389) (32768,15.7488) (65536,17.154) (131072,18.6828) (262144,19.9881) (524288,21.3091) (1.04858e+06,22.409) (2.09715e+06,24.0504) (4.1943e+06,23.8886) (8.38861e+06,25.9275) (1.67772e+07,28.7007) (3.35544e+07,33.1539) (6.71089e+07,40.2077) (1.34218e+08,39.6957) (2.68435e+08,35.5281) (5.36871e+08,28.9845) (1.07374e+09,27.2922) (2.14748e+09,27.2717)};
%lowerbound of lambda_G-lambda_{G/b}, b = 16
\addplot coordinates { (16,0.397839) (32,0.7153) (64,1.21739) (128,2.03933) (256,3.23762) (512,4.72172) (1024,6.4366) (2048,8.31525) (4096,9.86363) (8192,11.1414) (16384,12.3621) (32768,13.2388) (65536,14.4476) (131072,15.6882) (262144,16.9451) (524288,17.5352) (1.04858e+06,16.914) (2.09715e+06,14.9852) (4.1943e+06,12.6745) (8.38861e+06,11.2678) (1.67772e+07,9.8115) (3.35544e+07,9.9817) (6.71089e+07,6.8638) (1.34218e+08,0.1374)};
%upperbound of lambda_G-lambda_{G/b}, b = 32
\addplot coordinates { (32,2.96563) (64,3.62745) (128,4.5168) (256,5.80042) (512,7.51067) (1024,9.54595) (2048,11.6813) (4096,13.7884) (8192,15.6919) (16384,17.3229) (32768,18.8568) (65536,20.354) (131072,22.0828) (262144,23.8041) (524288,25.2411) (1.04858e+06,26.933) (2.09715e+06,28.7824) (4.1943e+06,28.8606) (8.38861e+06,30.6355) (1.67772e+07,32.6847) (3.35544e+07,36.1539) (6.71089e+07,43.0183) (1.34218e+08,43.542) (2.68435e+08,39.6957) (5.36871e+08,35.5281) (1.07374e+09,28.9845) (2.14748e+09,27.2922)};
%lowerbound of lambda_G-lambda_{G/b}, b = 32
\addplot coordinates { (32,0.957839) (64,1.6353) (128,2.42539) (256,3.62733) (512,5.32162) (1024,7.30577) (2048,9.5446) (4096,11.5152) (8192,13.2636) (16384,14.9574) (32768,16.2941) (65536,17.7628) (131072,19.1796) (262144,20.6602) (524288,21.6531) (1.04858e+06,21.5192) (2.09715e+06,19.914) (4.1943e+06,17.7958) (8.38861e+06,16.5208) (1.67772e+07,15.4354) (3.35544e+07,16.3551) (6.71089e+07,11.674) (1.34218e+08,6.8843)};

    \end{axis}
  \end{tikzpicture}

\vspace{-.2cm}
 \caption{RockYou} \label{fig:hashcostrocyou}
%\vspace{-.15cm}
\end{subfigure}
%\vspace{-0.2cm}
\caption{Tuning Password Hash Cost Parameters (000webhost, Yahoo!, RockYou)}
\label{fig:hashcost-app}
%\vspace{-0.4cm}
\end{figure*}
}

\fullversion{}{
Figure~\ref{fig:PCPfull} shows the full plots of the upper/lower bounds of different PCPs including count-min sketch (CMK) with different threshold parameters $r\in \{10^4, 10^5, 10^6\}$.

%%%%PCP comparison full
\begin{figure*}\centering
%rockyou
\input{Plots/PCPfull-rockyou}
%\vspace{0.5em}
%000webhost
\input{Plots/PCPfull-000webhost}
%\vspace{0.5em}
%\hfill
%\newline
%neopets
\input{Plots/PCPfull-neopets}
%\vspace{0.5em}
%\hfill
%\newline
%\vspace{-0.2cm}
\caption{Comparison of Password Composition Policies (full plots)}
\label{fig:PCPfull}
%\vspace{-0.4cm}
\end{figure*}
}
\section{Parameter Setting}\label{app:parametersetting}

To generate high-confidence theoretical bounds of $\lambda_G$ and $\lambda(S,G)$, we assign the parameters values to guarantee every bound in Section~\ref{sec:theoreticalanalysis} holds with probability {\em at least} $0.99$.

Recall that the parameter $\epsilon$ in  Corollary~\ref{crl:upperbound-priorwork}, Theorem~\ref{thm:bound1}, Corollary~\ref{crl:bound2}, Theorem~\ref{thm:lowerbound-priorwork}\fullversion{}{, and $\epsilon_1$ in Corollary~\ref{crl:bound3final-lower}} bound the difference between $\lambda_G$ and $\lambda(S,G)$. We fix this parameter as $\epsilon=\epsilon_1=\left( \frac{\ln(\delta_1)}{-2N}\right)^{\frac{1}{2}}$ where we set $\delta_1 = 0.00009$ for all bounds of $\lambda_G$ and $\lambda(S,G)$. Then the upper bound in Corollary~\ref{crl:upperbound-priorwork} ($\mathtt{FrequencyUB}(S,G)$) holds with probability at least $1-\delta_1\geq 0.99$.

We denote $\delta_{2,j} = \exp\left(\frac{-2t^2}{Nj^2}\right)$ in Theorem~\ref{thm:lowerbound-priorwork}. We set $\delta_{2,j} = 0.01 - \delta_1$  (i.e. $t = \left(\frac{Nj^2\ln(\delta_{2,j})}{-2}\right)^{\frac{1}{2}}$) such that the lower bound in Theorem~\ref{thm:lowerbound-priorwork}  ($\mathtt{PriorLB(S,G,j)}$) holds with probability at least $1-\delta_1-\delta_{2,j}\geq 0.99$ for any $j\geq 2$.

We denote $\delta_{3} = \exp\left(\frac{-2t^2}{d}\right)$ in both Theorem~\ref{thm:bound1} and  Corollary~\ref{crl:bound2}. We set $d=25000$ and $\delta_3 = 0.01 - \delta_1$ (i.e. $t = \left( \frac{d\ln(\delta_3)}{-2}\right)^{\frac{1}{2}}$) such that the lower bound in Corollary~\ref{crl:bound2} ($\mathtt{ExtendedLB}(S,G)$) extends the lower bound in Theorem~\ref{thm:bound1} ($\mathtt{SamplingLB}(S,G)$), and both of them hold with probability at least $1-\delta_3\geq 0.99$.

We denote $\delta_{4,i} = \exp\left(\frac{-2(N-i)^2\epsilon_{2,i}^2}{N(i+1)^2}\right)$ in Theorem~\ref{thm:bound3final-lower}. we set $q=1.002$, $i'=4$, $\{\delta_{4,0}, \delta_{4,1}, \delta_{4,2}, \delta_{4,3}, \delta_{4,4}\} = \{0.00009, 0.000165, 0.00175,$ $0.00175, 0.0012\}$, and $\{\hat{x}_{\epsilon_{3,0}}, \hat{x}_{\epsilon_{3,1}}, \hat{x}_{\epsilon_{3,2}}, \hat{x}_{\epsilon_{3,3}}, \hat{x}_{\epsilon_{3,4}}\} = \{7.0/N, 11.0/N,$ $14.0/N, 16.3/N, 18.5/N\}$. Note that $\epsilon_{2,i} =\left(\frac{N(i+1)^2\ln(\delta_{4,i})}{-2(N-i)^2}\right)^{\frac{1}{2}}$ and $\epsilon_{3,i} = \frac{1}{q^{i+1}}\left(\frac{1-\hat{x}_{\epsilon_{3,i}}}{1-q\hat{x}_{\epsilon_{3,i}}}\right)^{N-i}-1$ for any $0\leq i\leq i'$. Both of the upper and lower bounds in Theorem~\ref{thm:bound3final-lower} ($\mathtt{LPUB}(S,G)$ and $\mathtt{LPLB}(S,G)$) hold with probabilities at least $1-\sum_{i=0}^{i'}\delta_{4,i}\geq 0.99$.

\fullversion{}{\section{Comparing LP1 with LPlower and LPupper}
In Section~\ref{sec:LPbounds} we first presented a linear program $\mathtt{LP1}$ to upper/lower bound $\lambda_G$ under the idealized assumption that the real probability distribution is consistent with a fixed probability mesh. We later showed how to eliminate this idealized assumption by adding a small slack terms in our constraints. We compare the upper/lower bounds derived using $\mathtt{LP1}$ (i.e., under idealized assumptions) with the upper/lower bounds generated using $\mathtt{LPlower}$ and $\mathtt{LPupper}$ respectively using the same parameter settings for all datasets. We find that the bounds are very close indicating that the small slack terms we added to eliminate the idealized assumption does not negatively impact the quality of the upper/lower bounds. The comparison results are shown in Figure~\ref{fig:LPcomp}.% in Appendix~\ref{app:figures}.

\input{Plots/LPcomp}}
\fullversion{}{\section{Bounding $\lambda(S,G)$}\label{app:lambdaSG}
By Theorem \ref{thm:ExpectationConcentration} the value $\lambda(S,G)$ is tightly concentrated around its mean $\mathbb{E}[\lambda(S,G)]=\lambda_G$. Thus, any high confidence upper/lower bound on $\lambda_G$ immediately yields a high confidence upper/lower bound on $\lambda(S,G)$ and vice versa. We formally state each corollary below.

Corollary \ref{crl:bound1} is a corollary of Theorem \ref{thm:bound1}.

\begin{corollary}\label{crl:bound1}
For any $G \geq 1$ and any parameters $0<d<N$, $0\leq\epsilon \leq 1$, $t\geq 0$, we have:
\begin{align*}
    &\Pr[\lambda(S,G)\geq \frac{1}{d}(h(D_1,D_2,G) - t) - \epsilon] \\
    \geq & 1 - \exp\left(-2N\epsilon^2\right) - \exp\left(\frac{-2t^2}{d}\right)
\end{align*}
where the randomness is taken over the samples $D_1 \leftarrow \mathcal{P}^{N-d}$ and $D_2 \leftarrow \mathcal{P}^d$. 
\end{corollary}

\begin{proof}
Using Theorem~\ref{thm:ExpectationConcentration} we have $\Pr[\lambda(S,G) \geq \lambda_G - \epsilon] \geq 1 - \exp\left(-2N\epsilon^2\right)$. Then using Theorem~\ref{thm:bound1}, for any $0\leq \epsilon\leq 1$ and any $t\geq 0$ we have:
\begin{align*}
    &\Pr[\lambda(S,G) \geq \frac{1}{d}(h(D_2,G) - t) - \epsilon] \\
\geq & \Pr[\lambda(S,G) \geq \lambda_G - \epsilon \bigwedge \lambda_G \geq \frac{1}{d}(h(D_2,G) - t)] \\
\geq & \Pr[\lambda(S,G) \geq \lambda_G - \epsilon] - \Pr[\lambda_G < \frac{1}{d}(h(D_2,G) - t)] \\
\geq & 1 - \exp\left(-2N\epsilon^2\right) - \exp\left(\frac{-2t^2}{d}\right).
\end{align*}

Therefore, for any guessing number $G>0$ and any parameter $d,\epsilon_1,t_2$, we can lower bound the percentage of cracked passwords as below:
\begin{align*}
   \lambda(S,G) \geq \frac{1}{d}(h(D_2,G) - t) - \epsilon 
\end{align*}
where the left inequality holds with probability at least $ 1 - \exp\left(-2N\epsilon^2\right) - \exp\left(\frac{-2t^2}{d}\right)$.
\end{proof}

Corollary \ref{crl:generalmodelLBound} follows from Theorem \ref{thm:bound2} by applying Theorem~\ref{thm:ExpectationConcentration}. 

\begin{corollary}\label{crl:generalmodelLBound}
For any guessing number $G>0$ and any parameters $0<d<N$, $0\leq \epsilon\leq 1$, $t\geq 0$, we have
\begin{align*}
    &\Pr[\lambda(S,G) \geq \frac{1}{d}(h'_M(D_1,D_2,G) - t) - \epsilon] \\
    \geq & 1 - \exp\left(-2N\epsilon^{2}\right) - \exp\left(\frac{-2t^2}{d}\right)
\end{align*}
where the randomness is taken over the sample $S \leftarrow \mathcal{P}^N$ of size $N$.
\label{cor:modelLbound}
\end{corollary}
\begin{proof}
Using Theorem~\ref{thm:BoundedDifferencesInequality}, for any $0\leq \epsilon\leq 1$ we have $\Pr[\lambda(S,G) \geq \sum_{i\leq G}p_i - \epsilon] \geq 1 - \exp\left(-2N\epsilon^{2}\right)$. Then using Theorem~\ref{thm:bound2}, for any $0\leq \epsilon\leq 1,t\geq 0$ we have:
\begin{align*}
    &\Pr[\lambda(S,G) \geq \frac{1}{d}(h'_M(D_1,D_2,G) - t) - \epsilon] \\
    \geq & 1 - \exp\left(-2N\epsilon^{2}\right) - \exp\left(\frac{-2t^2}{d}\right).
\end{align*}
\end{proof}

Corollary \ref{crl:bound2-lambdaSG} follows from Corollary \ref{crl:bound2} by applying Theorem~\ref{thm:ExpectationConcentration}. 

\begin{corollary}\label{crl:bound2-lambdaSG}
Let $M$ be a password cracking model and $M^*$ be the corresponding hybrid attack model. Let parameters $G, d>0$ , $t >0$ and $\epsilon > 0$ be given then
\begin{align*}
&\Pr[\lambda(S,G) \geq \frac{1}{d}(h'_{M^*}(D_1,D_2,G) - t) - \epsilon] \\
&\geq 1 - \exp\left(-2N\epsilon^{2}\right) - \exp\left(\frac{-2t^2}{d}\right)
\end{align*}
where the randomness is taken over the set $S$ of size $N$.
\end{corollary}

Corollary \ref{crl:bound3ideal} follows from Theorem \ref{thm:bound3ideal} by applying Theorem~\ref{thm:ExpectationConcentration}. 
\begin{corollary}\label{crl:bound3ideal}
Given a probability distribution $\mathcal{P}$ which is consistent with a finite mesh $X_l = \{x_1,\ldots, x_l\}$ for any $G\geq 0$, any $i'\geq 0$, any $1 \leq \epsilon_1 > 0$ and any $\mathbf{\epsilon_2}=\{\epsilon_{2,0},\cdots,\epsilon_{2,i'}\}\in [0,1]^{i'+1}$, we have:
\begingroup\makeatletter\def\f@size{9}\check@mathfonts
\begin{align*}
    &\Pr\left[ \lambda(S,G) \geq \min\limits_{1\leq idx \leq l} \mathtt{LP1}(X_l,F^S,idx,G,1,i',\mathbf{\epsilon_2}) \right] \geq 1-\delta \\ 
    &\Pr\left[\lambda(S,G) \leq \max\limits_{1\leq idx \leq l} |\mathtt{LP1}(X_l,F^S,idx,G,-1,i',\mathbf{\epsilon_2})|\right] \geq 1-\delta
\end{align*}
\endgroup
where $\delta= \exp\left(-2N\epsilon_1^2\right) + 2\times\sum_{0\leq i\leq i'}\exp\left(\frac{-2(N-i)^2\epsilon_{2,i}^2}{N(i+1)^2}\right)$. The randomness is taken over the sample $S \leftarrow \mathcal{P}^N$ of size $N$.
\end{corollary}

Corollary \ref{crl:bound3mid} follows from Theorem \ref{thm:bound3mid}  by applying Theorem~\ref{thm:ExpectationConcentration}.

\begin{corollary}\label{crl:bound3mid}
Suppose that the probability distribution $\mathcal{P}$ is $l$-partially consistent with a mesh $X=\{x_1,x_2, ..., \}$ then for any $G\geq 0$, any $i'\geq 0$, any $0\leq \epsilon_1\leq 1$, any $\mathbf{\epsilon_2}=\{\epsilon_{2,0},\cdots,\epsilon_{2,i'}\}\in [0,1]^{i'+1}$ we have:
\begingroup\makeatletter\def\f@size{9}\check@mathfonts
\begin{align*}
 &\Pr\left[ \lambda(S,G) \geq \min\limits_{1\leq idx \leq l} \mathtt{LP1a}(X_l,F^S,idx,G,1,i',\mathbf{\epsilon_2}) \right] \geq 1- \delta \\
 &\Pr\left[\lambda(S,G) \leq \max\limits_{1\leq idx \leq l} |\mathtt{LP1a}(X_l,F^S,idx,G,-1,i',\mathbf{\epsilon_2})|   \right] \geq 1-\delta
\end{align*}
\endgroup
where $\delta= \exp\left(-2N\epsilon_1^2\right) + 2\times\sum_{0\leq i\leq i'}\exp\left(\frac{-2(N-i)^2\epsilon_{2,i}^2}{N(i+1)^2}\right)$. The randomness is taken over the sample $S \leftarrow \mathcal{P}^N$ of size $N$.
\end{corollary}

Corollary \ref{crl:bound3final-lower} follows from Theorem \ref{thm:bound3final-lower} by applying Theorem~\ref{thm:ExpectationConcentration}.

\begin{corollary}\label{crl:bound3final-lower}
Given an unknown password distribution $\mathcal{P}$, a sample set $S\leftarrow \mathcal{P}^N$,
%let $S$ be a sample set containing $N$ items independently randomly sampled from $\mathcal{P}$, and let $F$ be its frequency list. 
for any $G\geq 0$, integer $i'\geq 0$, $\mathbf{\epsilon_2}=\{\epsilon_{2,0},\cdots,\epsilon_{2,i'}\}\in [0,1]^{i'+1}$, $\mathbf{\hat{x}_{\epsilon_3}}=\{\hat{x}_{\epsilon_{3,0}},\cdots,\hat{x}_{\epsilon_{3,i'}}\}$, $\mathbf{\epsilon_3}=\{\epsilon_{3,i}=\frac{1}{q^{i+1}}(\frac{1-\hat{x}_{\epsilon_{3,i}}}{1-q\hat{x}_{\epsilon_{3,i}}})^{N-i}-1, 0\leq i\leq i'\}$, we can bound $\lambda(S,G)$ using a set of fine-grained meshes $X_{l,q}$ we generate as below:
\begingroup\makeatletter\def\f@size{8.5}\check@mathfonts
\begin{align*}
   &\Pr[\lambda(S,G) \geq \min\limits_{1\leq idx\leq l+1}\mathtt{LPlower}(X_{l,q},F^S,idx,G,i',\mathbf{\epsilon_2},\mathbf{\epsilon_3},\mathbf{\hat{x}_{\epsilon_3}}) - \epsilon_1] \\
   &\geq 1-\delta \\
   &\Pr[\lambda(S,G) \leq \max\limits_{1\leq idx\leq l+1}\mathtt{LPupper}(X_{l,q},F^S,idx,G,i',\mathbf{\epsilon_2},\mathbf{\epsilon_3},\mathbf{\hat{x}_{\epsilon_3}}) - \epsilon_1] \\
   &\geq 1-\delta  
\end{align*}
\endgroup
where each inequality holds with probability at least $\delta = \exp\left(-2N\epsilon_1^2\right) + 2\times\sum_{0\leq i\leq i'}\exp\left(\frac{-2(N-i)^2\epsilon_{2,i}^2}{N(i+1)^2}\right)$. The randomness is taken over the sample $S$ of size $N$.
\end{corollary}}
\section{Missing Proofs}

\fullversion{}{\subsection{Missing Proofs in Section~\ref{sec:lowerbound-priorwork}}\label{app:lb1}

Here we present the missing proofs of Theorem~\ref{thm:lowerbound-priorwork} in Section~\ref{sec:lowerbound-priorwork}.

Consider an arbitrary sample set $S=\{s_1,...,s_N\}\leftarrow \mathcal{P}^N$ with size $N$. 
%For each $1\leq k \leq N$, let $y_k\in\{0,1\}$ be an indicator such that $y_k=1$ if and only if $s_k$ is cracked after an attacker makes $G=NL$ guesses ($L>1$ is a parameter we can choose). 
%Let $p_i$ be the real distribution of $pwd_i$ ordered by the real probability distribution (i.e., $p_1\geq p_2\geq p_3,...$). Let $f^S_i$ be the frequency of $pwd_i$ in the sample set $S$. 
Same to Blocki et al~\cite{SP:BloHarZho18}, We define $B_i$ to be an indicator random variable such that $B_i=1$ if and only if $f^S_i\geq j$ while $p_i<\frac{1}{NL}$; otherwise, $B_i=0$. We call a sample $s_i$ is $(j,L)-bad~overestimate$ if $B_i = 1$. To analyzing the lower bound of the number of cracked user passwords in the sample set $S$, we first bound the number of samples in $S$ that are $(j,L)-bad~overestimate$.
\begin{lemma}\label{lem:jLbad}
For any $t\geq 0$, we have:
$$Pr[\sum_{i} f^S_i \times B_i \geq \frac{N}{(j-1)!L^{j-1}} + t ] \leq \exp\left(\frac{-2t^2}{N j^2}\right).$$
where the randomness is taken over the sample set $S$ with size $N$.
\end{lemma}
\begin{proof}
We let $g(s_1,s_2,...,s_N)$ be a function that outputs the number of $(j,L)-bad~overestimate$ samples, i.e. $\sum_{i} f^S_i \times B_i$. Then for any two different input $s_1,...,s_i,...,s_N$ and $s_1,...,s'_i,...,s_N$ which only differ on the $i-th$ input, the difference between two outputs of $g$ will be no greater than $j$, i.e., $\sup\limits_{s_1,...,s_i,...,s_N,s'_i}|g(x_1,...,s_i,...,x_N) - g(s_1,...,s'_i,...,s_N)| \leq j$. Claim 6 in Blocki et al~\cite{SP:BloHarZho18} proves that $\mathbb{E}(\sum_{i} f^S_i \times B_i)\leq \frac{N}{(j-1)!L^{j-1}}$.

Using Theorem~\ref{thm:BoundedDifferencesInequality}, we have:
$\Pr[\sum_{i} f^S_i \times B_i \geq \frac{N}{(j-1)!L^{j-1}} + t]
\leq  \Pr[g(s_1,...,s_N) \geq \mathbb{E}(g(s_1,...,s_N)) + t] \leq  \exp\left(\frac{-2t^2}{N j^2}\right)$.
\end{proof}

\begin{lemma}\label{lem:lowerbound-priorwork}
For any $L\geq 1$, $t_1 \geq 0$ and any integers $N \geq 1$ and $j\geq 2$ setting $G=NL$ we have:
\begin{align*}
&\Pr\left[\lambda(S,G) \geq f(S,L) -t/N\right]
\geq  1 - \exp\left(\frac{-2t^2}{Nj^2}\right) 
\end{align*}
where $f(S,L)\doteq \frac{1}{N} \sum_{i: f_i^S \geq j} f_i^S - \frac{N}{(j-1)! L^{j-1}}$ and the randomness is taken over the sample set $S\leftarrow \mathcal{P}^N$.
\end{lemma}
\begin{proof}
%\proofof{Theorem~\ref{thm:lowerbound-priorwork}}
We consider a perfect knowledge attacker who makes $G=N\times L$ guesses for any $L\geq 1$. The $G$ guesses contains the $G $ passwords with the top $G$ highest probabilities. Since there are at most $NL$ passwords with probability no less than $\frac{1}{NL}$, the perfect knowledge attacker making $G=NL$ guesses must crack all of them. 
Note that for any password $pwd_i$ satisfying $f^S_i\geq j$, we have either $p_i\geq \frac{1}{NL}$ or $B_i=1$. Therefore, at least $\sum_{i:f^S_i\geq j}f^S_i - \sum_{i} f^S_i \times B_i$ will be cracked by a perfect knowledge attacker making $G=NL$ guesses, i.e. $N\lambda(S,G)\geq \sum_{i:f^S_i\geq j}f^S_i - \sum_{i} f^S_i \times B_i$. Using Lemma~\ref{lem:jLbad} we have:
\begin{align*}
    &Pr[\lambda(S,G) \geq \frac{1}{N}(\sum_{i:f^S_i\geq j}f^S_i - \frac{N}{(j-1)!L^{j-1}} - t)] \\
    \geq & Pr[\sum_{i:f^S_i\geq j}f^S_i - \sum_{i} f^S_i B_i \geq \sum_{i:f^S_i\geq j}f^S_i - \frac{N}{(j-1)!L^{j-1}} - t] \\
    =& \Pr[\sum_{i} f^S_i \times B_i \leq \frac{N}{(j-1)!L^{j-1}} + t] \geq 1 - \exp\left(\frac{-2t^2}{N j^2}\right).
\end{align*}
\end{proof}

\begin{remindertheorem}{Corollary~\ref{thm:lowerbound-priorwork}}
\thmpriorlb
\end{remindertheorem}
\begin{proofof}{Theorem~\ref{thm:lowerbound-priorwork}}
%We then apply McDiarmid's inequality to argue that the lower bound holds with high probability. As an immediate corollary of Theorem \ref{thm:lowerbound-priorwork} we can show that, except with negligible probability, we have $\lambda_{NL} \geq f(S,L) - t/N - \epsilon$ for any constant $\epsilon>0$.
Using Lemma~\ref{lem:lowerbound-priorwork}, we can apply McDiarmid's inequality to argue that the lower bound holds with high probability. We can show that, except with negligible probability, we have $\lambda_{NL} \geq f(S,L) - t/N - \epsilon$ for any constant $\epsilon>0$.
Using Theorem~\ref{thm:ExpectationConcentration} we have:
\begin{align*}
    &\Pr[\lambda_G \geq \frac{1}{N}(\sum_{i:f^S_i\geq j}f^S_i - \frac{N}{(j-1)!L^{j-1}} - t)-\epsilon] \\
    \geq & \Pr[\lambda_G \geq \lambda(S,G) - \epsilon \bigwedge \lambda(S,G) \\
    \geq & \frac{1}{N}(\sum_{i:f^S_i\geq j}f^S_i - \frac{N}{(j-1)!L^{j-1}} - t)] \\
    \geq & \Pr[\lambda_G \geq \lambda(S,G) - \epsilon] \\
    & - \Pr[\lambda(S,G) \geq \frac{1}{N}(\sum_{i:f^S_i\geq j}f^S_i - \frac{N}{(j-1)!L^{j-1}} - t)] \\
    \geq & 1 - \exp\left(-2N\epsilon^2\right) - \exp\left(\frac{-2t^2}{N j^2}\right).
\end{align*} %$\hfill \square$
\end{proofof}
}

\subsection{Missing Proofs of LP Bounds in Section~\ref{sec:LPbounds}}\label{app:LPmissingproofs}

\begin{reminderlemma}{Lemma \ref{lem:goodturing}}
\lemmagoodturingideal
\end{reminderlemma}

\begin{proofof}{Lemma \ref{lem:goodturing}}
We first prove that $\sum_{j}h_j\times x_j\times \bpdf(i,N,x_j)$ can be bounded using $\frac{i+1}{N-i}\mathbb{E}(F^S_{i+1})$ as below:
\begingroup\makeatletter\def\f@size{9.5}\check@mathfonts
\begin{align*}
    &\sum_{j} h_j x_j \bpdf(i,N,x_j) \\
    =& \sum_{j} h_j x_j \frac{N!}{(N-i)!i!}x_j^i(1-x_j)^{N-i} \\
    \fullversion{}{=& \sum_{j} h_j \frac{i+1}{N-i} (1-x_j) \frac{N!}{(N-i-1)!(i+1)!}x_j^{i+1}(1-x_j)^{N-i-1} \\}
    =& \frac{i+1}{N-i}\sum_{j} h_j (1-x_j) \bpdf(i+1,N,x_j) \\
    =& \frac{i+1}{N-i}\mathbb{E}(F^S_{i+1}) - \frac{i+1}{N-i}\sum_{j} h_j x_j \bpdf(i+1,N,x_j) 
    \vspace{-0.15cm}
\end{align*}
\endgroup
Since $0 \leq \frac{i+1}{N-i}\sum_{j} h_j\times x_j\times \bpdf(i+1,N,x_j) \leq \frac{i+1}{N-i}$, we have $\frac{i+1}{N-i}\mathbb{E}(F^S_{i+1}) - \frac{i+1}{N-i} \leq \sum_{j} h_j\times x_j\times \bpdf(i,N,x_j) \leq \frac{i+1}{N-i}\mathbb{E}(F^S_{i+1})$.

Next we will prove that when $i$ is small, $F^S_{i+1}$ is highly concentrated on $\mathbb{E}(F^S_{i+1})$. We define $Y_1,...,Y_N$ to be $N$ independent password sample random variables, and consider $F^S_{i+1} = h(Y_1,...,Y_N)$ to be a function that outputs the number of distinct passwords that appear exact $i+1$ times among all $N$ samples in the sample set $S$. Then we have $\sup\limits_{y_1,...,y_i,...,y_N,y'_i}|h(y_1,...,y_i,...,y_N) - h(y_1,...,y'_i,...,y_N)| \leq 1$. Using Theorem~\ref{thm:BoundedDifferencesInequality}, for any $t_1 \geq 0$ we have:
\begin{align*}
    &\Pr[h(Y_1,...,Y_N) \geq \mathbb{E}(F^S_{i+1}) - t_1] \geq 1-\exp\left(-2t_1^2/N\right) \\
    &\Pr[h(Y_1,...,Y_N) \leq \mathbb{E}(F^S_{i+1}) + t_1] \geq 1-\exp\left(-2t_1^2/N\right)
\end{align*}

Denote $t_1 = \frac{N-i}{i+1}\epsilon_{2,i}$. Since $h(Y_1,...,Y_N) = F^S_{i+1}$ we have $\frac{(i+1)F^S_{i+1}}{N-i}-\epsilon_{2,i} - \frac{i+1}{N-i} \leq \sum_{j} h_j\times x_j\times \bpdf(i,N,x_j) \leq \frac{(i+1)F^S_{i+1}}{N-i}+\epsilon_{2,i}$ with probability at least $1 - 2\times\exp\left(\frac{-2(N-i)^2\epsilon_{2,i}^2}{N(i+1)^2}\right)$.
\end{proofof}

\fullversion{}{
\newcommand{\lemmagoodturingp}{Fix any $i\geq 0$ and $0\leq \epsilon_{2,i}\leq 1$ and assume that $x_{\ell}<\frac{1}{N}$ and that $\mathcal{P}$ is $l$-partially consistent with our mesh $X$. If $i=0$ let $W_0$ be an indicator random variable which is $1$ if and only if 
$\frac{(i+1)F^S_{i+1}}{N-i} -\epsilon_{2,i} - \frac{i+1}{N-i} - p \leq \sum_{j=1}^{\ell} h_jx_j \cdot \bpdf(i,N,x_j)$ and $ \sum_{j=1}^{\ell} h_jx_j \cdot \bpdf(i,N,x_j) \leq \frac{(i+1)F^S_{i+1}}{N-i} +  \epsilon_{2,i} - p\times \bpdf(i,N,x_{\ell})$. 
Similarly, if $i>0$ let $W_i$ be an indicator random variable which is $1$ if and only if
$\frac{(i+1)F^S_{i+1}}{N-i}-\epsilon_{2,i} - \frac{i+1}{N-i} - p\times \bpdf(i,N,x_{\ell}) \leq \sum_{j=1}^{\ell} h_j\times x_j\times \bpdf(i,N,x_j) $ and $ \sum_{j=1}^{\ell} h_j\times x_j\times \bpdf(i,N,x_j) \leq  \epsilon_{2,i}$.
Then the constraints hold with probability \[ \Pr[W_i = 1] \geq 1- 2\times\exp\left(\frac{-2(N-i)^2\epsilon_{2,i}^2}{N(i+1)^2}\right) \] where the randomness is taken over the selection of $S \leftarrow \mathcal{P}^N$.
}

\begin{reminderlemma}{Lemma \ref{lem:goodturing_p}}
\lemmagoodturingp
\end{reminderlemma}
\begin{proofof}{Lemma \ref{lem:goodturing_p}}
Note that for any $x<x_l< \frac{1}{N}$ we have $\bpdf(0,N,x_l) \leq \bpdf(0,N,x) \leq 1$ for $i=0$ and  $0 \leq \bpdf(i,N,x) \leq \bpdf(i,N,x_l)$ for any $i\geq 1$. Then when $i=0$ we have the following bounds on  $\sum_{j=1}^l h_j\times x_j\times \bpdf(i,N,x_j)$:

\begin{align*}
    &\sum_{j=1}^l h_j x_j \bpdf(i,N,x_j) \\
    =& \sum_{j\geq 1} h_j x_j \bpdf(i,N,x_j) - \sum_{j > l} h_j x_j \bpdf(i,N,x_j) \\
    \leq &\sum_{j\geq 1} h_j x_j \bpdf(i,N,x_j) - \sum_{j > l} h_j x_j \bpdf(i,N,x_l) \\
    =& \sum_{j\geq 1} h_j x_j \bpdf(i,N,x_j) - p \bpdf(i,N,x_l) \ , 
\end{align*}
and similarly, for our lower bound we have
\begin{align*}
    &\sum_{j=1}^l h_j x_j \bpdf(i,N,x_j) \\
    =& \sum_{j\geq 1} h_j x_j \bpdf(i,N,x_j) - \sum_{j > l} h_j x_j \bpdf(i,N,x_j) \\
    \geq &\sum_{j\geq 1} h_j x_j \bpdf(i,N,x_j) - \sum_{j > l} h_j x_j \\
    =& \sum_{j\geq 1} h_j x_j \bpdf(i,N,x_j) - p; \ .
\end{align*}
Similarly, when $i>0$ we can derive the following upper/lower bounds

\begin{align*}
    &\sum_{j=1}^l h_j x_j \bpdf(i,N,x_j) \\
    =& \sum_{j\geq 1} h_j x_j \bpdf(i,N,x_j) - \sum_{j > l} h_j x_j \bpdf(i,N,x_j) \\
    \leq &\sum_{j\geq 1} h_j x_j \bpdf(i,N,x_j) \ , 
\end{align*}
and
\begin{align*}
    &\sum_{j=1}^l h_j x_j \bpdf(i,N,x_j) \\
    =& \sum_{j\geq 1} h_j x_j \bpdf(i,N,x_j) - \sum_{j > l} h_j x_j \bpdf(i,N,x_j) \\
    \geq &\sum_{j\geq 1} h_j x_j \bpdf(i,N,x_j) - \sum_{j > l} h_j x_j \bpdf(i,N,x_l) \\
    =& \sum_{j\geq 1} h_j x_j \bpdf(i,N,x_j) - p \bpdf(i,N,x_l) \ .
\end{align*}

Applying Lemma~\ref{lem:goodturing}, for any $i\geq 0$ and $0\leq \epsilon_{2,i}\leq 1$, we have $\frac{(i+1)F_{i+1}}{N-i}-\epsilon_{2,i} - \frac{i+1}{N-i} \leq \sum_{j\geq 1} h_j\times x_j\times \bpdf(i,N,x_j) \leq \frac{(i+1)F_{i+1}}{N-i}+\epsilon_{2,i}$ with probability at least $1 - 2\times\exp\left(\frac{-2(N-i)^2\epsilon_{2,i}^2}{N(i+1)^2}\right)$. Therefore, we have:
\begin{align*}
&\text{if }i=0, \frac{(i+1)F_{i+1}}{N-i}-\epsilon_{2,i} - \frac{i+1}{N-i} - p \\
&\leq \sum_{j=1}^l h_j\times x_j\times \bpdf(i,N,x_j) \\
&\leq \frac{(i+1)F_{i+1}}{N-i}+\epsilon_{2,i} - p\times \bpdf(i,N,x_l); \\
&\text{if }i\geq 1, \frac{(i+1)F_{i+1}}{N-i}-\epsilon_{2,i} - \frac{i+1}{N-i} - p\times \bpdf(i,N,x_l) \\
&\leq \sum_{j=1}^l h_j\times x_j\times \bpdf(i,N,x_j) \leq \frac{(i+1)F_{i+1}}{N-i}+\epsilon_{2,i}.
\end{align*}
where for any $i\geq 0$ and $0\leq \epsilon_{2,i}\leq 1$ the inequality holds with probability at least $1 - 2\times\exp\left(\frac{-2(N-i)^2\epsilon_{2,i}^2}{N(i+1)^2}\right)$.
\end{proofof}
}

\begin{theorem}\label{thm:bound3ideal}
Given a probability distribution $\mathcal{P}$ which is consistent with a finite mesh $X_l = \{x_1,\ldots, x_l\}$ for any $G\geq 0$, any $i'\geq 0$, any $\mathbf{\epsilon_2}=\{\epsilon_{2,0},\ldots,\epsilon_{2,i'}\}\in [0,1]^{i'+1}$, we have:
\begin{align*}
    &\Pr\left[ \lambda_G \geq \min\limits_{1\leq idx \leq l} \mathtt{LP1}(X_l,F^S,idx,G,1,i',\mathbf{\epsilon_2}) \right] \geq 1-\delta \\ 
    &\Pr\left[\lambda_G \leq \max\limits_{1\leq idx \leq l} |\mathtt{LP1}(X_l,F^S,idx,G,-1,i',\mathbf{\epsilon_2})|\right] \geq 1-\delta
\end{align*}
where $\delta = 2\times\sum_{0\leq i\leq i'}\exp\left(\frac{-2(N-i)^2\epsilon_{2,i}^2}{N(i+1)^2}\right)$ and the randomness is taken over the sample $S \leftarrow \mathcal{P}^N$ of size $N$.
\end{theorem}

\subsection{Other Missing Proofs}\label{app:othermissingproofs}
\begin{remindertheorem}{Theorem~\ref{thm:ExpectationConcentration}}
\expectationconcentration
\end{remindertheorem}

\begin{proofof}{Theorem~\ref{thm:ExpectationConcentration}}
Recall that the samples $s_1,s_2,...,s_N$ in $S$ are $N$ independent random variables sampled from the real password distribution $\mathcal{P}$. For any two sample set $S=\{s_1,...,s_i,...,s_N\}$ and $S'=\{s_1,...,s'_{i},...,s_N\}$ that only differs on one sample $s_i$ and $s'_{i}$, we have $N\times|\lambda(S,G)-\lambda(S',G)|\leq 1$. Therefore, using Theorem~\ref{thm:BoundedDifferencesInequality}, for any parameter $0\leq \epsilon \leq 1$ we have:
\begin{align*}
    &\Pr[\lambda(S,G) \geq \lambda_G + \epsilon] \leq \exp\left(-2N\epsilon^2\right) \\ &\Rightarrow \Pr[\lambda(S,G) \leq \lambda_G + \epsilon] \geq 1 - \exp\left(-2N\epsilon^2\right)\\
    &\Pr[\lambda(S,G) \leq \lambda_G - \epsilon] \leq \exp\left(-2N\epsilon^2\right) \\ &\Rightarrow \Pr[\lambda(S,G) \geq \lambda_G - \epsilon] \geq 1 - \exp\left(-2N\epsilon^2\right)
\end{align*}
\end{proofof}

\fullversion{}{
\begin{remindertheorem}{Theorem~\ref{thm:bound1}}
\samplinglbthm
\end{remindertheorem}

\begin{proofof}{Theorem~\ref{thm:bound1}}
Fixing any $D_1$ we have $\mathbb{E}_{D_2}[h(D_1,D_2,G)] =  $ $d\times \sum_{i\in T(D_1,G)}p_i$ $\leq d\times \sum_{i\leq G}p_i = d\lambda_G$ where the expectation is taken over the selection of $D_2$. Given an arbitrary dataset $D = \{s_1,\ldots, s_d\} \in \mathcal{P}^d$, an item $s \in \mathcal{P}$ and an index $j \leq d$, we define $D^{j,s} = \{s_1, \ldots, s_{j-1}, s, s_{j+1}, \ldots s_d \}$ to be the new dataset obtained by swapping item $s_j$ for $s$. Observe that for all $D_1 \in \mathcal{P}^{N-d}$ we have $ \sup_{D_2 \in \mathcal{P}^d, s \in \mathcal{P}, j \leq d}|h(D_1,D_2,G) - h(D_1, D_2^{j,s},G)| \leq 1.$ Thus, we can apply the Bounded Differences Inequality (Theorem \ref{thm:BoundedDifferencesInequality}) to get:

\begin{align*}
&\Pr[h(D_1,D_2,G) \leq d\times \sum_{i\in T(D_1,G)}p_i + t] \geq 1 - \exp\left(-2t^2/d\right) \\
&\Rightarrow \Pr[\lambda_G \geq \frac{1}{d}(h(D_1,D_2,G) - t)]  \geq 1 - \exp\left(-2t^2/d\right)
\end{align*}

The last line follows since $\lambda_G =  \sum_{i\leq G}p_i \geq \sum_{i \in T(D_1,G)} p_i$. 
\end{proofof}
}
\fullversion{}{\section{Missing Proofs of Intermediate LP in Section~\ref{sec:midLP}}\label{app:IntermediateLPsProofs}

We can use the observations described in Section~\ref{sec:midLP} to update Constraint (2) in $\mathtt{LP1}$. Similarly, we can replace Constraint (3) in $\mathtt{LP1}$ with $\sum_{j=1}^{\ell} h_j x_j = 1 - p$. We call this updated linear program $\mathtt{LP1a}$as shown below. 

%%%%%This is where LP1a originally located%%%%%%%%%
\fbox{\begin{minipage}{8cm}
\textbf{Linear Programming Task 1a: }\\ $\mathtt{LP1a}(G,b,X_{\ell},F^S,idx,i',\mathbf{\epsilon_2})$ \\
\textbf{Input Parameters:} $G$, $b$, $X_{\ell}=\{x_1,...,x_{\ell}\}$, $F^S=\{F^S_1,...,F^S_N\}$, $idx$, $i'$, $\mathbf{\epsilon_2}=\{\epsilon_{2,0},\ldots,\epsilon_{2,i'}\}$ \\
\textbf{Variables:} $h_1,...,h_{\ell},c,p$ \\
%\textbf{Objective:} $\min\left(b\times (\sum_{j<idx}h_j+c)\right)$ \\ 
\textbf{Objective:} $\min\left(b\times (\sum_{j<idx}h_j\times x_j +c\times x_{idx})\right)$ \\ 
\textbf{Constraints:} 
\begin{enumerate}
    \item $\sum_{j<idx}h_j+c = G$
    \item $\forall 0\leq i\leq i'$:
    \begin{enumerate}
        \item for $i=0$, $\frac{(i+1)F^S_{i+1}}{N-i}-\epsilon_{2,i} - \frac{i+1}{N-i} - p \leq \sum_{j=1}^{\ell} h_j\times x_j\times \bpdf(i,N,x_j) \leq \frac{(i+1)F^S_{i+1}}{N-i}+\epsilon_{2,i} - p\times \bpdf(i,N,x_{\ell}) $
        \item for $1\leq i\leq i'$, $\frac{(i+1)F^S_{i+1}}{N-i}-\epsilon_{2,i} - \frac{i+1}{N-i} - p\times \bpdf(i,N,x_{\ell}) \leq \sum_{j=1}^{\ell} h_j\times x_j\times \bpdf(i,N,x_j) \leq \frac{(i+1)F^S_{i+1}}{N-i}+\epsilon_{2,i} $
    \end{enumerate}
    \item $ \sum_{j=1}^{\ell} h_j\times x_j = 1 - p$
    \item $0 \leq c \leq h_{idx}$
\end{enumerate}
(\textbf{Note:} we consider $idx=1,2,...,\ell+1$. When $idx=\ell+1$, we define $h_{\ell+1}=G$; $x_{\ell+1} = 0$ for $b=1$; $x_{\ell+1} = x_{\ell}$ for $b=-1$.)
\end{minipage}}

Lemma \ref{lem:goodturing_p} shows that with high probability over the selection of $S$ each of the constraints in $\mathtt{LP1a}$ will be {\em consistent} with the distribution $\mathcal{P}$ provided that $x_{\ell} < 1/N$ and $\mathcal{P}$ is $l$-partially consistent with the mesh $X$. We say that $\mathcal{P}$ is $l$-partially consistent with the mesh $X$ if for all passwords $pwd_i$ in the support of $\mathcal{P}$ with corresponding probability $p_i$ we either have (1) $p_i \in X_{\ell}$ or (2) $p_i < x_{\ell}$. Intuitively, Lemma \ref{lem:goodturing_p} follows by upper/lower bounding $\sum_{j=l+1}^{\infty} x_j h_j \bpdf(i,N, x_j)$ and then applying our prior bounds from Lemma~\ref{lem:goodturing}. The full proof can be found in Appendix~\ref{app:othermissingproofs}.  
%\pnote{add to the lemma: Let $X=\{x_1,x_2,...\}$, $\forall P\in D(X)$ s.t. $x_{\ell}<1/N$, $S<-- P^N$}

%moved the definition of \lemmagoodturingp to MissingProofs.tex
%\newcommand{\lemmagoodturingp}{Fix any $i\geq 0$ and $0\leq \epsilon_{2,i}\leq 1$ and assume that $x_{\ell}<\frac{1}{N}$ and that $\mathcal{P}$ is $l$-partially consistent with our mesh $X$. If $i=0$ let $W_0$ be an indicator random variable which is $1$ if and only if $\frac{(i+1)F^S_{i+1}}{N-i} -\epsilon_{2,i} - \frac{i+1}{N-i} - p \leq \sum_{j=1}^{\ell} h_jx_j \cdot \bpdf(i,N,x_j)$ and $ \sum_{j=1}^{\ell} h_jx_j \cdot \bpdf(i,N,x_j) \leq \frac{(i+1)F^S_{i+1}}{N-i} +  \epsilon_{2,i} - p\times \bpdf(i,N,x_{\ell})$. Similarly, if $i>0$ let $W_i$ be an indicator random variable which is $1$ if and only if $\frac{(i+1)F^S_{i+1}}{N-i}-\epsilon_{2,i} - \frac{i+1}{N-i} - p\times \bpdf(i,N,x_{\ell}) \leq \sum_{j=1}^{\ell} h_j\times x_j\times \bpdf(i,N,x_j) $ and $ \sum_{j=1}^{\ell} h_j\times x_j\times \bpdf(i,N,x_j) \leq  \epsilon_{2,i}$. Then the constraints hold with probability \[ \Pr[W_i = 1] \geq 1- 2\times\exp\left(\frac{-2(N-i)^2\epsilon_{2,i}^2}{N(i+1)^2}\right) \] where the randomness is taken over the selection of $S \leftarrow \mathcal{P}^N$. }

\begin{lemma}\label{lem:goodturing_p}
\lemmagoodturingp
\end{lemma}

Similar to Theorem~\ref{thm:bound3ideal}, using Lemma~\ref{lem:goodturing_p}, we can derive the high-confidence upper and lower bounds of $\lambda_G$ as below: 

\begin{theorem}\label{thm:bound3mid}
Suppose that the probability distribution $\mathcal{P}$ is $l$-partially consistent with a mesh $X=\{x_1,x_2, ..., \}$ then for any $G\geq 0$, any $i'\geq 0$, any $\mathbf{\epsilon_2}=\{\epsilon_{2,0},\ldots,\epsilon_{2,i'}\}\in [0,1]^{i'+1}$, we have 

\begin{align*}
 &\Pr\left[ \lambda_G \geq \min\limits_{1\leq idx \leq l} \mathtt{LP1a}(X_{\ell},F^S,idx,G,1,i',\mathbf{\epsilon_2}) \right] \geq 1- \delta \\
 &\Pr\left[\lambda_G \leq \max\limits_{1\leq idx \leq l} |\mathtt{LP1a}(X_{\ell},F^S,idx,G,-1,i',\mathbf{\epsilon_2})|   \right] \geq 1-\delta
\end{align*}

where each inequality holds with probability at least $\delta=1 - 2\times\sum_{0\leq i\leq i'}\exp\left(\frac{-2(N-i)^2\epsilon_{2,i}^2}{N(i+1)^2}\right)$. The randomness is taken over the sample $S \leftarrow \mathcal{P}^N$ of size $N$.
\end{theorem}

As an immediate corollary of Theorem \ref{thm:bound3mid} and Theorem~\ref{thm:ExpectationConcentration} we can also bound $\lambda(S,G)$ --- \fullversion{see the full version}{see Corollary \ref{crl:bound3mid} in Appendix~\ref{app:lambdaSG}}.
}
\fullversion{}{\section{Missing Proofs of Final LP in Section~\ref{sec:finalLP-mainbody}}\label{app:FinalLPsProofs}

In this section, we provide a complete description of the linear programming approach we describe in Section~\ref{sec:finalLP-mainbody} and a complete proof of Theorem~\ref{thm:bound3final-lower}.

Recall that we fix a particular mesh $X_{l,q} = \{x_1,\ldots, x_l \}$ where where $x_l = \frac{1}{10^4 N}$, $l=\lfloor\frac{\ln (\frac{1}{x_l})}{\ln q}\rfloor + 1$ and for each $1 \leq i < l$ we have $x_{i} = q \cdot x_{i+1} = q^{l-i} x_l$. Here, $q>1$ is an parameter that determines how fine-grained our mesh values are. We also consider an ideal mesh $X^r = \{ x_1^r, x_2^r, \ldots \} = \{ \Pr[pwd] : pwd \in \mathcal{P}\}$ for the real distribution $\mathcal{P}$ along with a corresponding histogram $H^r = \{h_1^r, h_2^r \ldots \}$ where $h_i^r \doteq \left| \{ pwd \in \mathcal{P}: \Pr[pwd] = x_i^r\} \right|$. Given the number of guesses $G > 0$ we have $\lambda_G = X_{i(G,H)}c(G,H)+ \sum_{j=1}^{i(G,H)-1} x_i^r h_i^r $ where we define $i(G,H) = \max \{ j : \sum_{i=1}^{j-1} h_i \leq G \} $ and $c(G,H) = G- \sum_{i=1}^{i(G,H)-1} h_i$ as before.

The LP to generate lower bound is not exactly the same as the LP to generate upper bound. We discuss them separately in the following sections. We present $\mathtt{LPlower}$ for lower bounds and prove the lower bound in Theorem~\ref{thm:bound3final-lower} in Appendix~\ref{sec:lplb-app}, and present $\mathtt{LPupper}$ for upper bounds and prove the upper bound in Theorem~\ref{thm:bound3final-lower} in Appendix~\ref{sec:lpub-app}.

\subsection{Lower Bound}\label{sec:lplb-app}

Let $x^r_k = \min\limits_{i\geq 1}\{x^r_i:x^r_i\geq x_l\}$ be the smallest probability in $P$ no less than $x_l$. We will use $X_l$ to estimate $X^r_k=\{x^r_1,...,x^r_k\}$, and let $p=\sum_{i>k}x^r_i\times h^r_i$ be the remaining probability. Specifically, for any $1\leq i\leq k$, we map $x^r_i\in X^r_k$ to its closest value $x_{b_i}$ where $b_i=\arg\max\limits_{1\leq j\leq l}\{x_{j}:x_{j} \leq x^r_i\}$ in $X_l$ that are equal to or smaller than $x^r_i$. Then we have $x_{b_i} \leq x^r_i\leq q\times x_{b_i}$. Therefore, for any $1\leq j\leq l$, the histogram $h_j$ of $x_j$ is the sum of all values in $H^r$ with the corresponding mesh values being mapped to $x_j$, i.e., $h_j = \sum_{i:b_i=j}h^r_i$. Recall that for any $G\geq 0$, $G = \sum_{i=1}^{i(G,H^r)-1}h^r_i + c(G,H^r)$. We can use $H$ instead of $H^r$ to represent $G$ by defining $i(G,H) = b_{i(G,H^r)}$ and $c(G,H) = \sum_{j:b_j=i(G,H)}h^r_j + c(G,H^r)$ such that $G= c(G,H) + \sum_{i=1}^{i(G,H)-1}h_i$. Then we have 
\begin{align*}
    \lambda_G &= c(G,H^r)x_{i(G,H^r)} + \sum_{i=1}^{i(G,H^r)-1}x_ih^r_i \\
    &\geq c(G,H^r)x_{b_i(G,H^r)} + \sum_{i=1}^{i(G,H^r)-1}x_{b_i}h^r_i \\
    &= c(G,H)x_{i(G,H)} + \sum_{i=1}^{i(G,H)-1}x_ih_i
\end{align*}

Therefore, given a set of fine-grained mesh values $X_l$ and a guessing number $G$, as long as we can find a lower bound of $c(G,H)x_{i(G,H)} + \sum_{i=1}^{i(G,H)-1}x_ih_i$ using linear programming, it will also be a lower bound of $\lambda_G$.

The linear programming task for bounding $c(G,H)x_{i(G,H)} + \sum_{i=1}^{i(G,H)-1}x_ih_i$ in this section is similar to $\mathtt{LP1a}$ we described in Section~\ref{sec:midLP}, except that Constraint (2) and (3) need to be modified considering the difference between our fine-grained mesh $X_l$ and the real probabilities $X^r_k$.
For Constraint (3), previously we have $\sum_{j=1}^k h^r_j\times x^r_j = 1-p$. After mapping $x^r_j$ to $x_{b_j}$ with $x_{b_j} \leq x^r_j\leq q\times x_{b_j}$ for each $1\leq j\leq k$, we have:
\begin{align*}
    &\sum_{j=1}^l h_j\times x_j = \sum_{j=1}^k h^r_j\times x_{b_j} \leq \sum_{j=1}^k h^r_j\times x^r_{j} = 1-p \\
    &\sum_{j=1}^l h_j\times x_j = \sum_{j=1}^k h^r_j\times x_{b_j} \geq \sum_{j=1}^k h^r_j\times (\frac{1}{q}\times x^r_{j}) = \frac{1-p}{q}
\end{align*}

Constraint (2) that bounds $\sum_{j=1}^l h_j\times x_j\times \bpdf(i,N,x_j)$ also need to be changed considering the difference between $\sum_{j=1}^l h_j\times x_j\times \bpdf(i,N,x_j)$ and $\sum_{j=1}^k h^r_j\times x^r_j\times \bpdf(i,N,x^r_j)$. We will present the change and the proof in Lemma~\ref{lem:goodturing_LPlower}later in this section.

We use the final linear program $\mathtt{LPlower}$ to lower bound $\lambda_G$. This linear program $\mathtt{LPlower}$ can be considered as an extension from $\mathtt{LP1}$ and $\mathtt{LP1a}$ removing all ideal settings. 

$\mathtt{LPlower}$ takes $h_1,...,h_l,c,p$ as variables, and minimize $\sum_{j<idx}(h_j\times x_j +c\times x_{idx})$ with the constraint $\sum_{j<idx}h_j+c = G$. Given an arbitrary guessing number $G$, fixing other input parameters, we can lower bound $\lambda_G$ as  $\min\limits_{1\leq idx\leq l+1}\mathtt{LPlower}(X_l,F^S,idx,G,i',\mathbf{\epsilon_2},\mathbf{\epsilon_3},\mathbf{\hat{x}_{\epsilon_3}})$ by running $\mathtt{LPlower}(X_l,F^S,idx,G,i',\mathbf{\epsilon_2},\mathbf{\epsilon_3},\mathbf{\hat{x}_{\epsilon_3}})$ for $l+1$ times with $idx\in\{1,2,...,l+1\}$.

The following lemma proves Constraint (2) in $\mathtt{LPlower}$ holds with high probability:

\begin{lemma}\label{lem:goodturing_LPlower}
Given $F^S=\{F_1,...,F_N\}$ from a sample set $S\leftarrow \mathcal{P}^N$, fix any $q>1$, $i\geq 0$, $0\leq \epsilon_{2,i} \leq 1$, $\frac{i+1}{N+1}\leq \hat{x}_{\epsilon_{3,i}}\leq 1$, $\epsilon_{3,i}=\frac{1}{q^{i+1}}\left(\frac{1-\hat{x}_{\epsilon_{3,i}}}{1-q\hat{x}_{\epsilon_{3,i}}}\right)^{N-i}-1 \in(0,1)$, and let $X_l=X_l^q$. If $i=0$ let $W_0$ be an indicator random variable which is $1$ if and only if $\frac{1}{q^{i+1}} \left(\frac{(i+1)F_{i+1}}{N-i}-\epsilon_{2,i} - \frac{i+1}{N-i} - p\right)  \leq \sum_{j=1}^l h_j x_j \bpdf(i,N,x_j)$ and $\sum_{j=1}^l h_j x_j \bpdf(i,N,x_j) \leq (1+\epsilon_{3,i}) \left(\frac{(i+1)F_{i+1}}{N-i} + \epsilon_{2,i} - p\times \bpdf(i,N,qx_l)\right) + \bpdf(i,N,\hat{x}_{\epsilon_{3,i}})$. Similarly, if $i>0$ let $W_i$ be an indicator random variable which is $1$ if and only if $\frac{1}{q^{i+1}}\left(\frac{(i+1)F_{i+1}}{N-i}-\epsilon_{2,i} - \frac{i+1}{N-i} - p\times \bpdf(i,N,qx_l)\right)  \leq \sum_{j=1}^l h_j x_j \bpdf(i,N,x_j)$ and $\sum_{j=1}^l h_j x_j \bpdf(i,N,x_j)  \leq (1+\epsilon_{3,i}) \left(\frac{(i+1)F_{i+1}}{N-i}+\epsilon_{2,i}\right) + \bpdf(i,N,\hat{x}_{\epsilon_{3,i}})$. Then the constraints hold with probability \[ \Pr[W_i = 1] \geq 1 - 2\times\exp\left(\frac{-2(N-i)^2\epsilon_{2,i}^2}{N(i+1)^2}\right) \] where the randomness is taken over the selection of $S \leftarrow \mathcal{P}^N$.
%Given $F^S=\{F_1,...,F_N\}$ from a sample set $S\leftarrow \mathcal{P}^N$, for any $q>1$, $i\geq 0$, $0\leq \epsilon_{2,i} \leq 1$, $\frac{i+1}{N+1}\leq \hat{x}_{\epsilon_{3,i}}\leq 1$, and $\epsilon_{3,i}=\frac{1}{q^{i+1}}\left(\frac{1-\hat{x}_{\epsilon_{3,i}}}{1-q\hat{x}_{\epsilon_{3,i}}}\right)^{N-i}-1 \in(0,1)$, let $X_l=X_l^q$, then we have:
%\begin{align*}
%    &if~i=0, \\
%    &\frac{1}{q^{i+1}}\left(\frac{(i+1)F_{i+1}}{N-i}-\epsilon_{2,i} - \frac{i+1}{N-i} - p\right) \\
%    &\leq \sum_{j=1}^l h_j\times x_j\times \bpdf(i,N,x_j) \\
%    &\leq (1+\epsilon_{3,i})\left(\frac{(i+1)F_{i+1}}{N-i}+\epsilon_{2,i} - p\times \bpdf(i,N,qx_l)\right) \\
%    &~~~ + \bpdf(i,N,\hat{x}_{\epsilon_{3,i}}) \\
%    &if~i\geq 1, \\
%    &\frac{1}{q^{i+1}}\left(\frac{(i+1)F_{i+1}}{N-i}-\epsilon_{2,i} - \frac{i+1}{N-i} - p\times \bpdf(i,N,qx_l)\right) \\
%    &\leq \sum_{j=1}^l h_j\times x_j\times \bpdf(i,N,x_j) \\
%    &\leq (1+\epsilon_{3,i})\left(\frac{(i+1)F_{i+1}}{N-i}+\epsilon_{2,i}\right) + \bpdf(i,N,\hat{x}_{\epsilon_{3,i}})
%\end{align*}
%where for any $i$ the two inequalities hold with probability at least $1 - 2\times\exp\left(\frac{-2(N-i)^2\epsilon_{2,i}^2}{N(i+1)^2}\right)$.
\end{lemma}
\begin{proof}
Recall that for any $x_l<\frac{1}{N}$, any individual $i\geq 0$, and any $0\leq \epsilon_{2,i}\leq 1$, the following constraint holds with probability at least $1 - 2\times\exp\left(\frac{-2(N-i)^2\epsilon_{2,i}^2}{N(i+1)^2}\right)$, Lemma~\ref{lem:goodturing_p} proves the following constraint holds with probability at least $1 - 2\times\exp\left(\frac{-2(N-i)^2\epsilon_{2,i}^2}{N(i+1)^2}\right)$:
\begin{align*}
&\text{if }i=0, \\
&\frac{(i+1)F_{i+1}}{N-i}-\epsilon_{2,i} - \frac{i+1}{N-i} - p \leq \sum_{j=1}^k h^r_j x^r_j \bpdf(i,N,x^r_j) \\
&\leq \frac{(i+1)F_{i+1}}{N-i}+\epsilon_{2,i} - p\times \bpdf(i,N,x^r_k); \\
&\text{if }i\geq 1, \\
&\frac{(i+1)F_{i+1}}{N-i}-\epsilon_{2,i} - \frac{i+1}{N-i} - p\times \bpdf(i,N,x_l) \\
&\leq \sum_{j=1}^l h_j x_j \bpdf(i,N,x_j) \leq \frac{(i+1)F_{i+1}}{N-i}+\epsilon_{2,i}.
\end{align*}

Note that $\sum_{j=1}^l h_j\times x_j\times\bpdf(i,N,x_j) = \sum_{j=1}^k h^r_j\times x_{b_j}\times\bpdf(i,N,x_{b_j})$. To prove that the difference between $\sum_{j=1}^k  h^{r}_j\times x^{r}_j\times \bpdf(i,N,x^{r}_j)$ and $\sum_{j=1}^l h_j\times x_j\times\bpdf(i,N,x_j)$ is small, we look into $\frac{h^{r}_j\times x^{r}_j\times \bpdf(i,N,x^{r}_j)}{h^r_j\times x_{b_j}\times\bpdf(i,N,x_{b_j})}$ for each $1\leq j\leq k$.

One the one side, since $x_{b_j}\leq x^r_j \leq q x_{b_j}$, we bound it as $\frac{h^{r}_j\times x^{r}_j\times \bpdf(i,N,x^{r}_j)}{h^r_j\times x_{b_j}\times\bpdf(i,N,x_{b_j})} = \frac{ (x^{r}_j)^{i+1}(1-x^r_j)^{N-i}}{x_{b_j}^{i+1}(1-x_{b_j})^{N-i}}\leq q^{i+1}$. Therefore, with probability at least $1 - \exp\left(\frac{-2(N-i)^2\epsilon_2^2}{N(i+1)^2}\right)$, we have 
\begingroup\makeatletter\def\f@size{9.5}\check@mathfonts
\begin{align*}
    &\sum_{j=1}^l h_j\times x_j\times\bpdf(i,N,x_j) \\
    =& \sum_{j:x^r_j\geq x_l}h^r_j\times x_{b_j}\times\bpdf(i,N,x_{b_j}) \\
    \geq &\frac{1}{q^{i+1}}\sum_{j:x^r_j\geq x_l} h^{r}_j\times x^{r}_j\times \bpdf(i,N,x^{r}_j) \\
    \geq &\begin{cases}
\frac{1}{q^{i+1}}\left(\frac{(i+1)F_{i+1}}{N-i}-\epsilon_2 - \frac{i+1}{N-i} - p\right) & i=0\\
\frac{1}{q^{i+1}}\left(\frac{(i+1)F_{i+1}}{N-i}-\epsilon_2 - \frac{i+1}{N-i} - p\times \bpdf(i,N,x^r_k)\right) & i>0
\end{cases} \\
\geq &\begin{cases}
\frac{1}{q^{i+1}}\left(\frac{(i+1)F_{i+1}}{N-i}-\epsilon_2 - \frac{i+1}{N-i} - p\right) & i=0\\
\frac{1}{q^{i+1}}\left(\frac{(i+1)F_{i+1}}{N-i}-\epsilon_2 - \frac{i+1}{N-i} - p\times \bpdf(i,N,qx_l)\right) & i>0
\end{cases}
\end{align*}
\endgroup

The left side of the inequality in this lemma is proved. On the other side, we consider two cases.
Define a function $f(x) = x^{i+1}(1-x)^{N-i}$ for any $i\geq 0$. Since
\begin{align*}
    f'(x) &= (i+1)x^i(1-x)^{N-i} - (N-i)x^{i+1}(1-x)^{N-i-1} \\
    &= x^i(1-x)^{N-i-1}((i+1)(1-x) - (N-i)x) \\
    &= x^i(1-x)^{N-i-1}(-(N+1)x + (i+1))
\end{align*}
$f(x)$ monotonically increases when $0\leq x\leq \frac{i+1}{N+1}$, and monotonically decreases when $\frac{i+1}{N+1} \leq x \leq 1$.
Therefore, for $x^r_{j}\leq \frac{i+1}{N+1}$, we have:
\begin{align*}
    \frac{h^{r}_j\times x^{r}_j\times \bpdf(i,N,x^{r}_j)}{h^r_j\times x_{b_j}\times\bpdf(i,N,x_{b_j})} = \frac{ (x^{r}_j)^{i+1}(1-x^r_j)^{N-i}}{x_{b_j}^{i+1}(1-x_{b_j})^{N-i}} \geq 1;
\end{align*}
for $x^r_{j} > \frac{i+1}{N+1}$. we have:
\begin{align*}
    \frac{h^{r}_j\times x^{r}_j\times \bpdf(i,N,x^{r}_j)}{h^r_j\times x_{b_j}\times\bpdf(i,N,x_{b_j})} &= \frac{ (x^{r}_j)^{i+1}(1-x^r_j)^{N-i}}{x_{b_j}^{i+1}(1-x_{b_j})^{N-i}} \\ 
    &\geq q^{i+1} \left(\frac{1-qx_{b_j}}{1-x_{b_j}}\right)^{N-i}.
\end{align*}

Picking a parameter $\hat{x}_{\epsilon_{3,i}}\geq \frac{i+1}{N+1}$ we can find the corresponding $\epsilon_{3,i} = \frac{1}{q^{i+1}}\left(\frac{1-\hat{x}_{\epsilon_{3,i}}}{1-q\hat{x}_{\epsilon_{3,i}}}\right)^{N-i}-1 \in(0,1)$ such that for any $x^r_j\leq qx_{b_j}\leq q\hat{x}_{\epsilon_{3,i}}$ we have $\frac{h^{r}_j\times x^{r}_j\times \bpdf(i,N,x^{r}_j)}{h^r_j\times x_{b_j}\times\bpdf(i,N,x_{b_j})} \geq q^{i+1} \left(\frac{1-qx_{b_j}}{1-x_{b_j}}\right)^{N-i} = \frac{1}{1+\epsilon_{3,i}}$.Therefore, we can get the right side of the inequality as below:
\begingroup\makeatletter\def\f@size{9.5}\check@mathfonts
\begin{align*}
    &\sum_{j=1}^l h_j\times x_j\times\bpdf(i,N,x_j) \\
    =& \sum_{j:q\hat{x}_{\epsilon_{3,i}}\geq x^r_j\geq x_l}h^r_j\times x_{b_j}\times\bpdf(i,N,x_{b_j}) \\
    &~~~+ \sum_{j: x^r_j > q\hat{x}_{\epsilon_{3,i}}}h^r_j\times x_{b_j}\times\bpdf(i,N,x_{b_j})\\
    \leq &(1+\epsilon_{3,i})\sum_{j:x^r_j\geq x_l} h^{r}_j\times x^{r}_j\times \bpdf(i,N,x^{r}_j) + \bpdf(i,N,\hat{x}_{\epsilon_{3,i}}) \\
    \leq &\begin{cases}
    (1+\epsilon_{3,i})\left(\frac{(i+1)F_{i+1}}{N-i}+\epsilon_{2,i} - p\times \bpdf(i,N,x^r_k)\right) & \\
    ~~~~~~~~~~~~ + \bpdf(i,N,\hat{x}_{\epsilon_{3,i}}) & i=0\\
    (1+\epsilon_{3,i})\left(\frac{(i+1)F_{i+1}}{N-i}+\epsilon_{2,i}\right) + \bpdf(i,N,\hat{x}_{\epsilon_{3,i}}) & i>0
    \end{cases} \\
    \leq &\begin{cases}
    (1+\epsilon_{3,i})\left(\frac{(i+1)F_{i+1}}{N-i}+\epsilon_{2,i} - p\times \bpdf(i,N,qx_l)\right) & \\
    ~~~~~~~~~~~~  + \bpdf(i,N,\hat{x}_{\epsilon_{3,i}}) & i=0\\
    (1+\epsilon_{3,i})\left(\frac{(i+1)F_{i+1}}{N-i}+\epsilon_{2,i}\right) + \bpdf(i,N,\hat{x}_{\epsilon_{3,i}}) & i>0
    \end{cases}
\end{align*}
\endgroup

The inequality holds with probability at least $1 - \exp\left(\frac{-2(N-i)^2\epsilon_2^2}{N(i+1)^2}\right)$.
\end{proof}

Therefore, the lower bound in Theorem~\ref{thm:bound3final-lower} is proved, i.e. $\lambda_G\geq \min_{1\leq idx \leq l+1}\mathtt{LPlower}(X_l,F^S,idx,G,i',\mathbf{\epsilon_{2}}, \mathbf{\epsilon_{3}}, \mathbf{\hat{x}_{\epsilon_{3}}})$ with probability at least $1-2\times \sum_{0\leq i\leq i'}\exp\left(\frac{-2(N-i)^2\epsilon_{2,i}^2}{N(i+1)^2}\right)$, where the randomness is taken over the sample set $S\leftarrow \mathcal{P}^N$ of size $N$. 

\subsection{Upper Bound}\label{sec:lpub-app}
Let $x^r_k = \min\limits_{i\geq 1}\{x^r_i:x^r_i\geq x_l/q\}$ be the smallest probability in $P$ no less than $x_l/q$. Note that for upper bounds the definition of $x^r_k$ is slightly different from the $X^r_k$ for lower bounds in Section~\ref{sec:lplb-app} above. We will use $X_l$ to estimate $X^r_k=\{x^r_1,...,x^r_k\}$, and let $p=\sum_{i>k}x^r_i\times h^r_i$ be the remaining probability. Specifically, for any $1\leq i\leq k$, we map $x^r_i\in X^r_k$ to its closest value $x_{a_i}$ where $a_i=\arg\min\limits_{1\leq j\leq l}\{x_{j}:x_{j} \geq x^r_i\}$ in $X_l$ that are equal to or greater than $x^r_i$. Then we have $\frac{x_{a_i}}{q} \leq x^r_i\leq x_{a_i}$. Therefore, for any $1\leq j\leq l$, the histogram $h_j$ of $x_j$ is the sum of all values in $H^r$ with the corresponding mesh values being mapped to $x_j$, i.e., $h_j = \sum_{i:a_i=j}h^r_i$. Recall that for any $G\geq 0$, $G = \sum_{i=1}^{i(G,H^r)-1}h^r_i + c(G,H^r)$. We can use $H$ instead of $H^r$ to represent $G$ by defining $i(G,H) = a_{i(G,H^r)}$ and $c(G,H) = \sum_{j:a_j=i(G,H)}h^r_j + c(G,H^r)$ such that $G= c(G,H) + \sum_{i=1}^{i(G,H)-1}h_i$. Then we have 

\begin{align*}
    \lambda_G &= c(G,H^r)x_{i(G,H^r)} + \sum_{i=1}^{i(G,H^r)-1}x_ih^r_i \\
    &\leq c(G,H^r)x_{a_i(G,H^r)} + \sum_{i=1}^{i(G,H^r)-1}x_{a_i}h^r_i \\
    &= c(G,H)x_{i(G,H)} + \sum_{i=1}^{i(G,H)-1}x_ih_i
\end{align*}

Therefore, given a set of fine-grained mesh values $X_l$ and a guessing number $G$, as long as we can find a upper bound of $c(G,H)x_{i(G,H)} + \sum_{i=1}^{i(G,H)-1}x_ih_i$ using linear programming, it will also be a upper bound of $\lambda_G$.

The linear programming task for bounding $c(G,H)x_{i(G,H)} + \sum_{i=1}^{i(G,H)-1}x_ih_i$ in this section is similar to $\mathtt{LP1a}$ we described in Section~\ref{sec:midLP}, except that Constraint (2) and (3) need to be modified considering the difference between our fine-grained mesh $X_l$ and the real probabilities $X^r_k$. For Constraint (3), previously we have $\sum_{j=1}^k h^r_j\times x^r_j = 1-p$. After mapping $x^r_j$ to $x_{a_j}$ with $\frac{x_{a_i}}{q} \leq x^r_i\leq x_{a_i}$ for each $1\leq j\leq k$, we have:
\begin{align*}
    &\sum_{j=1}^l h_j\times x_j = \sum_{j=1}^k h^r_j\times x_{a_j} \geq \sum_{j=1}^k h^r_j\times x^r_{j} = 1-p \\
    &\sum_{j=1}^l h_j\times x_j = \sum_{j=1}^k h^r_j\times x_{a_j} \leq \sum_{j=1}^k h^r_j\times (q\times x^r_{j}) = q(1-p)
\end{align*}

Constraint (2) that bounds $\sum_{j=1}^l h_j\times x_j\times \bpdf(i,N,x_j)$ also need to be changed considering the difference between $\sum_{j=1}^l h_j\times x_j\times \bpdf(i,N,x_j)$ and $\sum_{j=1}^k h^r_j\times x^r_j\times \bpdf(i,N,x^r_j)$. We will present the change and the proof in Lemma~\ref{lem:goodturing_LPupper}later in this section.

We use the final linear program $\mathtt{LPupper}$ (see Appendix~\ref{app:LPupper}) to upper bound $\lambda_G$. This linear program $\mathtt{LPupper}$ can be considered as an extension from $\mathtt{LP1}$ and $\mathtt{LP1a}$ removing the two ideal settings. 

$\mathtt{LPupper}$ takes $h_1,...,h_l,c,p$ as variables, and maximize $\sum_{j<idx}(h_j\times x_j +c\times x_{idx})$ with the constraint $\sum_{j<idx}h_j+c = G$. Given an arbitrary guessing number $G$, fixing other input parameters, we can upper bound $\lambda_G$ as  $\max\limits_{1\leq idx\leq l+1}\mathtt{LPupper}(X_l,F^S,idx,G,i',\mathbf{\epsilon_2},\mathbf{\epsilon_3},\mathbf{\hat{x}_{\epsilon_3}})$ by running $\mathtt{LPupper}(X_l,F^S,idx,G,i',\mathbf{\epsilon_2},\mathbf{\epsilon_3},\mathbf{\hat{x}_{\epsilon_3}})$ for $l+1$ times with $idx\in\{1,2,...,l+1\}$.

The following lemma proves Constraint (2) in $\mathtt{LPupper}$ holds with high probability:

\begin{lemma}\label{lem:goodturing_LPupper}
Given $F^S=\{F_1,...,F_N\}$ from a sample set $S\leftarrow \mathcal{P}^N$, fix any $q>1$, $i\geq 0$, $0\leq \epsilon_{2,i} \leq 1$, $\frac{i+1}{N+1}\leq \hat{x}_{\epsilon_{3,i}}\leq 1$, $\epsilon_{3,i}=\frac{1}{q^{i+1}}\left(\frac{1-\hat{x}_{\epsilon_{3,i}}}{1-q\hat{x}_{\epsilon_{3,i}}}\right)^{N-i}-1 \in(0,1)$, and let $X_l=X_l^q$. If $i=0$ let $W_0$ be an indicator random variable which is $1$ if and only if $\frac{1}{1+\epsilon_{3,i}} (\frac{(i+1)F_{i+1}}{N-i} - \epsilon_{2,i} - \frac{i+1}{N-i} - p - \bpdf(i,N,q\hat{x}_{\epsilon_{3,i}})) \leq \sum_{j=1}^l h_j x_j \bpdf(i,N,x_j)$ and $\sum_{j=1}^l h_j x_j \bpdf(i,N,x_j) \leq q^{i+1} (\frac{(i+1)F_{i+1}}{N-i} + \epsilon_{2,i} - p\times \bpdf(i,N,x_l))$. Similarly, if $i>0$ let $W_i$ be an indicator random variable which is $1$ if and only if $\frac{1}{1+\epsilon_{3,i}} (\frac{(i+1)F_{i+1}}{N-i}-\epsilon_{2,i} - \frac{i+1}{N-i} - p \times \bpdf(i,N,x_l) - \bpdf(i,N,q\hat{x}_{\epsilon_{3,i}})) \leq \sum_{j=1}^l h_j x_j \bpdf(i,N,x_j)$ and $\sum_{j=1}^l h_j x_j \bpdf(i,N,x_j) \leq q^{i+1} (\frac{(i+1)F_{i+1}}{N-i} + \epsilon_{2,i})$.

Then the constraints hold with probability \[ \Pr[W_i = 1] \geq 1 - 2\times\exp\left(\frac{-2(N-i)^2\epsilon_{2,i}^2}{N(i+1)^2}\right) \] where the randomness is taken over the selection of $S \leftarrow \mathcal{P}^N$.

%Given $F^S=\{F_1,...,F_N\}$ from a sample set $S\leftarrow \mathcal{P}^N$, for any $q>1$, $i\geq 0$, $0\leq \epsilon_{2,i} \leq 1$, $\frac{i+1}{N+1}\leq \hat{x}_{\epsilon_{3,i}}\leq 1$, and $\epsilon_{3,i}=\frac{1}{q^{i+1}}\left(\frac{1-\hat{x}_{\epsilon_{3,i}}}{1-q\hat{x}_{\epsilon_{3,i}}}\right)^{N-i}-1 \in(0,1)$, let $X_l=X_l^q$, then we have:

%\begin{align*}
%    &if~i=0, \\
%    &\frac{1}{1+\epsilon_{3,i}}(\frac{(i+1)F_{i+1}}{N-i}-\epsilon_{2,i} - \frac{i+1}{N-i} - p - \bpdf(i,N,q\hat{x}_{\epsilon_{3,i}})) \\ 
%    &\leq \sum_{j=1}^l h_j\times x_j\times \bpdf(i,N,x_j) \leq q^{i+1}(\frac{(i+1)F_{i+1}}{N-i}+\epsilon_{2,i} - p\times \bpdf(i,N,x_l))
%\end{align*}

%\begin{align*}
%    &if~i\geq 1, \\
%    &\frac{1}{1+\epsilon_{3,i}}(\frac{(i+1)F_{i+1}}{N-i}-\epsilon_{2,i} - \frac{i+1}{N-i} \\
%    &~~~~~~~~ - p\times \bpdf(i,N,x_l) - \bpdf(i,N,q\hat{x}_{\epsilon_{3,i}})) \\
%    &\leq \sum_{j=1}^l h_j\times x_j\times \bpdf(i,N,x_j) \leq q^{i+1}(\frac{(i+1)F_{i+1}}{N-i}+\epsilon_{2,i})
%\end{align*}

%where for any $i$ the two inequalities hold with probability at least $1 - 2\times\exp\left(\frac{-2(N-i)^2\epsilon_{2,i}^2}{N(i+1)^2}\right)$.
\end{lemma}

\begin{proof}
Recall that for any $x_l<\frac{1}{N}$, any individual $i\geq 0$, and any $0\leq \epsilon_{2,i}\leq 1$, the following constraint holds with probability at least $1 - 2\times\exp\left(\frac{-2(N-i)^2\epsilon_{2,i}^2}{N(i+1)^2}\right)$, Lemma~\ref{lem:goodturing_p} proves the following constraint holds with probability at least $1 - 2\times\exp\left(\frac{-2(N-i)^2\epsilon_{2,i}^2}{N(i+1)^2}\right)$:
\begin{align*}
&\text{if }i=0, \\
&\frac{(i+1)F_{i+1}}{N-i}-\epsilon_{2,i} - \frac{i+1}{N-i} - p \leq \sum_{j=1}^k h^r_j x^r_j \bpdf(i,N,x^r_j) \\
&\leq \frac{(i+1)F_{i+1}}{N-i}+\epsilon_{2,i} - p\times \bpdf(i,N,x^r_k); \\
&\text{if }i\geq 1, \\
&\frac{(i+1)F_{i+1}}{N-i}-\epsilon_{2,i} - \frac{i+1}{N-i} - p\times \bpdf(i,N,x_l) \\
&\leq \sum_{j=1}^l h_j x_j \bpdf(i,N,x_j) \leq \frac{(i+1)F_{i+1}}{N-i}+\epsilon_{2,i}.
\end{align*}

Note that $\sum_{j=1}^l h_j\times x_j\times\bpdf(i,N,x_j) = \sum_{j=1}^k h^r_j\times x_{a_j}\times\bpdf(i,N,x_{a_j})$. To prove that the difference between $\sum_{j=1}^k  h^{r}_j\times x^{r}_j\times \bpdf(i,N,x^{r}_j)$ and $\sum_{j=1}^l h_j\times x_j\times\bpdf(i,N,x_j)$ is small, we look into $\frac{h^{r}_j\times x^{r}_j\times \bpdf(i,N,x^{r}_j)}{h^r_j\times x_{a_j}\times\bpdf(i,N,x_{a_j})}$ for each $1\leq j\leq k$.

One the one side, since $\frac{x_{a_j}}{q}\leq x^r_j \leq \times x_{a_j}$, we have $\frac{h^{r}_j\times x^{r}_j\times \bpdf(i,N,x^{r}_j)}{h^r_j\times x_{a_j}\times\bpdf(i,N,x_{a_j})} = \frac{ (x^{r}_j)^{i+1}(1-x^r_j)^{N-i}}{x_{a_j}^{i+1}(1-x_{a_j})^{N-i}}\geq \frac{1}{q^{i+1}}$. Therefore, with probability at least $1 - \exp\left(\frac{-2(N-i)^2\epsilon_{2,i}^2}{N(i+1)^2}\right)$, we have 
\begin{align*}
    &\sum_{j=1}^l h_j\times x_j\times\bpdf(i,N,x_j) \\
    =& \sum_{j:x^r_j\geq \frac{x_l}{q}}h^r_j\times x_{a_j}\times\bpdf(i,N,x_{a_j}) \\
    \leq &q^{i+1}\sum_{j:x^r_j\geq \frac{x_l}{q}} h^{r}_j\times x^{r}_j\times \bpdf(i,N,x^{r}_j) \\
    \leq &\begin{cases}
    q^{i+1}\left(\frac{(i+1)F_{i+1}}{N-i}+\epsilon_2 - p\times \bpdf(i,N,x^r_k)\right) & i=0\\
    q^{i+1}\left(\frac{(i+1)F_{i+1}}{N-i}+\epsilon_2\right) & i>0
    \end{cases} \\
    \leq &\begin{cases}
    q^{i+1}\left(\frac{(i+1)F_{i+1}}{N-i}+\epsilon_2 - p\times \bpdf(i,N,x_l)\right) & i=0\\
    q^{i+1}\left(\frac{(i+1)F_{i+1}}{N-i}+\epsilon_2\right) & i>0
    \end{cases}
\end{align*}

The right side of the inequality in this lemma is proved. On the other side, we consider two cases.
Recall that for any $i\geq 0$ $f(x) = x^{i+1}(1-x)^{N-i}$ monotonically increases when $0\leq x\leq \frac{i+1}{N+1}$, and monotonically decreases when $\frac{i+1}{N+1} \leq x \leq 1$.
Therefore, for $x^r_{j}\leq \frac{i+1}{N+1}$, we have:
\begin{align*}
    \frac{h^{r}_j\times x^{r}_j\times \bpdf(i,N,x^{r}_j)}{h^r_j\times x_{a_j}\times\bpdf(i,N,x_{a_j})} = \frac{ (x^{r}_j)^{i+1}(1-x^r_j)^{N-i}}{x_{a_j}^{i+1}(1-x_{a_j})^{N-i}} \leq 1;
\end{align*}
for $x^r_{j} > \frac{i+1}{N+1}$. we have:
\begin{align*}
    \frac{h^{r}_j\times x^{r}_j\times \bpdf(i,N,x^{r}_j)}{h^r_j\times x_{a_j}\times\bpdf(i,N,x_{a_j})} &= \frac{ (x^{r}_j)^{i+1}(1-x^r_j)^{N-i}}{x_{a_j}^{i+1}(1-x_{a_j})^{N-i}} \\
    & \leq \frac{1}{q^{i+1}}\left(\frac{1-\frac{x_{a_j}}{q}}{1-x_{a_j}}\right)^{N-i}.
\end{align*}

Picking a parameter $\hat{x}_{\epsilon_{3,i}}\geq \frac{i+1}{N+1}$ we can find the corresponding $\epsilon_{3,i} = \frac{1}{q^{i+1}}\left(\frac{1-\hat{x}_{\epsilon_{3,i}}}{1-q\hat{x}_{\epsilon_{3,i}}}\right)^{N-i}-1 \in(0,1)$ such that for any $x_{a_j}\leq q\hat{x}_{\epsilon_{3,i}}$ we have $\frac{h^{r}_j\times x^{r}_j\times \bpdf(i,N,x^{r}_j)}{h^r_j\times x_{a_j}\times\bpdf(i,N,x_{a_j})} \leq \frac{1}{q^{i+1}} \left(\frac{1-\frac{x_{a_j}}{q}}{1-x_{a_j}}\right)^{N-i} = (1+\epsilon_{3,i})$.Therefore, we can get the right side of the inequality as below:
\begin{align*}
    &\sum_{j=1}^l h_j\times x_j\times\bpdf(i,N,x_j) \\
    =& \sum_{j:x^r_j\geq \frac{x_l}{q}}h^r_j\times x_{a_j}\times\bpdf(i,N,x_{a_j}) \\
    \geq & \sum_{j:q\hat{x}_{\epsilon_{3,i}} \geq x^r_j\geq \frac{x_l}{q}}h^r_j\times x_{a_j}\times\bpdf(i,N,x_{a_j}) \\
    \geq & \frac{1}{1+\epsilon_{3,i}}\left(\sum_{j:q\hat{x}_{\epsilon_{3,i}} \geq x^r_j\geq \frac{x_l}{q}} h^{r}_j\times x^{r}_j\times \bpdf(i,N,x^{r}_j)\right) \\
    \geq &\frac{1}{1+\epsilon_{3,i}}\left(\sum_{j:x^r_j\geq \frac{x_l}{q}} h^{r}_j x^{r}_j \bpdf(i,N,x^{r}_j) - \bpdf(i,N,q\hat{x}_{\epsilon_{3,i}}) \right)  \\
    \geq &\begin{cases}
    \frac{1}{1+\epsilon_{3,i}}(\frac{(i+1)F_{i+1}}{N-i}-\epsilon_{2,i} - \frac{i+1}{N-i} - p & \\
    ~~~~ - \bpdf(i,N,q\hat{x}_{\epsilon_{3,i}})) & i=0\\
    \frac{1}{1+\epsilon_{3,i}}(\frac{(i+1)F_{i+1}}{N-i}-\epsilon_{2,i} - \frac{i+1}{N-i} & \\
    ~~~~ - p\times \bpdf(i,N,x^r_k) - \bpdf(i,N,q\hat{x}_{\epsilon_{3,i}})) & i>0
    \end{cases} \\
    \geq &\begin{cases}
    \frac{1}{1+\epsilon_{3,i}}(\frac{(i+1)F_{i+1}}{N-i}-\epsilon_{2,i} - \frac{i+1}{N-i} - p & \\
    ~~~~ - \bpdf(i,N,q\hat{x}_{\epsilon_{3,i}})) & i=0\\
    \frac{1}{1+\epsilon_{3,i}}(\frac{(i+1)F_{i+1}}{N-i}-\epsilon_{2,i} - \frac{i+1}{N-i} & \\
    ~~~~ - p\times \bpdf(i,N,x_l) - \bpdf(i,N,q\hat{x}_{\epsilon_{3,i}})) & i>0
    \end{cases}
\end{align*}

The inequality holds with probability at least $1 - \exp\left(\frac{-2(N-i)^2\epsilon_2^2}{N(i+1)^2}\right)$.
\end{proof}

Therefore, the upper bound in Theorem~\ref{thm:bound3final-lower} is proved, i.e. $\lambda_G\leq \max\limits_{1\leq idx \leq l+1}\mathtt{LPupper}(X_l,F^S,idx,G,i',\mathbf{\epsilon_{2}}, \mathbf{\epsilon_{3}}, \mathbf{\hat{x}_{\epsilon_{3}}})$ with probability at least $1-2\times \sum_{0\leq i\leq i'}\exp\left(\frac{-2(N-i)^2\epsilon_{2,i}^2}{N(i+1)^2}\right)$, where the randomness is taken over the sample set $S$ of size $N$.

%\fullversion{\section{The Linear Program \texttt{LPupper}}\label{app:LPupper}}{\subsection{The Linear Program \texttt{LPupper}}\label{app:LPupper}}

}
\end{document}